\documentclass[aoas]{imsart}

\RequirePackage[authoryear]{natbib}
\RequirePackage{graphicx}

\usepackage{subcaption}
\usepackage{caption}
\usepackage{bbm}

\usepackage{array}
\usepackage[]{appendix}
\usepackage{comment}
\newcolumntype{P}[1]{>{\centering\arraybackslash}p{#1}}

\usepackage{amsmath}
\usepackage{amssymb}
\usepackage{amsthm}
\usepackage{amsfonts}
\usepackage{tikz}
\usetikzlibrary {arrows.meta}
\usetikzlibrary {patterns,patterns.meta}

\usepackage{enumerate}
\usepackage{subcaption}
 \newcommand{\ind}{\perp\!\!\!\!\perp} 
\usepackage[colorlinks,citecolor=blue,urlcolor=blue]{hyperref}
\usepackage{tikz-network}
\usepackage{float}
\usepackage{tikz-cd}

\usepackage[labelfont=bf]{caption}

\allowdisplaybreaks
\usepackage[ruled,linesnumbered]{algorithm2e}
\usepackage{algorithmic}

\usepackage{array}

\newcolumntype{P}[1]{>{\centering\arraybackslash}p{#1}}

\startlocaldefs

\theoremstyle{plain}

\newtheorem{thm}{Theorem}

\newtheorem{prop}{Proposition}

\RestyleAlgo{boxruled}

\endlocaldefs

\begin{document}

\begin{frontmatter}

\title{Identifying Peer Influence in Therapeutic Communities Adjusting for Latent Homophily}

\runtitle{Identifying Peer Influence}

\begin{aug}

\author[A]{\fnms{Shanjukta}~\snm{Nath}\ead[label=e1]{shanjukta.nath@uga.edu}},
\author[B]{\fnms{Keith}~\snm{Warren}\ead[label=e2]{warren.193@osu.edu}}
\and
\author[C]{\fnms{Subhadeep}~\snm{Paul}\ead[label=e3]{paul.963@osu.edu}}
%%%%%%%%%%%%%%%%%%%%%%%%%%%%%%%%%%%%%%%%%%%%%%
%% Addresses                                %%
%%%%%%%%%%%%%%%%%%%%%%%%%%%%%%%%%%%%%%%%%%%%%%
\address[A]{Department of Agricultural and Applied Economics, University of Georgia, \printead[presep={,\ }]{e1}}

\address[B]{College of Social Work, The Ohio State University \printead[presep={,\ }]{e2}}

\address[C]{Department of Statistics, The Ohio State University \printead[presep={,\ }]{e3}}

\end{aug}

\begin{abstract}
We investigate peer role model influence on successful graduation from Therapeutic Communities (TCs) for substance abuse and criminal behavior. We use data from 3 TCs that kept records of exchanges of affirmations among residents and their precise entry and exit dates, allowing us to form peer networks and define a causal effect of interest. The role model effect measures the difference in the expected outcome of a resident (ego) who can observe one of their peers graduate before the ego's exit vs not graduating. To identify peer influence in the presence of unobserved homophily in observational data, we model the network with a latent variable model. We show that our peer influence estimator is asymptotically unbiased when the unobserved latent positions are estimated from the observed network. We additionally propose a measurement error bias correction method to further reduce bias due to estimating latent positions. Our simulations show the proposed latent homophily adjustment and bias correction perform well in finite samples. We also extend the methodology to the case of binary response with a probit model. Our results indicate a positive effect of peers' graduation on residents' graduation and that it differs based on gender, race, and the definition of the role model effect. A counterfactual exercise quantifies the potential benefits of an intervention directly on the treated resident and indirectly on their peers through network propagation.
\end{abstract}


\end{frontmatter}

\section{Introduction}

In application problems from social sciences, social work, economics, and public health, a common form of data collected is a network along with node-level responses and attributes. Such node-level attributes might include responses to questions in a survey, behavioral outcomes, economic variables, and health outcomes, among others.  A fundamental scientific problem associated with such network-linked data is \textit{identifying the peer influence} network-connected neighbors exert on each other's outcomes. 

Therapeutic communities (TCs) for substance abuse and criminal behavior are mutual aid-based programs where residents are kept for a fixed amount of time and are expected to graduate successfully at the end of the program. An important question in TCs is whether the propensity of a resident to graduate successfully is causally impacted by the peer influence of successful graduates or role models among their social contacts. This research aims to quantify the causal role model effect in TCs and develop methods to estimate it using individual-level data from 3 TCs in a midwestern city in the United States.

There is more than half a century of research on understanding how the network and the attributes affect each other \citep{lazer2010coevolution,manski1993identification,lazer2001co,shalizi2011homophily,christakis2007spread}. 
For example, researchers have found evidence of social or peer influence on employee productivity, wages, entrepreneurship \citep{chan2014compensation,herbst2015peer,baird2023optimal,zimmerman2019elite}, school and college achievement \citep{sacerdote2011peer,denning2023class}, criminal behavior \citep{stevenson2017breaking,bayer2009building,gaviria2001school,deming2011better}, patterns of exercising \citep{aral2017exercise}, physical and mental health outcomes \citep{kiessling2023long,christakis2007spread,fowler2008dynamic}. On the other hand, biological and social networks have been demonstrated to display homophily or social selection, whereby individuals who are similar in characteristics tend to be linked in a network \citep{mcpherson2001birds,lazer2010coevolution,dean2017friendship,shalizi2011homophily,vanderweele2011sensitivity}.

Estimating \textit{causal} peer influence has been an active topic of research with advancements in methodologies reported in \cite{manski1993identification,shalizi2011homophily,vanderweele2011sensitivity,vanderweele2013social,bramoulle2009identification}. It has been shown that peer effects can be identified avoiding the reflection problem \cite{manski1993identification} if the social network is more general than just a collection of connected peer groups \citep{bramoulle2009identification,goldsmith2013social}. However, a possible confounder is an unobserved variable (latent characteristics) that affects both the response and selection of network neighbors, creating omitted variable bias. Peer influence cannot generally be separated from this latent homophily, using observational social network data without additional methodology \citep{shalizi2011homophily}.

Peer influence is a core principle in substance abuse treatment, forming the basis of mutual aid-based programs, including 12-Step programs \citep{white2008twelve}, recovery housing \citep{jason2022dynamic,mericle2023social} and TCs \citep{gossop2000therapeutic}. However, quantitative analyses have nearly always failed to distinguish between peer influence and homophily.  It is known that people who are successful in recovering from substance abuse are likely to have social networks that include others in recovery \citep{best2019we,roxburgh2023composition}, but these studies have not distinguished between peer influence and homophily.  In the specific case of TCs, there is evidence that program graduates cluster together \citep{warren2020building}, but this analysis also did not distinguish between peer influence and homophily.  In a study that analyzed a longitudinal social network of friendship nominations in a TC using a stochastic actor-oriented model \citep{snijders2017stochastic}, it was found that resident program engagement was correlated with that of peers, but that homophily and not peer influence explained the correlation \citep{kreager2019evaluating}.  The question of whether peers influence each other \textit{causally} in mutual aid-based programs for substance abuse, therefore, remains unresolved. 

Among substance abuse programs that emphasize mutual aid, the most highly structured and professionalized are TCs.  %They, therefore, offer unique opportunities for social network analysis aimed at disentangling peer influence from homophily. 
These are residential treatment programs for substance abuse based on mutual aid between recovering peers. They typically have a maximum time period for treatment (180 days), which, along with the residential nature of the program, ensures a reasonably stable turnover of participants.  The clinical process primarily occurs through social learning within the network of peers who are themselves in recovery \citep{gossop2000therapeutic,yates2017integration}.  Peers who exemplify recovery and prosocial behavior and who are most active in helping others are known as role models \citep{gossop2000therapeutic}.  The professional staff's job is to encourage recovery-supportive interactions between residents while also demonstrating prosocial behavior in line with TC values \citep{gossop2000therapeutic}. 

This paper aims to estimate peer influence, separating it from latent homophily in TCs. The data we analyze comes from three 80-bed TC units. One is a unit for women, and the other two are for men. The units are minimum security correctional units for felony offenders. The offenses of the residents include possession of drugs, robbery, burglary, and domestic violence, among others. All units kept records of each resident's entry and exit dates, along with written affirmations and corrections exchanged that form the basis of our social networks. Affirmations are written records of appreciation for good behavior that one resident gives to another. This could include doing a good job on a routine task or helping with others' recovery. Some examples of affirmations from the dataset used in this paper are having prosocial fun on New Year's Eve, helping another resident on the computer, sharing personal information during phase group, and doing a good job in the kitchen with a great attitude. These networks, along with the precise entry and exit dates, aid us in correctly identifying the role models to estimate the causal parameter of interest.

The learnings from this paper can help TC clinicians design network interventions that might increase the graduation rate. For example, one can assign a successful ``buddy" to the ``at-risk" TC residents if network influence meaningfully impacts graduation status. On the contrary, if latent homophily explains the correlation between the graduation status of network-connected neighbors, then such an intervention is less likely to impact graduation positively. Therefore, it is crucial to separate the two effects.  The same principle applies to a variety of interventions that could be applied in TCs and other mutual aid-based systems that aim at recovery from substance abuse.

\begin{figure}[!htbp] 
  \caption{Peer effects in TC}
    \label{introfig} 

 \begin{minipage}[b]{0.5\linewidth}
       \begin{center}

    \includegraphics[width=\linewidth]{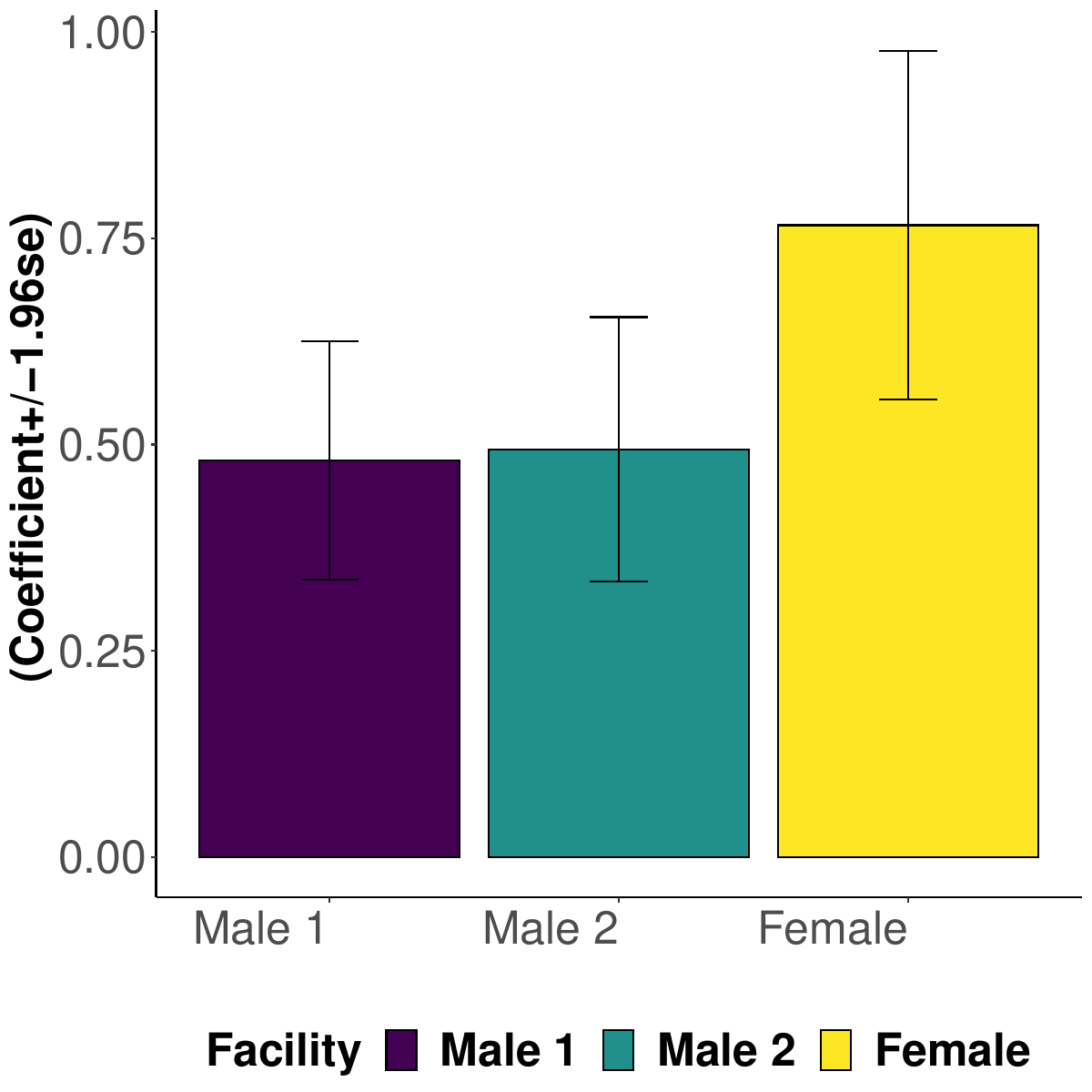}
    \caption*{a. Estimated peer effect} 
        \end{center}

  \end{minipage}%%
  \begin{minipage}[b]{0.5\linewidth}
       \begin{center}

    \includegraphics[width=\linewidth]{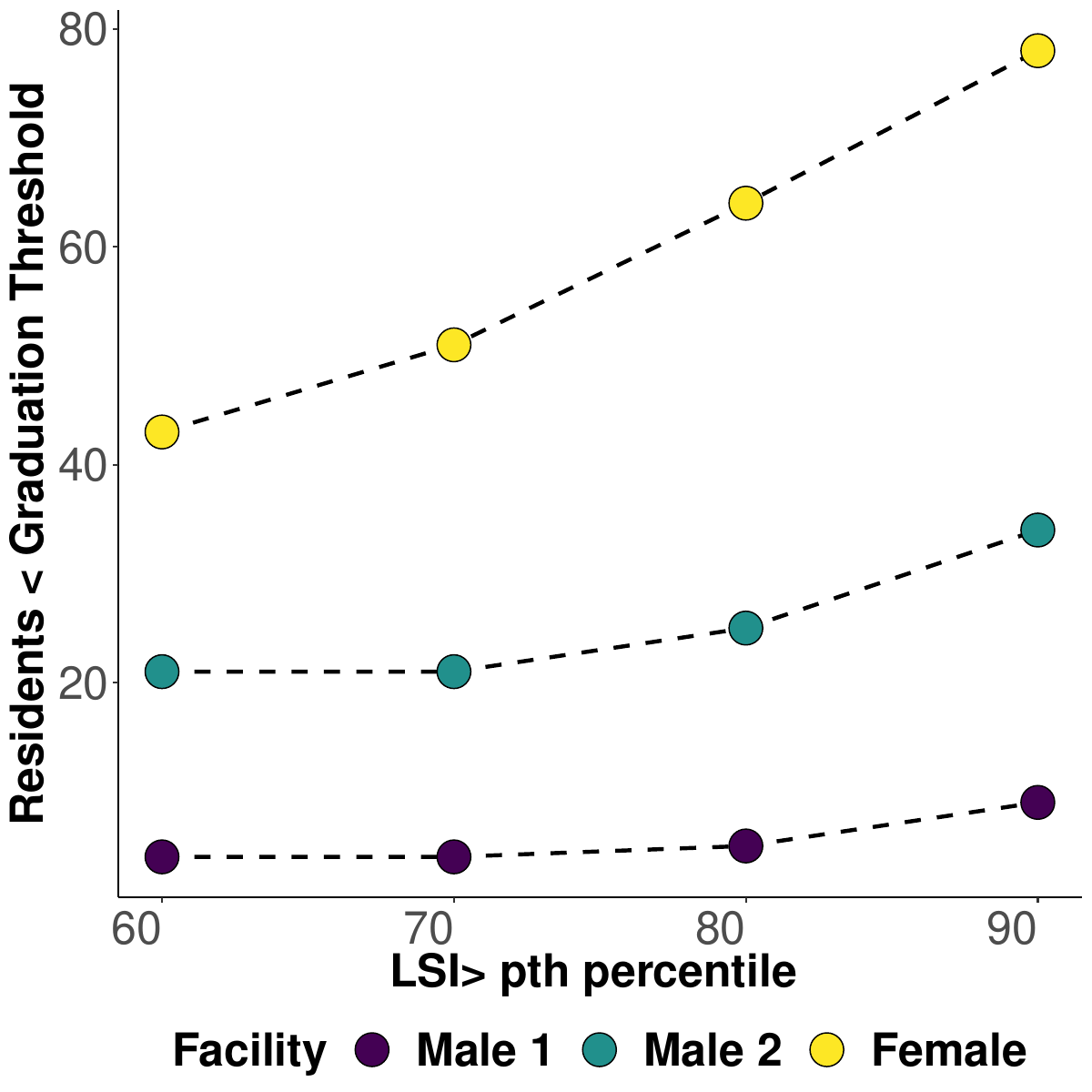}
    \caption*{b. Counterfactual Intervention} 
        \end{center}

  \end{minipage}
  \begin{minipage}{14.0cm}
\footnotesize{
    {Notes: Panel (a) displays the estimated causal peer effect from Table \ref{rhotable} (columns 2, 4, and 6) for the three TCs. Panel (b) shows the remaining residents below the graduation threshold post \textit{counterfactual} interventions for four Level of Service Inventory (LSI) cut-offs. LSI is a covariate collected at the time of entry into the TC, and it is negatively correlated with successful graduation from TC.}}
    \end{minipage}
\end{figure}
Recently, several authors have put forth solutions to the problem of separating peer influence from latent homophily in observational data. The methodology includes using latent communities from a stochastic block model (SBM) \citep{mcfowland2021estimating}, and joint modeling an outcome equation and social network formation model \citep{goldsmith2013social,hsieh2016social}. There is another strand of research on estimating the effects of peer attributes on an agent's outcome with randomized peer groups \cite{li2019randomization,basse2024randomization}, or spillovers in exogenous treatments with randomized treatment assignments \cite{toulis2013estimation}.  

These papers focus on a setup where there is a random group formation, and the authors have developed randomization-based inference frameworks to estimate the spillovers in peer attribute or treatment on the agent's response. %These papers have also compared the performance of the randomization inference framework with model-based procedures. 
These papers are motivated by examples such as random 2 or 4 person dorm room assignments  or the random assignment of CEOs to group meetings. The primary difference between the setup and the empirical application discussed in this paper to the above-stated literature is that our residents do not interact in randomly assigned peer groups. Instead, our network is constructed based on whom the residents interact with on their own. Hence we must consider the possibility that the residents choose their peers based on some unobserved (latent) characteristics. The frameworks of \cite{li2019randomization,basse2024randomization} might be relevant in other settings of mutual aid-based substance use treatment, e.g., recovery houses if we assume that residents only interact with their housemates and not with residents of other houses within the same agency. However, the open interaction structure of Therapeutic Communities means that it is not possible to find such limitations on interaction. Therefore, it is not clear how one can create peer groups in order to define peer effect in terms of the composition of those groups. Second, the randomized inference procedures in \cite{li2019randomization,basse2024randomization} crucially depend on the assumption that the peers are randomly allocated and not selected by the students. This assumption is violated in our settings since it is very likely that the residents choose whom they want to interact with.

We develop a new method to estimate peer influence by adjusting for homophily by combining a latent position random graph model with measurement error models. We model the network with a random dot product graph model (RDPG) and estimate the latent homophily factors from a Spectral Embedding of the Adjacency matrix (ASE).  We then include these estimated factors in our linear peer influence outcome model. A measurement error bias correction procedure is employed while estimating the outcome model. The methodology is then extended to the context of binary responses with a probit regression model.

Our approach improves upon the asymptotic unbiasedness in \cite{mcfowland2021estimating} by proposing a method for reducing bias in finite samples. Further, we model the network with the RDPG model, a latent position model that is more general than the SBM. We provide an expression for the upper bound on the bias as a function of $n$ and network sparsity. Our simulation results show the method is effective in providing an estimate of peer influence parameter with low bias both when the network is generated from RDPG and the SBM models.

In the context of TCs, we first define a causal ``role model influence'' parameter with the help of $do-$ interventions \citep{pearl2009causality}. We then illustrate the need for adjusting for homophily with the help of a directed acyclic graph (DAG) \citep{pearl2009causality,hernan2009invited}. We show that the peer influence parameter in our model is equivalent to this role model influence. 

Our results show significant peer influence in all TC units (Panel (a) in Figure \ref{introfig}). However, there are significant differences in peer effect by gender and race \citep{bostwick2022nevertheless,gong2021gender,gershenson2022long}. We see a substantial reduction in the peer influence coefficient ($\approx$ 19\% drop) once we correct for latent homophily and adjust for the bias in the female correctional unit. On the contrary, we see marginal changes in the two male units. Additionally, the strength of peer influence is at least 30\% higher in the female unit than in the male unit. We perform a series of robustness checks to ensure the validity of our results (probit regression, binarizing network, additive and multiplicative latent effects model, an alternative definition of role model effect, and directed sender and receiver networks). We also explore heterogeneous effects by race. We construct separate peer variables for white and non-white peer affirmations and interact these with the white dummy. We find that peer graduation of both races impacts resident's graduation positively.

Finally, we do a \textit{counterfactual} analysis to estimate the direct and the cascading effect of a do-intervention of assigning a successful buddy to ``at-risk" residents in the unit. At the beginning of the data collection process, a variable called Level of Service Inventory (LSI) is measured, which is negatively correlated with graduation status (see Figure \ref{LSIgrad}). We use LSI as a proxy for targeting the intervention and compute the overall impact on residents. Panel (b) in Figure \ref{introfig} suggests that such an intervention can be useful for policymakers to improve graduation from these units as the intervention propagates positively through the network.

 \section{Methods}
\label{methods}

\subsection{Data and Statistical Problem Setup}

In the TCs, the residents dynamically enter and exit the units over time. We use the notation $i \in \{1,2,..,n\}$ to denote a resident in TC. We observe timestamps for entry and exit for every resident in the TC. The timestamps are critical to our identification. The entry and exit dates for the $i$th resident are denoted by $t_{i}^{\text{entry}}\in \mathcal{T}$ and $t_{i}^{\text{exit}} \in \mathcal{T}$ respectively, where $\mathcal{T} = \{1,...,\tau\}$. We define another variable $T_{i}= t_{i}^{\text{exit}} -t_{i}^{\text{entry}} $, which measures the time spent in TC.

\begin{figure}[!htbp]
    \centering
    \caption{Time spent in Therapeutic Community}
    \label{timefig1}
    \begin{center}
    \begin{minipage}{.50\textwidth}
     \begin{tabular}[b]{@{}cc@{}}

\begin{tikzpicture}[scale=1]
\begin{scope}[>=latex] 
\draw [line width=0.7pt][<->](-0.5,0)--(9,0); 
\end{scope}
\draw [pattern={Lines[angle=45,distance={3pt/sqrt(2)}]},pattern color=blue,line width=0.6pt](3,0)--(3,0.5)--(6,0.5)--(6,0)--cycle;
\draw [pattern={Lines[angle=135,distance={3pt/sqrt(2)}]},pattern color=red,line width=0.6pt](1,0)--(1,0.5)--(4.5,0.5)--(4.5,0)--cycle;
\pgftext[base,x=1cm,y=-0.4cm] {\small $t_{i}^{\text{entry}}$};
\pgftext[base,x=3 cm,y=-0.4cm] {\small $t_{j}^{\text{entry}}$};
\pgftext[base,x=6 cm,y=-0.35cm] {\small $t_{j}^{\text{exit}}$};
%\pgftext[base,x=4.5cm,y=-0.3cm] {\small $4.5$};
\pgftext[base,x=4.5cm,y=-0.35cm] {\small $t_{i}^{\text{exit}}$};
%\pgftext[base,x=7cm,y=-0.3cm,] {\small $180$}; 
\end{tikzpicture}
\\\\
\textbf{(a) Overlap}
\end{tabular}
\vspace{0.5cm}
\end{minipage}
    \begin{minipage}{.50\textwidth}
     \begin{tabular}[b]{@{}cc@{}}

\begin{tikzpicture}[scale=1]
\begin{scope}[>=latex] 
\draw [line width=0.7pt][<->](-0.5,0)--(9,0); 
\end{scope}
\draw [pattern={Lines[angle=45,distance={3pt/sqrt(2)}]},pattern color=blue,line width=0.6pt](1,0)--(1,0.5)--(4.5,0.5)--(4.5,0)--cycle;
\draw [pattern={Lines[angle=135,distance={3pt/sqrt(2)}]},pattern color=red,line width=0.6pt](3,0)--(3,0.5)--(6,0.5)--(6,0)--cycle; 
\pgftext[base,x=1cm,y=-0.4cm] {\small $t_{j}^{\text{entry}}$};
\pgftext[base,x=3 cm,y=-0.4cm] {\small $t_{i}^{\text{entry}}$};
\pgftext[base,x=6 cm,y=-0.35cm] {\small $t_{i}^{\text{exit}}$};
%\pgftext[base,x=4.5cm,y=-0.3cm] {\small $4.5$};
\pgftext[base,x=4.5cm,y=-0.35cm] {\small $t_{j}^{\text{exit}}$};
%\pgftext[base,x=7cm,y=-0.3cm,] {\small $180$}; 
\end{tikzpicture}
\\\\
\textbf{(b) Overlap}
\end{tabular}
\vspace{0.5cm}
\end{minipage}
   \begin{minipage}{.50\textwidth}
     \begin{tabular}[b]{@{}cc@{}}

\begin{tikzpicture}[scale=1]
\begin{scope}[>=latex] 
\draw [line width=0.7pt][<->](-0.5,0)--(9,0); 
\end{scope}
\draw [pattern={Lines[angle=45,distance={3pt/sqrt(2)}]},pattern color=blue,line width=0.6pt](5,0)--(5,0.5)--(6.3,0.5)--(6.3,0)--cycle;
\draw [pattern={Lines[angle=135,distance={3pt/sqrt(2)}]},pattern color=red,line width=0.6pt](1,0)--(1,0.5)--(4.0,0.5)--(4.0,0)--cycle;
\pgftext[base,x=1cm,y=-0.4cm] {\small $t_{i}^{\text{entry}}$};
\pgftext[base,x=5 cm,y=-0.4cm] {\small $t_{j}^{\text{entry}}$};
\pgftext[base,x=6.3 cm,y=-0.35cm] {\small $t_{j}^{\text{exit}}$};
%\pgftext[base,x=4.5cm,y=-0.3cm] {\small $4.5$};
\pgftext[base,x=4.0cm,y=-0.35cm] {\small $t_{i}^{\text{exit}}$};
%\pgftext[base,x=7cm,y=-0.3cm,] {\small $180$}; 
\end{tikzpicture}
\\\\
\textbf{(b) No Overlap}
\end{tabular}
%\vspace{0.5cm}
\end{minipage}
\end{center}
\end{figure}

The time in TC for the residents is depicted in Figure \ref{timefig1}. The entry and exit dates for residents $i$ and $j$ can vary. In Figure \ref{timefig1} (a) and (b), the time resident $i$ spent in the unit is shaded in red, and the time resident $j$ spent is shaded in blue. The overlap in their time spent in TC is shown by the time between $t_{j}^{\text{entry}}$ and $t_{i}^{\text{exit}}$. Since $i$ and $j$ overlap, they can enter each other's network if they exchange affirmations during their time together in the facility. Figure \ref{timefig1} (a) and (b) are two example scenarios of overlap. The real data can involve more scenarios for overlap between the two residents. However, on the contrary, in Figure \ref{timefig1} (c), there is no overlap between the two residents' time in TC. This would mean that, in this case, residents $i$ and $j$ cannot be connected in the peer networks.

In our empirical setting, the entry dates span a period of three years between 2005 and 2008 (see Figure \ref{time} (a)). A resident must leave the TC within 180 days, either with successful graduation or failure to graduate. Residents may leave early from the program, both with successful and failed graduation. However, most residents stay in the program longer, as shown in the boxplots of Figure \ref{time} (b), with the median time of stay being 149 and 150 in the two male units and 124 in the female unit. In all units, there were a few residents who left very early. We provide more details on the relationship between graduation and time spent in TC in section \ref{dataempiricalsection} along with the reasons for leaving the TCs early.

\begin{figure}[!htbp] 
  \caption{Variation in entry dates and time spent with the peers in TCs}
    \label{time} 

 \begin{minipage}[b]{0.5\linewidth}
       \begin{center}

    \includegraphics[width=\linewidth]{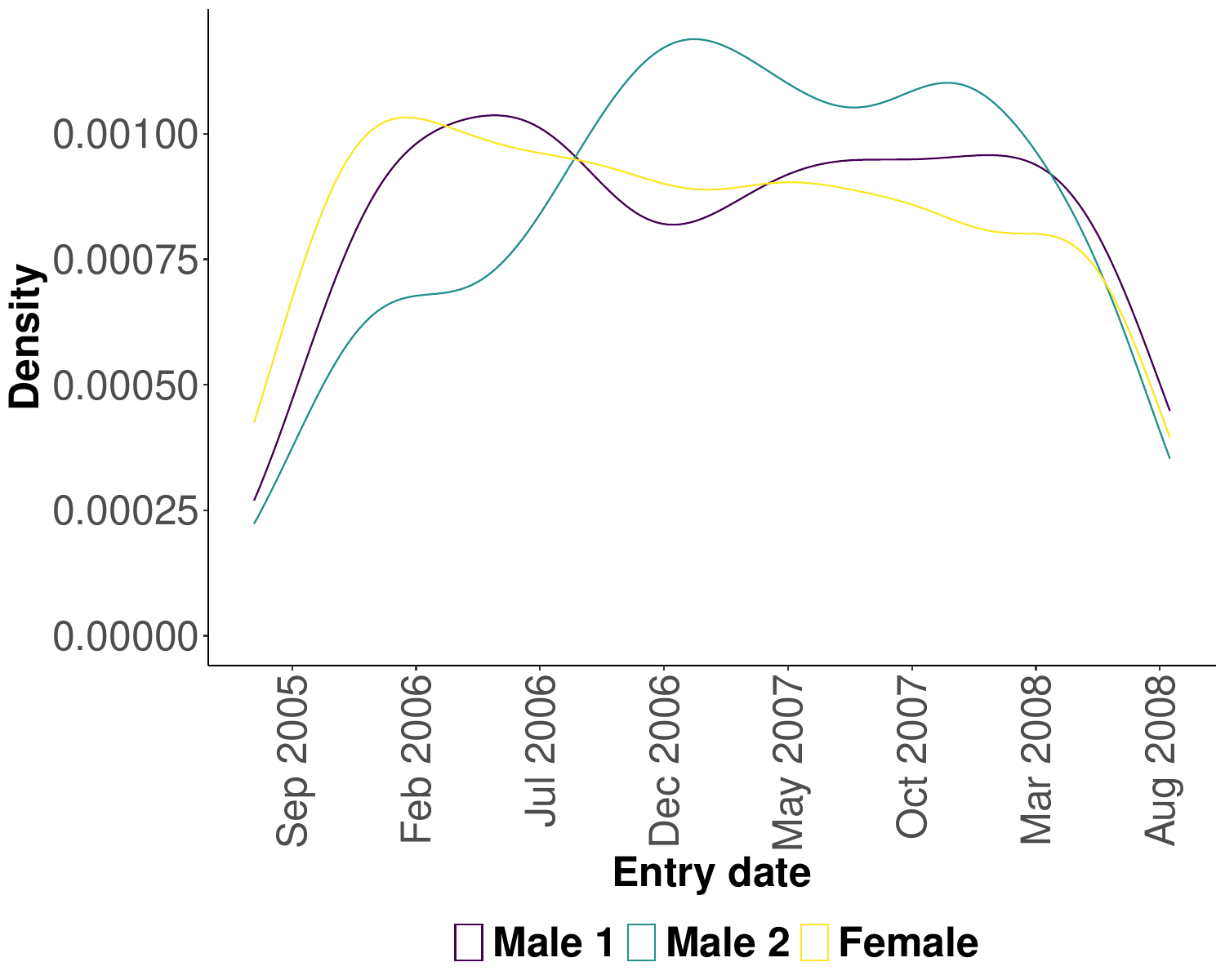}
    \caption*{a. Entry Date} 
        \end{center}

  \end{minipage}%%
  \begin{minipage}[b]{0.5\linewidth}
       \begin{center}

    \includegraphics[width=\linewidth]{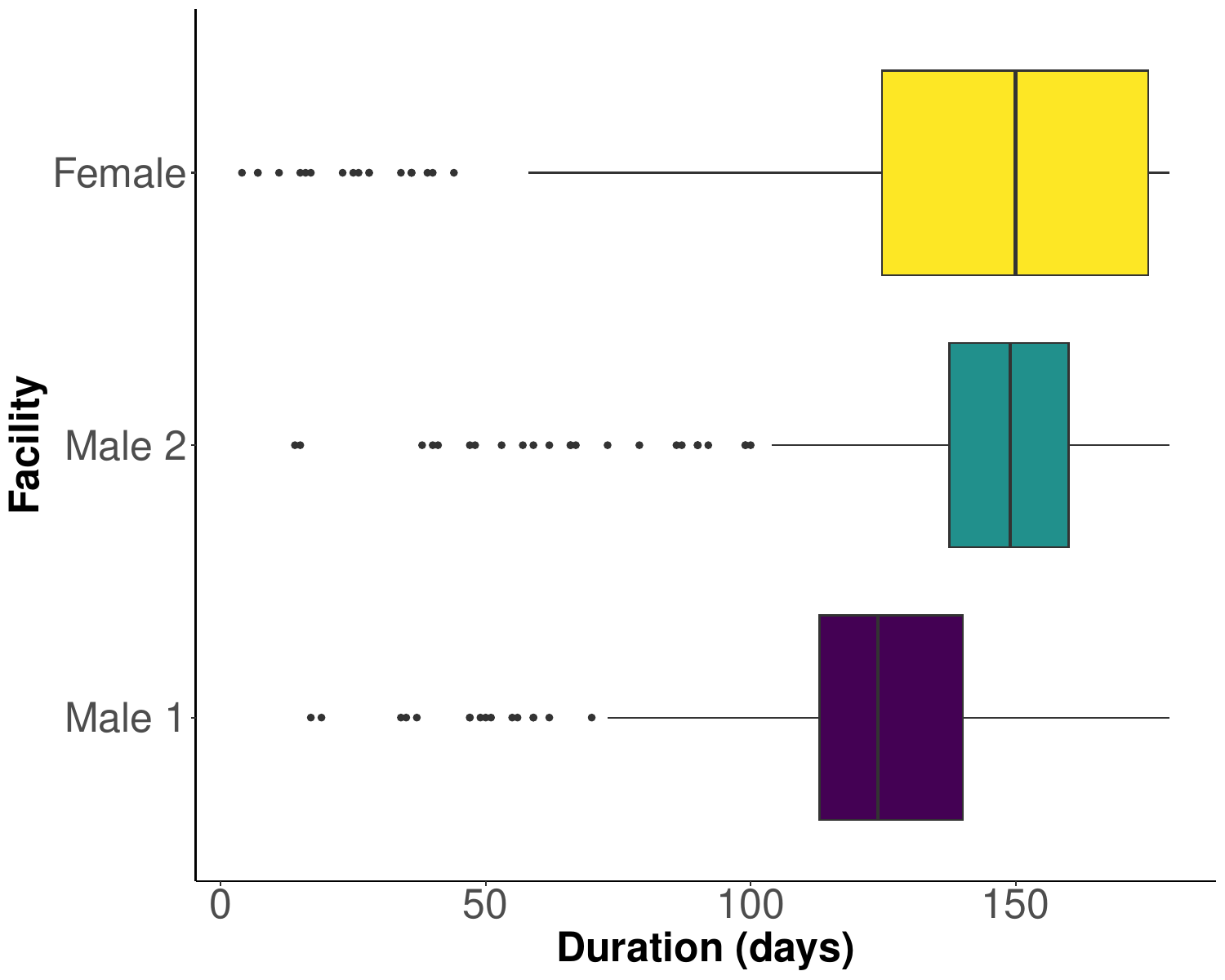}
    \caption*{b. Time in TC} 
        \end{center}

  \end{minipage}
  \begin{minipage}{15.5cm}
\footnotesize{
    {Notes: Panel (a) shows the entry dates of residents in each community across the three-year period. Panel (b) illustrates the distribution of the residents' duration spent in the TCs. }}
    \end{minipage}
\end{figure}
Let the random variable $S_i \in \{0,1\}$ denote the final graduation status of individual $i$ within the study population, with the value 1 denoting successful graduation and 0 a failure to graduate. This study aims to understand the effect of receiving affirmations from peers whom the \textit{resident can observe} to be successful in graduating on their own graduation. We call this the ``role model effect''.

Let $A$ denote the $n \times n $ weighted (assumed undirected for easier exposition) adjacency matrix composed of the counts of affirmations, i.e., $A_{ij}$ records the number of affirmations node $i$ has sent (and/or received) to (from) node $j$. Note that $i$ can send (receive) affirmations to $j$ only during the time they are both part of the TC. Therefore, we use the term that $j$ is a ``peer" of $i$ if $A_{ij}>0$, i.e., there is at least one affirmation exchanged between $i$ and $j$. An entry $A_{ij}$ maybe 0 because $i$ and $j$ are contemporary, and $i$ did not exchange affirmations with $j$ or because they are not contemporary. Note that our framework does not require the network to be undirected and can work similarly with a directed network. In the data analysis section, we report results for both this undirected network and when directed networks are formed only based on sending affirmations or receiving affirmations. 
%(see Table \ref{rhotablerobustnessdirected}).}

We define the variable $Y^{(i)}_j$ as the outcome of individual $j$ as observed by $i$ at time $t_i^{exit}$. For a contemporary resident $j$, if $j$ has already left the TC (whether successful or unsuccessful) by time $t_i^{exit}$, then $Y^{(i)}_j = S_j$, and otherwise we set $Y^{(i)}_j = 0$. Formally, $Y^{(i)}_j$ can be defined as follows
\[
Y^{(i)}_j = \begin{cases}
    S_j & \text{ if } t_i^{entry} < t_{j}^{exit} < t_i^{exit} \\
    0 & \text{ if } t_j^{entry}< t_i^{exit} < t_{j}^{exit}\\
    99 & \text{otherwise},
\end{cases}
\]
where 99 is used as some arbitrary number. In the above definition, we set $Y_{j}^{i}$ to be 99 for subjects $j$ who did not overlap with $i$, either because they left the facility before $i$ entered, i.e., $t_j^{exit}<t_i^{entry}$, or because they entered the facility after $i$ left, i.e., $t_i^{exit}<t_j^{entry}$.  This is because  $i$ technically cannot ``observe" the outcome of anyone who left the TC before $i$ entered or entered the TC after $i$ left, and therefore, this value can be set to any arbitrary number. In what follows, for defining the causal estimands of role model effect, we will always condition on $A_{ij}>0$, which restricts to the population who are contemporary to $i$ and interacted with it. Clearly, the subjects who have left the TC even before $i$ entered or entered after $i$ exited cannot interact with $i$.

\subsection{The role model effect estimand}
\label{rolemodeleffect}

We assume that for each unit $i$, there is a $d$ dimensional vector of unobserved latent characteristics $U_i$ that affect both the outcome $S_i$ and the formation of network $A$. The average role model effect cannot be estimated without additional methodology from our observational data due to this unobserved homophily. We explain the problem using the following Directed Acyclic Graph (DAG) \citep{pearl2009causality,mcfowland2021estimating,sridhar2022estimating}. The vertices of the graph represent various random variables (Figure \ref{fig:DAG}). We use the circles in the vertices for unobserved variables and squares for the observed variables.

\usetikzlibrary{positioning}
%\definecolor{offwhite}{HTML}{F2EDED}
\tikzset{
    > = stealth,
    every node/.append style = {
        text = black
    },
    every path/.append style = {
        arrows = ->,
        draw = black,
        fill = black
    },
    hidden/.style = {
        draw = black,
        shape = circle,
        inner sep = 1pt
    }
    ,
    hidden1/.style = {
        draw = black,
        shape = rectangle    }
}
\begin{figure}[!htbp]
    \caption{Directed Acyclic Graph for Role Model Effect in TC}
    \label{fig:DAG}

    \begin{center}
    \tikz{
    \node[hidden1] (a) {$Y_{j}^{i}$};
    \node[hidden1] (z) [right = of a] {$S_{i}$};
    \node[hidden1] (b) [above right = of z] {$T_{i}$};
    \node[hidden1] (x) [below left = of z] {$T_{j}$};
    \node[hidden1] (c) [below right = of x] {$A_{ij}$};
        \node[hidden] (d) [left = of x] {$U_{j}$};
                \node[hidden] (e) [right = of b] {$U_{i}$};
    \path (e) edge (c);
    \path (d) edge (a);
        \path (d) edge (x);
    \path (a) edge (z);
    \path (b) edge (z);
    \path (c) edge (z);
    \path (d) edge (c);
       \path (x) edge (a);
              \path (e) edge (b);
       \path (e) edge (z);
}
\end{center}
\end{figure}

Generally speaking, we are interested in the causal effect of $Y_j^i$ on $S_i$ conditioning on $A_{ij}$. In the DAG, the causal path is 
\[
Y_{j}^{i} \rightarrow S_{i}.
\]
However, we can see that there are several non-causal or backdoor paths between $Y_{j}^{i}$ and $S_{i}$. These are paths with an incoming arrow into $Y_j^i$ \citep{pearl2009causality}. We can enumerate them below.
\begin{enumerate}
    \item $Y_{j}^{i} \leftarrow U_{j} \rightarrow A_{ij}\rightarrow S_{i}$
        \item $Y_{j}^{i} \leftarrow T_j \leftarrow  U_{j}\rightarrow A_{ij} \rightarrow S_{i}$
        \item $Y_{j}^{i} \leftarrow U_{j} \rightarrow A_{ij}\leftarrow U_{i} \rightarrow S_{i}$
        \item $Y_{j}^{i} \leftarrow  U_{j} \rightarrow A_{ij}\leftarrow U_{i}\rightarrow T_{i} \rightarrow S_{i}$       
        \item $Y_{j}^{i}  \leftarrow T_{j} \leftarrow U_{j} \rightarrow A_{ij} 
        \leftarrow U_i
        \rightarrow S_{i}$
        \item $Y_{j}^{i} \leftarrow T_{j} \leftarrow U_{j} \rightarrow A_{ij} 
        \leftarrow U_i
        \rightarrow T_i
        \rightarrow S_{i}$
\end{enumerate}

To estimate the effect along the causal path, we would like to close these non-causal backdoor paths. Among these backdoor paths, the last 4 paths would have been blocked due to the presence of the collider variable $A_{ij}$ (two incoming arrows into the variable,  \cite{pearl2009causality}). However, if we were to condition it on a collider, then that would unblock the backdoor path. Therefore, it is our conditioning on $A_{ij}$ that opens up the backdoor paths 3-6. This has also been noted in \cite{mcfowland2021estimating,sridhar2022estimating}. On the other hand, conditioning on $U_{i}, U_{j}$ closes all the six backdoor paths and allows us to identify the causal effect. 
Also, conditioning on $U_{i}, U_{j}$ does not unblock any other back door path as it is not a collider on any back door paths. In summary, our conditioning set comprises of $C=\{U_{i}, U_{j}, A_{ij}\}$ which satisfies the two backdoor criteria i) no node in $C$ contains any descendent of $Y_{j}^{i}$ (no post-treatment bias) and ii) $C$ blocks all the backdoor paths between $Y_{j}^{i}$ and $S_{i}$.

The causal ``role model'' effect that an individual $i$ receives from an individual $j$ (the role model)  can be described as follows. 
\begin{equation} 
\label{eqrolemodel}
\phi_{ij} = \mathbb E[S_i | A_{ij} >0, do(Y^{(i)}_j = 1)] - \mathbb E[S_i | A_{ij} >0, do(Y^{(i)}_j = 0) ],
\end{equation}
Here, the causal effect is defined using the do-operator for $Y^{(i)}_j$ \cite{pearl2009causality}. In this context, we consider the possibility that $Y^{(i)}_j$ can be changed by the experimenter. The $do(Y^{(i)}_j = a)$ can be interpreted as an intervention where the experimenter fixes $Y^{(i)}_j$ at a value $a$ without disturbing any other covariate in the outcome model  \citep{sridhar2022estimating,pearl2009causality}. The intervention $do(Y^{(i)}_j = 1)$ is then interpreted as the intervention that $j$ is made to graduate successfully before $i$, while the intervention $do(Y^{(i)}_j = 0)$ is the intervention that $j$ is either made to fail or made to exit after $i$,  where $j$ and $i$ exchanged affirmations during their stay in the unit. The expectation operator for the first term on the right-hand side is over the outcome distribution if the experimenter makes the peer $j$ graduate successfully before $i$ leaves the TC. The second term's expectation is over the outcome distribution if the experimenter makes $Y^{(i)}_j =0$. Therefore, the quantity $\phi_{ij}$ can be considered the difference between expectations of successful graduation of $i$ caused by the peer $j$ being a successful graduate, i.e., being a role model for $i$. The total role model influence of all $i's$ peers on $i$ can be written as $\phi_i = \sum_j \phi_{ij}.$

We also define a second causal estimand by changing the definition of the intervention. In the second definition, we only consider the peers who have left the TC before $t_i^{exit}$ and define $Y_j^{(i)} = S_j$ if $t_j^{exit} <t_i^{exit}$. Then for each individual, their potential peer network only consists of individuals who have interacted with them but exited the unit before them. Accordingly we redefine the role model effect  $\phi_{ij}^{1}$ as follows:
\begin{equation} 
\label{eqrolemodel1}
 \phi_{ij}^{1}  = \mathbb E[S_i | A_{ij} >0, t_j^{exit}<t_i^{exit}, do(S_j = 1)] - \mathbb E[S_i | A_{ij} >0, t_j^{exit}<t_i^{exit}, do(S_j = 0) ],   
\end{equation}

For a ``peer" $j$ (i.e., $A_{ij}>0$), who has left the TC before the exit date of $i$ (and therefore $i$ can observe their graduation status), the intervention $do(S_j=1)$ is then the $j$th peer is made to graduate successfully, and the intervention $do(S_j=0)$ is that the peer is made to fail.

We will use the adjacency matrix to estimate a proxy for the unobserved homophily $U$. We denote this proxy variable by $U^{*}$. Note that controlling for $U^{*}$ in the regression instead of the true $U$ does not close the backdoor pathways. Therefore, we continue to see bias in the peer influence parameter and cannot capture the causal effect of interest (\cite{mcfowland2021estimating} called this trading a larger omitted variable bias for a smaller measurement error bias). Nevertheless, we will show that using $U^{*}$ estimated through spectral embedding of the observed network leads to an asymptotically unbiased estimation of peer influence. 

Now that we have laid out the core issue about the causal effect, we provide our structural model below. We observe an $n \times p$ matrix $Z$ of measurements of $p$ dimensional covariates at each node. Using the DAG, we then propose the following linear structural model for the outcome.
\begin{align}
   S_i &= \alpha_0 + \gamma^T Z_i + \rho \frac{\sum_j A_{ij} Y_j^{(i)}}{\sum_j A_{ij}} +\beta^T U_i + \epsilon_S.
   \label{dgp}
\end{align}
We note here that this linear peer effects model is similar to the longitudinal model in \cite{mcfowland2021estimating,sridhar2022estimating} and differs from the equilibrium peer effect model of \cite{bramoulle2009identification} where the response is simultaneously present on both sides of the model. For the second formulation of causal effect $\phi_{ij}^{1}$, when we write the outcome equation, even though we have the same variable $S$ on both left and right-hand side of the equation, due to carefully tracking the exit times of the residents, if $S_j$ is in the equation for $S_i$, then $S_i$ does not appear in the equation for $S_j$. This is because $S_{i}$ is observed at time $t_{i}^{\text{exit}}$ and $S_j$ happens before $t_{i}^{\text{exit}}$. We will further model the network $A$ as a low-rank model in the next section. The following proposition, which is similar to Proposition 3.1 in \cite{sridhar2022estimating}, shows that if we can identify the parameter $\rho$, we can identify the total role model effect on vertex $i$ given by $\phi_i = \sum_j \phi_{ij}$.
\begin{prop}
   Assume the data generating model in Equation \ref{dgp} and suppose the $U_i$s are observed. Let $\phi_i$ be the total role model effect on node $i$ defined earlier. Then $\phi_i = \rho$ for all $i$. 
   \label{lemma1}
\end{prop}

We further remark that both $Y_{j}^{i}$ and $t_{j}^{exit} $ are dependent on $U_j$. For $t_{j}^{exit}$, as we explained earlier, while the time when they enter the TC is unrelated to their latent traits, the decision to stay longer or leave early may be related to latent traits. Therefore it is plausible that $U_j$ impacts $t_{j}^{exit}$. The dependence of  $Y_{j}^{i}$ on $U_j$ can be seen by noting that, for contemporary residents, 
\[
Y_{j}^{i} =S_j \mathbbm{1}\{t_{j}^{exit} < t_i^{exit}\}
\]
where $\mathbbm{1}\{.\}$ denotes the indicator function. Since both the final graduation status $S_j$ and $t_{j}^{exit}$ depends on $U_j$, therefore $Y_{j}^{i}$ also depends on $U_j$.

Lastly, we remark that we assume that $Z_i \ind A_{ij}|(U_i, U_j)$, that is the covariate $Z_i$ only affects the outcome and not the network formation. This assumption has also been made in \cite{mcfowland2021estimating}. Therefore, we do not include covariates in our network model. If, indeed, this assumption is violated and $Z_i$ does affect network formation, then since our network model is misspecified, we expect that our estimate ${U}^*_i$ of $U_i$ will be biased, and this will, in turn, lead to bias in $\rho$. In the empirical data analysis, we have also considered the possibility that $Z$ might also affect the network formation. Fortunately, since $Z$ is observed we can both include it in the network model and condition on it in the outcome model, and therefore avoid bias in estimation of $\rho.$ We provide an updated DAG with this relationship in the Appendix Figure \ref{updateddag}. This updated DAG gives rise to two additional non-causal backdoor paths (in addition to the six identified above). (1) $Y_{j}^{i}  \leftarrow Z_{j} \rightarrow A_{ij} \rightarrow S_{i}$ and (2) $Y_{j}^{i}  \leftarrow Z_{j} \rightarrow A_{ij} \leftarrow Z_{i} \rightarrow S_{i}$. 
Now, given the additional backdoor paths, the conditioning set will be updated to $C=\{U_{i},U_{j},A_{ij},Z_{i},Z_{j}\}$. The updated conditioning set satisfies the two backdoor criteria stated earlier. Therefore our structural model for the outcome does not change. However, we need to update the model for network formation to also include the covariates in $Z$.

\subsection{Estimating homophily from network model}
We develop the methods under a more general setting than our statistical problem such that the methods and the accompanying statistical theory are of independent interest and are applicable more widely to observational data in other contexts of social sciences. We assume we have access to a network encoding relational data among a set of $n$ entities whose adjacency matrix is $A$. Therefore, the element $A_{ij}$ is $1$ if $i$ and $j$ are connected and is $0$ otherwise. The diagonal elements of $A$ are assumed to be $0$. We observe an $n$ dimensional vector $Y$ of univariate responses at the vertices of the network over $t$ time points (with $t \geq 2$, but assumed to be finite). We further observe an $n \times p$ matrix $Z$ of measurements of $p$ dimensional covariates at each node. We assume the following data-generating model on the responses: 
\begin{equation}
Y_{i,t}=\alpha_0 + \alpha_1 Y_{i,(t-1)} + \beta^T U_{i} +  \rho  \frac{\sum_{j} A_{ij}Y_{j,(t-1)}}{\sum_j A_{ij}} + \gamma^T Z_i  + V_{i,t},
\label{narmodel}
\end{equation}
where $V_{i,t}$ are iid random variables with $\mathbbm E[V_{it}]=0$ and $\text{Var}[V_{it}]=\sigma^2$, and $U$ is a $n \times d$ matrix of latent homophily variables such that each row $U_i \in \mathbb{R}^d$ of the matrix represents a vector of latent variable values for a node. The parameter $\rho$ is the network peer influence parameter of interest. For any node $i$, the variable $\frac{\sum_{j} A_{ij}Y_{j,(t-1)}}{\sum_j A_{ij}}$, measures a weighted average of the responses of the network connected neighbors of $i$ in the previous time point $t-1$. Therefore, the above model asserts that the outcome of $i$ is a function of the weighted average of outcomes of its network-connected neighbors in the previous time point, values of the covariates, and a set of latent variables representing unobserved characteristics. Further, since we assume $t \geq 2$, writing the model in Equation \ref{narmodel} replacing $t$ with $t-1$, we can easily see that $Y_{j, (t-1)}$ is dependent on $U_j$, matching our descriptions of dependence of $Y_j^i$ on the latent variables from the previous section.

In the outcome model in Equation \eqref{narmodel}, we distinguish between \textit{observed} covariates $Z$ and \textit{unobserved} latent variables $U$. While they both can create omitted variable bias for estimating network influence,  we can directly control for $Z$ since the researcher observes those variables. The variables in $U$ on the other hand, are unobserved confounders that capture several unobserved characteristics of the individuals.

We assume that the selection of network neighbors happens based on homophily or similarity on those unobserved characteristics. Therefore, it is possible to statistically model and extract the latent homophily information from the observed network. More precisely, we model the network to be generated following the Random Dot Product Graph (RDPG) model \cite{sussman12,athreya2017statistical,athreya2016limit,tang2018limit,rubin2022statistical,xie2023efficient} defined as follows. Let $X_1, \ldots, X_n$ be $d$-dimensional vectors of latent positions such that $\|X_i\| \leq 1$ and $X_iX_j^T \in [0,1]$ for all $i\neq j$, where $\|.\|$ denotes the vector Euclidean norm. Let $\theta_n$ be a scaling factor such that,  
\begin{equation}
A \sim Bernoulli (P), \quad P = \theta_n XX^T.
\label{RDPG}
\end{equation}
We define the scaled latent positions under this model as $U= (\theta_n)^{1/2}X$. 

Clearly, $P=UU^T$, and we assume that this $U$ matrix is the same matrix of latent factors as in the outcome equation. The latent factors $U_i$ are therefore assumed to be obtained from a $d$ dimensional continuous latent space that satisfies the specified constraint. The scaling factor $\theta_n$ controls the sparsity of the resulting network with growing $n$ since the number of edges in the network is $O(n^2\theta_n)$. While $\theta_n =O(1)$ will lead to a dense graph, a typical poly-log degree growth rate $\theta_n = O(\frac{(\log n)^c}{n})$ for some $c>1$ leads to a sparse graph. The RDPG model contains the (positive semidefinite) Stochastic Block Model (SBM), Degree corrected SBM, and mixed membership SBM as special cases \citep{athreya2017statistical,rubin2022statistical}. The SBM and its extensions are random graph models with a latent community structure which have been extensively studied in the literature \citep{hll83,rcy11,lei2015consistency,pc15,athreya2017statistical,paul2020random}.

As stated earlier, we assume that the observed nodal covariates $Z_i$s do not directly affect the network formation and are, therefore, not part of the network generating process in Equation \eqref{RDPG}. However, the unobserved latent variables $U_i$s are correlated with the observed covariates $Z_i$s, the network links $A_{ij}$s, as well as the outcomes of the previous time point $Y_{i,t-1}$. Therefore, controlling for $U_i$s is important in determining both the network influence and the effect of the covariates $Z$.  In the RDPG model, since $P_{ij} = U_iU_j^T$, the probability of a connection between nodes $i$ and $j$ depends on their positions $U_i$ and $U_j$ on the underlying latent space. Nodes that are closer to each other in terms of direction (angular coordinate) are more likely to have a higher dot product and, consequently, higher propensity to form ties. Therefore, the variables $U_i$s can capture the unobserved characteristics of the individuals, which leads to the selection of network ties.

Since we have used the time lag of the response on the right-hand side of the equation, the model can be estimated using Ordinary Least Squares (OLS).  It was argued in  \cite{mcfowland2021estimating} that when the latent factors in \eqref{narmodel} are latent communities from the SBM, then replacing estimated communities in place of the true communities, one can obtain an asymptotically unbiased estimate of $\rho$. Under the more general RDPG model, we provide an explicit asymptotic upper bound on the bias of $\rho$. Further, we propose a method that allows us to obtain a bias-corrected estimator of $\rho$, which we show has some advantages over the estimator in \cite{mcfowland2021estimating} in finite samples.

We estimate the latent factors through a $d$ dimensional Spectral Embedding of the Adjacency matrix (ASE method). The spectral embedding performs a singular value decomposition (SVD) of the adjacency matrix $A$ (since $A$ is symmetric, one can also use eigendecomposition). Let $Q_{n \times d}$ be the singular vectors corresponding to the $d$ largest singular values and $\Sigma$ be the diagonal matrix containing those singular values. Then we estimate the latent factor as $\hat{U} =Q\Sigma^{1/2}$. In the second step, we use these estimated factors as predictors in the outcome model
\begin{equation}
Y_{i,t}=\alpha_0 + \alpha_1 Y_{i,(t-1)} +  \beta^T \hat{U}_{i} +  \rho  \frac{\sum_{j} A_{ij}Y_{j,(t-1)}}{\sum_j A_{ij}} + \gamma^T Z_i + V_{i,t}.
\label{estimatedmodel}
\end{equation}

However, similar to the method considered in \cite{mcfowland2021estimating}, replacing $U$ with estimated $\hat{U}$ will lead to bias in the estimates. 
For example, in the context of linear regression, when regressing $Y$ on just $\hat{U}$, it is well known that due to the presence of estimation error in $\hat{U}$, this estimator is biased in the finite sample and is biased asymptotically unless the estimation error vanishes  \citep{davies1975effect}. The problem is exacerbated in the context of the peer effect model as in \eqref{estimatedmodel} due to various dependencies.
However, the following result provides an asymptotic upper bound on the bias. The notation $g(n)=O(f(n))$ for two functions of $n$ means that $g(n)$ and $f(n)$ are of the same asymptotic order, i.e., the ratio $\frac{g(n)}{f(n)}$ is a constant $c$ that does not depend on $n$.
\begin{thm}
\label{theorem1}
    Assume the network $A $ is generated according to the $d$-dimensional RDPG model with parameter $\theta_n, X$ as described above. Let $\hat{U}$ be the estimated latent factor matrix from the Adjacency Spectral Embedding (ASE) method. Further assume that $n\theta_n = \omega (\log n)^{4c}$ for some $c>1$ and let $c'>0$ be another constant. Then the bias in the estimate $\hat{\rho}_{n}$ is given by 
    \[
    \mathbbm E[\hat{\rho}_{n} -\rho|A] = O\left(\frac{(\log n)^{2c}}{n\theta_n} +   \frac{1}{n^{c'}} +  \frac{(\log n)^{c}}{n^c (n\theta_n)^{1/2}}\right).
    \]
    \label{asympbias}
\end{thm}
Theorem \ref{asympbias} provides an asymptotic upper bound on the bias of the least squares estimator for $\rho$. Clearly, the bias decreases with increasing $n$, and as $n \to \infty$, the bias vanishes. We further note that as the network's density increases, i.e., $\theta_n$ increases, the bias in the OLS estimator decreases. This result can be compared with Theorems 1 and 2 in \cite{mcfowland2021estimating}. Compared to Theorem 1 of \cite{mcfowland2021estimating}, which was for SBM, this result holds for a more general model and allows for sparsity in the network. In comparison to Theorem 2 of \cite{mcfowland2021estimating}, which showed the asymptotic bias with continuous latent space model converges to 0, this result provides an explicit expression for the asymptotic bias as a function of $n$. 

We summarize the assumptions under which Theorem \ref{asympbias} is valid and comment on how they might be violated. 
\begin{enumerate}
    \item \textbf{Network model:} We have assumed that the network is generated from the random dot product graph model, which is a multiplicative latent variable model. While the model is a general model for network, it may not be a good fit to an observed network. The result in Theorem \ref{asympbias} also assumes that the number of dimensions $d$ is known. For the real data analysis, we estimate $d$ using a data-driven cross-validation procedure \cite{li2020network}. We also provide a robustness check using another model for the network - the additive and multiplicative latent variable model \cite{hoff2021additive}. 
    \item \textbf{Node-level covariates:} We have also assumed that the observed node level covariates may affect the outcome but are not needed in the model for network formation when latent positions are used in the model. This could be because the covariates do not explain anything more than the latent positions or because they are independent of the latent positions. As we have already discussed at the end of section \ref{rolemodeleffect}, this assumption may be violated, and then we need to incorporate the observed covariates in the model for the network data. 
    \item \textbf{Outcome model:} We have assumed that the outcome model is a linear function of own response from the previous time point, the average of peer response from the previous time point, observed covariates, and latent homophily variables. Such a linear relationship may not be appropriate for certain datasets, including if the outcome is binary or categorical or if there is other evidence of non-linearity in the relationship. In a subsequent subsection, we extend Theorem \ref{theorem1} to the case when the outcome model is a probit regression model (Theorem \ref{theorem3probit}).
\end{enumerate}

We remark that in our model, while $U\beta$ is identifiable, the parameter $\beta$ cannot be identified separately. This is because the latent variable $U$ can be estimated only up to the ambiguity of a matrix $R$ with orthonormal columns. Clearly, both $U$ and $UR$ will lead to the same $UU^T$ for an orthogonal matrix $R$ since $RR^T=I$. This is a limitation of all multiplicative latent variable models, including the RDPG model \citep{athreya2017statistical}, latent space model \citep{hoff2002latent}, and the additive and multiplicative effects model \citep{hoff2021additive}. Therefore, we will only be interested in identifying $\rho$ and $ \gamma$ parameters correctly.

\subsection{Bias-corrected estimator}
Next, we further construct a bias-corrected estimator where the central idea is to correct for the finite sample bias using the corrected score function methodology from the well-developed theory of measurement error models \citep{stefanski1985effects,stefanski1985covariate,schafer1987covariate,nakamura1990corrected,novick2002corrected}. Let $d_i = \sum_j A_{ij}$ and $D$ denote the diagonal matrix containing $d_i$s as the diagonal elements. Define $L = D^{-1}A$, and define $W=[1_n \quad Y_{t-1} \quad LY_{t-1} \quad Z] $ as the matrix collecting all the predictor variables except for $\hat{U}$. Let $\eta=[\alpha_0, \alpha_1, \rho, \gamma]$. If it is known that $\hat{U}_i = U_i + \xi_i$ and $\text{Var}(\xi_i)=\Delta_i$, i.e., if the error covariance matrices for the different nodes are known, then \cite{nakamura1990corrected} provides the following bias-corrected estimator for the linear regression model parameters. Define the following quantities.
\begin{align*}
    \Omega  = \begin{pmatrix}
\sum_i \Delta_i & 0 \\
0 & 0
    \end{pmatrix}, \, 
    M_{WU} = \begin{pmatrix}
        \hat{U}^T \\ W^T
    \end{pmatrix} \begin{pmatrix}
        \hat{U} & W
    \end{pmatrix}, \, 
    M_Y = \begin{pmatrix}
        \hat{U}^T \\ W^T
    \end{pmatrix} Y.
\end{align*}
Then, the bias-corrected estimators of $\eta, \beta, \sigma^2$ are :
\begin{align*}
\begin{pmatrix}
    \hat{\eta} \\  \hat{\beta}
\end{pmatrix}  = (M_{WU}-\Omega)^{-1}M_Y. 
\quad \hat{\sigma}^2  = \frac{1}{n-d-p}\left(\|Y-[\hat{U} \,\,W] \begin{pmatrix}
    \hat{\eta} \\  \hat{\beta}
\end{pmatrix}\|_2^2 - \hat{\beta}^T(\sum_i \Delta_i) \hat{\beta}\right),
\end{align*}
while the standard errors for the $j$th element of the vector of parameters $\hat{\eta}, \hat{\beta}$ can be obtained as $\sqrt{\hat{\sigma}^2 (M_{WU}-\Omega)^{-1}_{jj}}$.

The main difficulty in applying the above idea in our context is that, in practice, the covariance matrix of the measurement error is unknown. We propose to use the results from \cite{athreya2016limit,tang2018limit,xie2023efficient} on asymptotic convergence of estimated node position for a fixed node $i$ to a multivariate normal distribution. In particular, we use the version of the result in \cite{xie2023efficient}, which is stated without any specific distribution assumption on the latent positions. Recall, $X_1, \ldots, X_n \in \mathcal{X} $ are latent positions. %Define $F_n(X) = \frac{1}{n}\sum_i 1 (X_i \leq X)$.
Following \cite{xie2023efficient}, we assume that there exists a distribution $F$ on the latent space $\mathcal{X}$, and %such that $\sum_{X \in \mathcal{X}}|F_n(X) - F(X)| \to 0$, as $n \to \infty$. 
define the second moment matrix of the random variable $X$ with respect to this limiting distribution $F$ as, $\Delta_F = \mathbbm E[XX^T] =  \int_{\mathcal{X}}XX^T F(dX)$. Then define the $d \times d$ matrix-valued function $\Sigma(x)$ of a vector $x$, as
\[
\Sigma(x) = \Delta_F^{-1}[\int_{X \in \mathcal{X}}\{X^Tx(1-\theta X^Tx)\}XX^TF(dX)]\Delta_F^{-1},
\]
where $\theta = \lim_{n \to \infty} \theta_n$. Then \cite{xie2023efficient} shows that if $A \sim RDPG (\theta_n, X_0)$ with $n\theta_n = \omega(\log n)^{4c}$, for some ``true" latent positions $X_0 = [X_{01},\ldots X_{0n}] \in \mathcal{X}^n$, then 
\[\sqrt{n}W_n \theta_n^{1/2}(\hat{X}_i -  X_{0i})  \to  N(0, \Sigma(X_{0i})).
\]
Here $\theta_n^{1/2}(\hat{X}_i -  X_{0i}) =\hat{U}_i$ is the estimated latent position from the ASE method described above, and $W_n$ is a sequence of orthogonal matrices. This result states that the estimated latent positions from ASE, suitably normalized, converge to an asymptotic normal distribution with a finite limiting covariance matrix. Therefore, an estimate of $\Sigma(X_{0i})$ from data will give us an estimate of the covariance matrix of the latent variables.

Recall that if $A \sim RDPG (\theta_n, X_0)$  then the latent variable for the $i$th node is given by $U_i =\theta_n^{1/2}X_{0i}$. To estimate this covariance matrix of the error $\Delta_i= \Sigma(X_{0i})$, we propose to replace the true latent position $X_{0i}$s with its estimate from the ASE, $\hat{X}_i = (\theta_n)^{-1/2}\hat{U}_i$. The resulting estimate of the covariance matrix for the $i$th node, $\hat{\Delta}_i$ is given below:
\begin{equation}
\hat{\Delta}_i = (\frac{1}{n}\sum_i \hat{U}_i\hat{U}_i^T)^{-1}(\sum_j \hat{U}_j^T\hat{U}_i(1-\hat{U}_j^T\hat{U}_i)\hat{U}_j \hat{U}_j^T)(\frac{1}{n}\sum_i \hat{U}_i\hat{U}_i^T)^{-1}.
\label{Deltai}
\end{equation}
Note that this expression does not involve $\theta_n$, and therefore, we have no need to estimate this sparsity parameter from data.

In the special case of SBM, the covariance matrix simplifies further. We assume the number of communities is the same as the number of dimensions of the latent positions, $d$. Let $B_k$ denote the unique row corresponding to the $k$th community and $c_i \in \{1,\ldots, d\}$ denotes the community the node $i$ belongs to. Let $\pi_k$ denote the proportion of nodes that belong to the community $k$. Then to apply corollary 2.3 of  \cite{tang2018limit}, we compute $\Delta_F =\sum_{k=1}^{K} \pi_k B_kB_k^T$ and
$\Sigma(B_q) =\Delta_F^{-1}(\sum_{k=1}^K \pi_k B_kB_k^T (B_q^TB_k- (B_q^TB_k)^2))\Delta_F^{-1}$. We propose to estimate $B_q$ with its natural estimate $\hat{B}_q$, which are the cluster centers, and $\pi_q$ with $\hat{\pi}_q$, which are the cluster size proportions. Therefore a plug-in estimator for the covariance matrix $\hat{\Delta}_q$ for any $\hat{U}_i$ whose true community is $q$, is
 \begin{align}   
 \hat{\Delta}_q = \hat{\Delta}_F^{-1}(\sum_k \hat{\pi}_q(\hat{B}_q^T\hat{B}_k - (\hat{B}_q^T\hat{B}_k)^2)\hat{B}_k\hat{B}_k^T)\hat{\Delta}_F^{-1},
 \label{Deltaq}
 \end{align}
where $\hat{\Delta}_F=\frac{1}{n}\left(\sum_k \hat{\pi}_k\hat{B}_k\hat{B}_k^T\right)$. 

Estimating the node locations themselves involves about $O(d)$ parameters, while the covariances are $O(d^2)$ parameters. Therefore, with dense enough graphs, we can estimate node variances separately using the formula in Equation \ref{Deltai}. However, for sparse graphs when there is a lack of data, we may not be able to estimate separate covariance matrices for each node well, and in that case, the second approach of estimating community-wise node variances using Equation \ref{Deltaq} might be beneficial. We find in our simulations that the community-wise estimate of the covariance matrix using Equation \ref{Deltaq} works well for the bias correction procedure even when the data-generating network model is more general than SBM.

\begin{algorithm}[H]
  \caption{Homophily and Measurement Bias Adjusted Network Influence Estimation}
  \label{alg:estimation}
  
  \begin{algorithmic}
    \STATE {\bfseries Input:} Network Adjacency matrix $A$, Responses $Y_1, \ldots, Y_t$, Observed Covariates $Z$, dimension of latent factors $d$
    \STATE {\bfseries Result:} Model parameters ($\hat{\alpha}_0, \hat{\alpha}_1, \hat{\beta},\hat{\gamma},\hat{\sigma}^2,\hat{\rho}$)
  \end{algorithmic}
  
  \begin{algorithmic}[1]
      \STATE SVD: $A= Q \Sigma Q^T, \quad \hat{U} = Q[1:d] (\Sigma[1:d])^{1/2} $ 
\STATE Estimate $\hat{\Delta}_i$ as,  $\hat{\Delta}_i = \hat{\Delta}_F^{-1}(\sum_j \hat{U}_j^T\hat{U}_i(1-\hat{U}_j^T\hat{U}_i)\hat{U}_j \hat{U}_j^T)\hat{\Delta}_F^{-1}$, with $\hat{\Delta}_F = \frac{1}{n}\sum_i \hat{U}_i\hat{U}_i^T$.
    \STATE $L=D^{-1}A$ where $D_{ii} =  \sum_j A_{ij}$ and $D_{ij}=0$ for $j \neq i$
\STATE $W=[1_n \quad Y_{t-1} \quad LY_{t-1} \quad Z] $
\STATE $
    \Omega  = \begin{pmatrix}
\sum_i \Delta_i & 0 \\
0 & 0
    \end{pmatrix}, \, 
    M_{WU} = \begin{pmatrix}
        \hat{U}^T \\ W^T
    \end{pmatrix} \begin{pmatrix}
        \hat{U} & W
    \end{pmatrix}, \, 
    M_Y = \begin{pmatrix}
        \hat{U}^T \\ W^T
    \end{pmatrix} Y.
$

\STATE $\begin{pmatrix}
   \hat{\alpha}_0 & \hat{\alpha}_1 & \hat{\gamma} & \hat{\rho} &  \hat{\beta}
\end{pmatrix}^T  = (M_{WU}-\Omega)^{-1}M_Y.$  

\STATE $\hat{\sigma}^2  = \frac{1}{n-d-p}\left(\|Y-[\hat{U} \,\,W] \begin{pmatrix}
    \hat{\eta} \\  \hat{\beta}
\end{pmatrix}\|_2^2 - \hat{\beta}^T(\sum_i \Delta_i) \hat{\beta}\right)$
    \STATE \textbf{return} $[   \hat{\alpha}_0, \hat{\alpha}_1, \hat{\gamma}, \hat{\rho},  \hat{\beta}, \hat{\sigma}^2 ]$
  \end{algorithmic}
\end{algorithm}

We also note that the results in \cite{tang2018limit,athreya2016limit,xie2023efficient} hold nodewise and therefore do not hold simultaneously for all $n$ nodes. However, we show in the simulations that our measurement bias correction methods with this covariance matrix estimate provide effective bias corrections, especially in small samples. To motivate the bias correction procedure, suppose we have access to the true latent positions for all nodes except for node $i$, i.e., we observe $U_1, \ldots U_{i-1}, U_{i+1}, \ldots U_n$. We call these latent position vectors together as $U_{-i}$ for ease of notation. For the $i$th node, we use our estimate $\hat{U}_i$ from the ASE. As described earlier, we can write the linear regression model of interest as $Y= \tilde{U}\beta + W\eta + V$, with $\mathbbm E[V]=0$, where $\tilde{U} = U_1, \ldots U_{i-1}, \hat{U}_i, U_{i+1}, \ldots U_n$. 
We have the following theorem.
\begin{thm}
    Consider the RDPG model described above with latent positions $U_1, \ldots, U_n$, and $n\theta_n = \omega(\log n)^{4c}$, and the linear regression model with the latent factors as $Y= \tilde{U}\beta + W\eta + V$, with $\mathbbm E[V]=0$. Here $\tilde{U} = U_1, \ldots U_{i-1}, \hat{U}_i, U_{i+1}, \ldots U_n$, where $U_{-i}$ be known latent positions, and $\hat{U}_i$ is estimated from the ASE method. Then the estimator using $U_{-i}, \hat{U}_i$ and without bias-correction, 
$\begin{pmatrix}
    \hat{\eta} &  \hat{\beta}
\end{pmatrix}^{T,(old)} =(\sum_j W_jW_j^T + \sum_{j\neq i} U_j U_j^T + \hat{U}_i\hat{U}_i^T)^{-1}(\sum_j Y_jW_j + \sum_{j\neq i} Y_j U_j + Y_i\hat{U}_i) $ converges in probability to $(\mathbbm E[W_iW_i^T] + \Delta_F/\theta + \Sigma(U_i)/n)^{-1}(\mathbbm E[W_iW_i^T] +\Delta_F/\theta) \begin{pmatrix}
    \eta_0 &  \beta_0
\end{pmatrix}^T$. However, the bias-corrected estimator $\begin{pmatrix}
    \hat{\eta} &  \hat{\beta}
\end{pmatrix}^{T,(new)} = (\sum_j W_jW_j^T + \sum_{j\neq i} U_j U_j^T + \hat{U}_i\hat{U}_i^T - \Sigma(U_i))^{-1}(\sum_j Y_jW_j + \sum_{j\neq i} Y_j U_j + Y_i\hat{U}_i)$ converges to $\begin{pmatrix}
    \eta_0 &  \beta_0
\end{pmatrix}^T$.
\label{theo2}
\end{thm}
The above theorem shows that the bias correction procedure converges to the target parameter when only one latent position is taken from the output of the ASE algorithm while the others are known. Our bias correction procedure generalizes this motivation for all nodes. The entire method for estimating peer influence, including homophily estimation and bias adjustment, is given in Algorithm \ref{alg:estimation}. This theorem also holds if, instead of one, a finite number of positions (say $s$ positions) are estimated from the ASE method. In that case $\tilde{U} = U_1, \ldots U_{i-s}, \hat{U}_1,\ldots \hat{U}_s, U_{i+1}, \ldots U_n$. Without the bias correction, the parameter estimate is biased, and the amount of bias can be characterized by the attenuation factor $(\mathbbm E[W_iW_i^T] +\mathbbm E[U_iU_i^T] + \Sigma(U_i)s/n)^{-1}(\mathbbm E[W_iW_i^T] +\mathbbm E[U_iU_i^T])$ when $s$ of the latent positions are estimated instead of being known. From the expression of the attenuation factor, we can see that it depends on the relative magnitude of the covariance of error due to estimated latent position $\Sigma(U_i)$ with respect to the second moment matrix of the true latent positions $\mathbbm E[U_iU_i^T] = \Delta_F/\theta$. The attenuation bias will be higher if $\Sigma(U_i)$ is comparable to $ \Delta_F/\theta$, leading to more error due to estimated latent positions. On the other hand, if the error in estimating latent positions is lower compared to $\Delta_F/\theta$, then the attenuation bias will not be significant.

Here we remark that there are two assumptions that this bias corrected homophily adjusted estimator relies on. First, the additive measurement error is an assumption under which the bias correction procedure has been developed. Such an assumption is commonly made in the literature. The measurement errors are typically assumed to be additive, and the covariances are learned from validation data or repeated measures. Therefore, the covariance of the measurement error is ``estimated". For example \cite{schafer1993likelihood} assumes that in place of true covariates $x_i$, one has proxy measurements $z_{ij}, \quad j=\{1,\ldots,m_i\}$ are available, where the $z_{ij}$s follow a normal distribution with (unknown) mean $x_i$ and (unknown) variance $\sigma^2_z$. The measurement error correction methodology then proceeds by estimating the mean and the covariance from this data. Our setup of estimating the true unknown latent position along with a covariance matrix of uncertainty around the estimate could be considered in a similar vein. Second, the bias correction procedure involves a summation of individual node-specific covariance matrix estimates. While the node-wise covariance estimates are theoretically valid, it is unclear if they hold for all $n$ nodes together.

However, our simulation study setup presents us with a compelling opportunity to validate both of these assumptions. Note our true data generating model, which is used in the simulation, Equation \ref{narmodel} uses the true latent factors or embeddings $U_i$s. The bias-corrected homophily adjusted estimator is derived with the above two assumptions of additive measurement error and validity of covariance matrix estimators on the node-wise embeddings $\hat{U}_i$s. However, these assumptions do not appear in the data-generating process and are purely used for obtaining the bias-corrected homophily adjusted estimator. Yet, we see in the simulation results that this estimator is able to provide an effective bias correction and recover the true peer influence parameter well (Figures \ref{fig:dcsbm} (a), (b), \ref{fig:sbm} (a), (b) in Section \ref{sims}). We, therefore, argue that this shows the estimator derived under our assumptions provides a good solution for the problem and, consequently, validates the assumptions.

\subsection{Result on probit model for binary response}

In this section, we propose a similar adjustment for latent homophily when the response is binary. We can make two distinct choices for the data-generating structural model here. We can use the same linear structural equation as in the case of continuous responses. This model is often referred to as the Linear Probability Model (LPM). The Ordinary Least Squares (OLS) is known to produce unbiased and consistent estimates in LPM \citep{wooldridge2010econometric} and consequently, the theory on latent homophily adjustments and bias correction from the previous sections apply.

However, the LPM may not be considered an appropriate model for binary responses due to the inability of a linear function of predictors to remain within the range $(0,1)$ and the inability of continuous error distributions to model binary responses. The LPM also implies that changing a predictor by unit amount changes the response by the same amount always, and therefore, as the value of covariates increases, will surely drive the estimated probability to outside the $(0,1)$ range \citep{wooldridge2010econometric}. Therefore, a second choice is a probit structural model for the binary response $Y$ as follows.
\begin{align}
Y_{i,t} & = 1 (Y_{i,t}^* >0) \nonumber \\
Y_{i,t}^* & =\alpha_0 + \alpha_1 Y_{i,(t-1)} + \beta^T U_{i} +  \rho  \frac{\sum_{j} A_{ij}Y_{j,(t-1)}}{\sum_j A_{ij}} + \gamma^T Z_i  + V_{i,t} \nonumber\\
V_{i,t} & \overset{i.i.d}{\sim} N(0,1).
\label{probitnar}
\end{align}
Here $Y_{i,t}^*$ denotes a latent variable such that if the variable takes a value greater than 0, then $Y_{i,t}$ is 1, and otherwise, $Y_{i,t}$ is 0. One can readily see that this model implies 
\[
P(Y_{i,t}=1) = \Phi(\alpha_0 + \alpha_1 Y_{i,(t-1)} + \beta^T U_{i} +  \rho  \frac{\sum_{j} A_{ij}Y_{j,(t-1)}}{\sum_j A_{ij}} + \gamma^T Z_i ), 
\]
where $\Phi(\cdot)$ denote the CDF function of the standard normal distribution. We borrow notations from the previous section and denote $W=[1_n \quad Y_{t-1} \quad LY_{t-1} \quad Z] $ as the matrix collecting all the observed predictor variables, and let $\eta=[\alpha_0, \alpha_1, \rho, \gamma]$. Therefore we can rewrite the above equation as $P(Y_{i,t}=1) = \Phi(\eta^TW_i + \beta^TU_i)$.

We will continue to call $\rho$ the peer influence parameter, however, the interpretation of this parameter is different from the linear model setup. For a probit model, it is natural to consider the partial effect of $LY_{t-1}$ defined as \cite{wooldridge2010econometric}
\[
PE_i = \frac{\partial P(Y_{i,t}=1|W_i,U_i)}{\partial (LY_{t-1})} = \rho \phi(\eta^TW_i + \beta^T U_i),
\]
where $\phi(\cdot)$ is the PDF function of the standard normal distribution. The partial effect for peer influence has the interpretation of how much the expected response (or probability of the response being 1) changes for one unit increase in the peer effects $LY_{t-1}$. Therefore the PE would be a comparable quantity to the peer effect from the linear model. However, there is an additional issue that unlike the linear models, the partial effects depend on the values of the covariates and the latent characteristics for the $i$th subject and hence are different for different subjects. Therefore we will be interested in the Average Partial Effects (APE) which averages the PE over all subjects.

Next, we make two additional assumptions on distributions of covariates. We assume that for any node $i$, $U_i$ and $\hat{U}_i$ are normally distributed. Further 
 for any node $i$, the variable $\frac{\sum_{j} A_{ij}Y_{j,(t-1)}}{\sum_j A_{ij}}$ is also normally distributed. The last assumption can be justified by noting that this variable is a weighted average of $n$ binary responses and, for large $n$, should behave similarly to a normal distribution. In probit models, typically, the true parameter $\rho$ associated with the peer influence predictor $LY_{t-1}$ cannot be estimated consistently when there are omitted variables, even when the omitted variable is independent of the covariate. When the omitted variables, in addition, are dependent on the observed covariate, then there are additional sources of bias. 
 We can rewrite $Y^*_{it}$ as follows,
 \[
 Y^*_{i,t} = \eta^TW_i + \beta_1^T\hat{U}_i + \beta^TU_i - \beta_1^T\hat{U}_i + V_{i,t}.
 \]
We let $q=\beta^TU_i - \beta_1^T\hat{U}_i$. Due to our assumptions of normality on $U_i, \hat{U}_i, (LY_{t-1})_i$, we can state that the conditional distribution of $q$ given $LY_{t-1}$ is $ q|(LY_{t-1})_i \sim N(\delta(LY_{t-1})_i, \tau^2)$ for some parameters $\delta, \tau^2$. Consequently, $(V_{it} + q)|LY_{t-1} \sim N(\delta(LY_{t-1})_i, \tau^2 +1 )$. Therefore
\begin{align*}
P(Y_{i,t}=1|W_i, \hat{U}_i) & = P(V_{it} + q < -\eta^TW_i - \beta_1^T\hat{U}_i) \\
&= P\left(\frac{V_{it} + q - \delta(LY_{t-1})_i }{\sqrt{\tau^2 +1}} < \frac{- \eta^TW_i - \beta_1^T\hat{U}-\delta (LY_{t-1})_i}{\sqrt{\tau^2 +1}}\right) \\
& = \Phi (\frac{ \eta^TW_i + \beta_1^T\hat{U}+\delta (LY_{t-1})_i}{\sqrt{\tau^2 +1}}).
\end{align*}
The partial effect is then obtained by 
 \[
 PE =  \frac{\rho+\delta}{\sqrt{\tau^2 +1}}  \phi(\frac{ \eta^TW_i + \beta_1^T\hat{U}+\delta (LY_{t-1})_i}{\sqrt{\tau^2 +1}})].
 \]

Therefore, we need to verify two things. First, we will show that $\text{Cov}(q, (LY_{t-1})_i|A,\hat{U}_i)$ converges to 0, and due to assumptions of normality, this implies that $q$ is independent of $(LY_{t-1})_i$. Consequently, $\delta$ converges to 0. Further, we will show that the conditional variance $\text{Var}(q|A)$ converges to 0 as $n$ goes to infinity. These are summarized in the following theorem, whose proof appears in the Appendix.

\begin{thm}
\label{theorem3probit}
    Assume the network $A $ is generated according to the $d$-dimensional RDPG model with parameter $\theta_n, X$ as described above and let the binary response $Y$ is generated following a probit model as in Equation \ref{probitnar}. 
    %Let $\hat{U}$ be the estimated latent factor matrix from the Adjacency Spectral Embedding (ASE) method. 
    Assume that $n\theta_n = \omega (\log n)^{4c}$ for some $c>1$. In addition, assume that the variables $U_i, \hat{U}_i$, and $(LY_{t-1})_i$ are all normally distributed for large $n$. Then, as $n \to \infty$, 
    both the peer influence parameter $\rho$ and the partial effect of peer influence is identified.
    \label{asympbiasprobit}
\end{thm}

\section{Simulation}
\label{sims}

We perform several simulation studies to compare the bias-corrected estimator using the estimated covariance matrix from ASE with (1) the estimator without any latent homophily variable and (2) the estimator with latent homophily vector from ASE but not corrected for measurement bias. 
In all cases, the network $A$ is generated from a specific instance of the RDPG model with 2-dimensional latent homophily variables collected in the matrix $U$. 
Next we generate the response at the first time point as $Y_{i1} =V_{i1}$, at the second time point as $Y_{i2}=\beta^TU_i  + \alpha Y_{i1} +  V_{i2}$ and at the third time point as $Y_{i3}=\alpha Y_{i2} +  \beta^TU_i + \rho \sum_j L_{ij}Y_{j2} + V_{i3}$, with $V_{i1}, V_{i2}, V_{i3}$ being generated i.i.d from $N(0,1)$ distribution. We set the parameters $\alpha=0.6$ and $\rho = 0.3$ and vary different parameters to investigate a variety of simulation setups. 

First, to understand the latent positions generated by various submodels of the RDPG class of models, we plot the simulated true latent positions along with those estimated from the spectral embedding of the adjacency matrix for 4 submodels of RDPG, namely, SBM, DC-SBM, MM-SBM, and DC-MM-SBM. In the case of SBM, the true latent positions are just two blue dots, each corresponding to one of the clusters. This underscores that in SBM, all nodes that belong to the same cluster have the same latent position. Therefore, one would expect the SBM to have limited ability to model an observed network well. The true latent positions for DCSBM are along two straight lines, one for each cluster. This is because the vectors of positions for nodes within a cluster are different only by their lengths, not in their direction. On the other hand, the latent positions in MMSBM can be considered an interpolation between two cluster center dots. Finally, the DC-MMSBM model has the most variation in terms of locations of nodes in latent space as the vectors differ in both direction and length.
\begin{figure}[!htbp]
\caption{True and Estimated Latent Positions for 4 sub-models within the RDPG Class.}
 \label{fig:latentpositions}
\begin{subfigure}{0.23 \textwidth}
     \centering
\includegraphics[width= \textwidth]{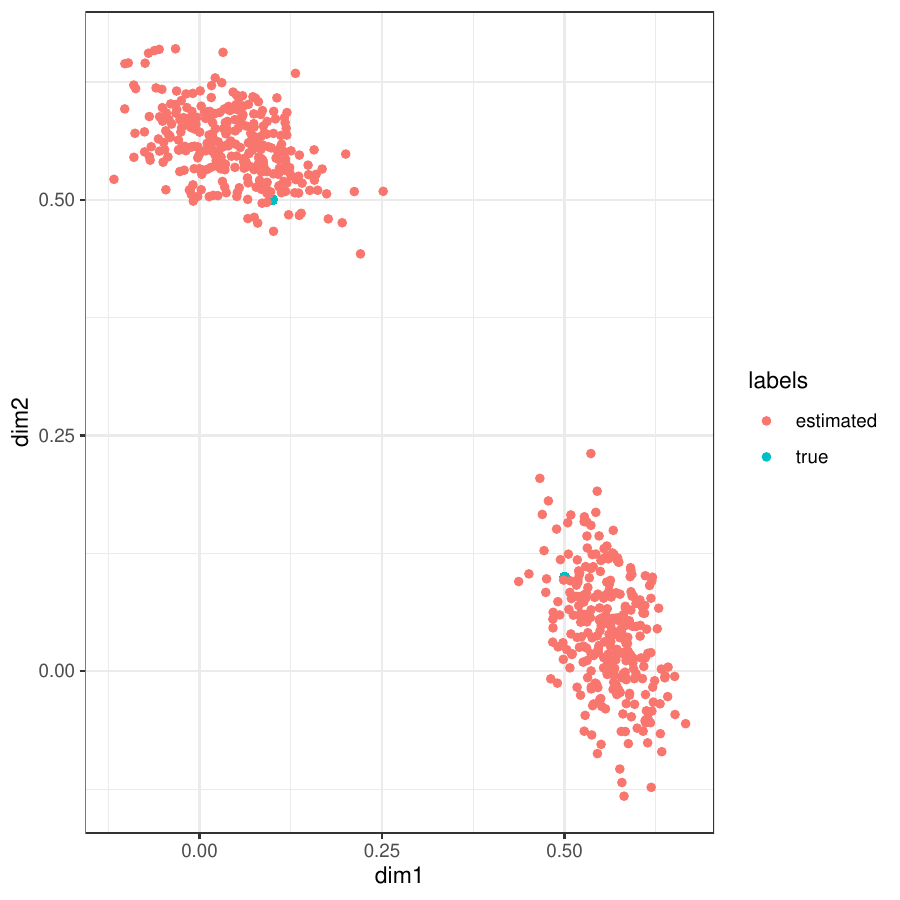} 
\subcaption{SBM}
\end{subfigure}%
\begin{subfigure}{0.23\textwidth}
     \centering
\includegraphics[width= \textwidth]{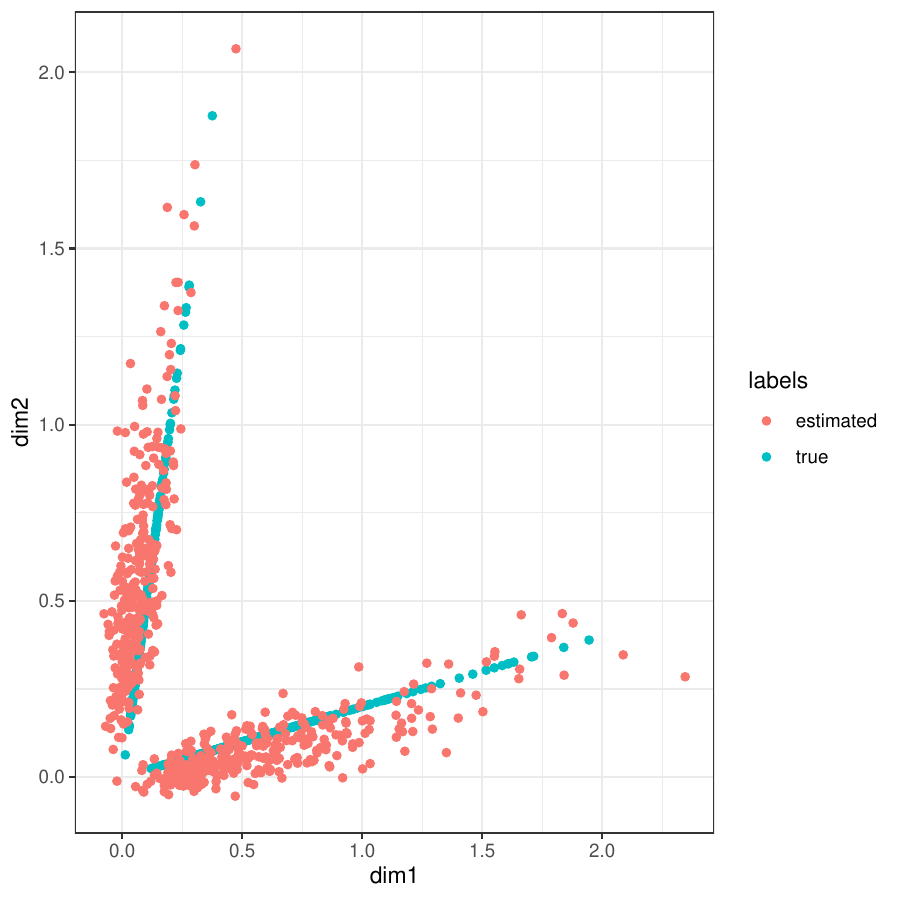} 
\subcaption{DCSBM}
\end{subfigure}
\begin{subfigure}{0.23 \textwidth}
     \centering
\includegraphics[width= \textwidth]{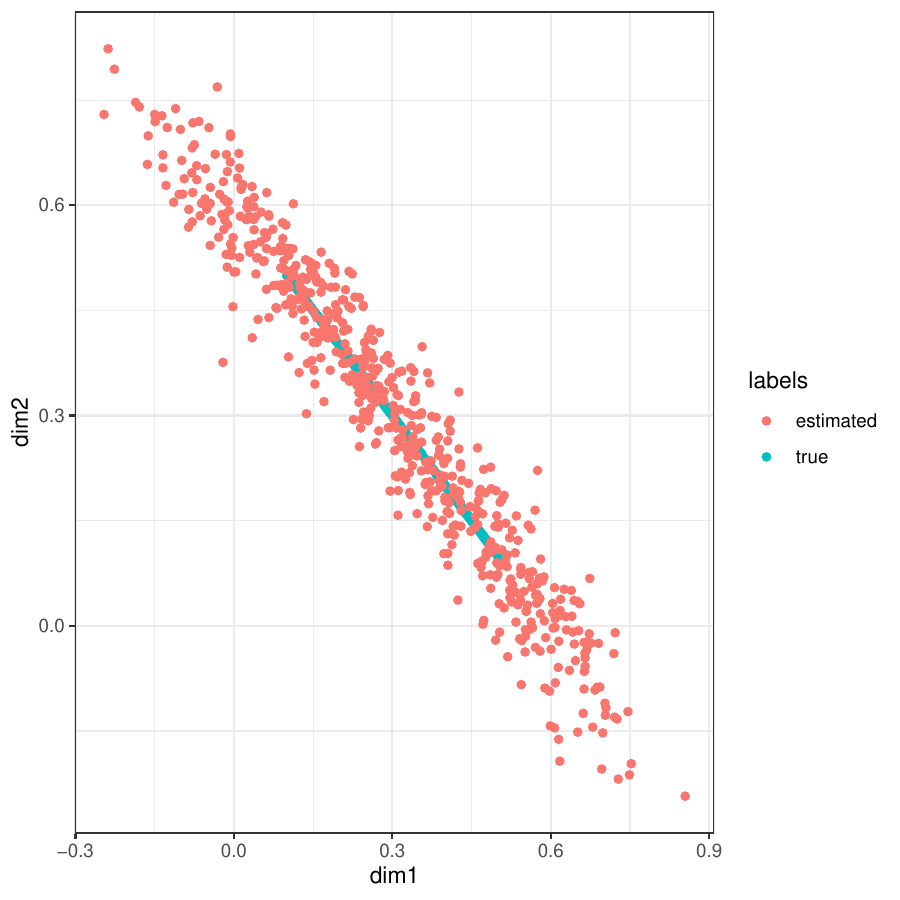} 
\subcaption{MMSBM}
\end{subfigure}%
\begin{subfigure}{0.23 \textwidth}
     \centering
\includegraphics[width= \textwidth]{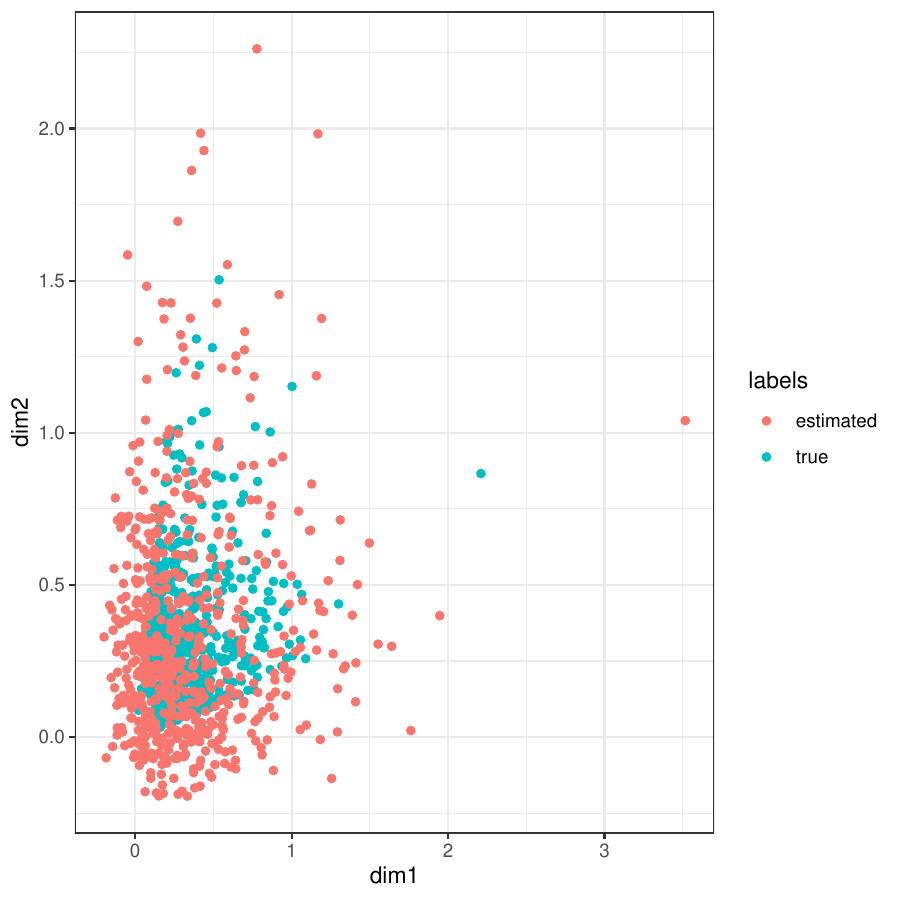} 
\subcaption{DC-MMSBM}
\end{subfigure}
 
\end{figure}

\subsection{DCSBM graph with increasing nodes}
\label{section31}
For this simulation, we generate a network from the Degree-corrected Stochastic Block Model (DCSBM), which is a special case of the RDPG model with the number of nodes increasing from 100 to 700 in increments of 100. Every node has a different latent position under this model. The degree heterogeneity parameters are generated from a log-normal distribution with log mean set at 0 and log SD set at 0.5. The matrix $U$ is formed as $U=\Theta H J$, where $\Theta$ is the diagonal matrix of degree parameters, $H$ is the community assignment matrix whose $i$th row is such that only one element takes the value of $1$, indicating its community membership and all other elements are 0s, and $J=\begin{pmatrix}
    0.5 & 0.1 \\
    0.1 & 0.5
\end{pmatrix} $ is a matrix whose rows provide the direction of the cluster centers. The community assignments are generated from a multinomial distribution with equal class probabilities. We form the matrix as $P=UU^T$ and scale all elements by a number to make the average density of the graph 0.20. We set $\beta=(1,3)$ in the outcome model. We report the bias and Standard Error (SE) of the bias for the three competing methods in Figure \ref{fig:dcsbm} (left). As the figure shows, the estimator without homophily correction remains biased even when the sample size increases, while the bias in the two homophily-corrected estimators decreases as the number of nodes increases. The figure further shows that the proposed measurement error bias-corrected estimator for the peer effect parameter $\rho$ has less bias than the estimator with homophily control but no bias correction,  especially in small samples. The bias in the parameter estimate goes close to 0 more quickly with the measurement error correction. Therefore, the proposed bias correction methodology works well.  
\begin{figure}[!htbp]
\caption{Comparison of estimates of the peer influence parameter $\rho$ when the underlying graph is generated from a DCSBM model}
 \label{fig:dcsbm}
\begin{subfigure}{0.5 \textwidth}
     \centering
\includegraphics[width= \textwidth]{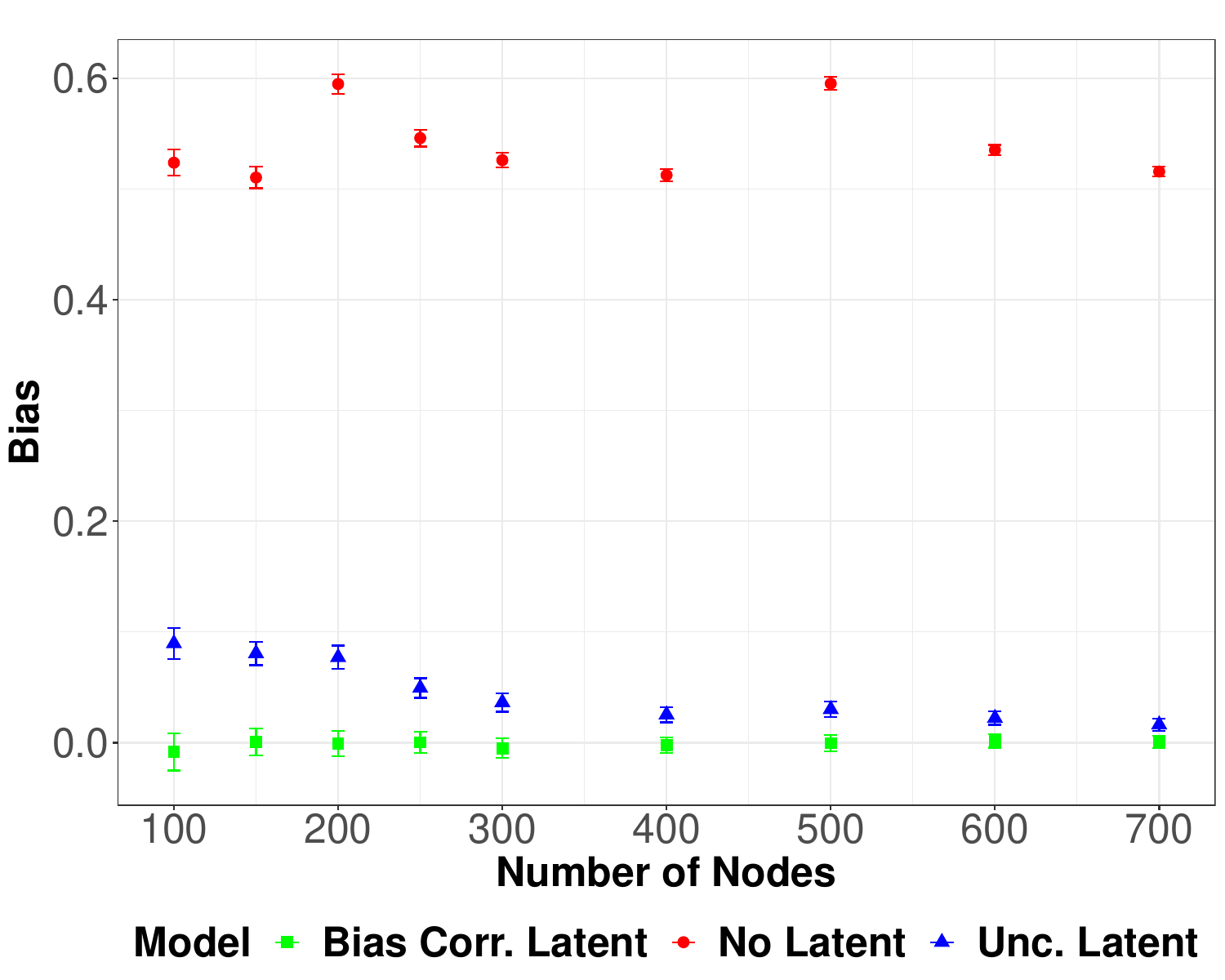} 
\subcaption{Increasing number of nodes}
\end{subfigure}%
\begin{subfigure}{0.5 \textwidth}
     \centering
\includegraphics[width= \textwidth]{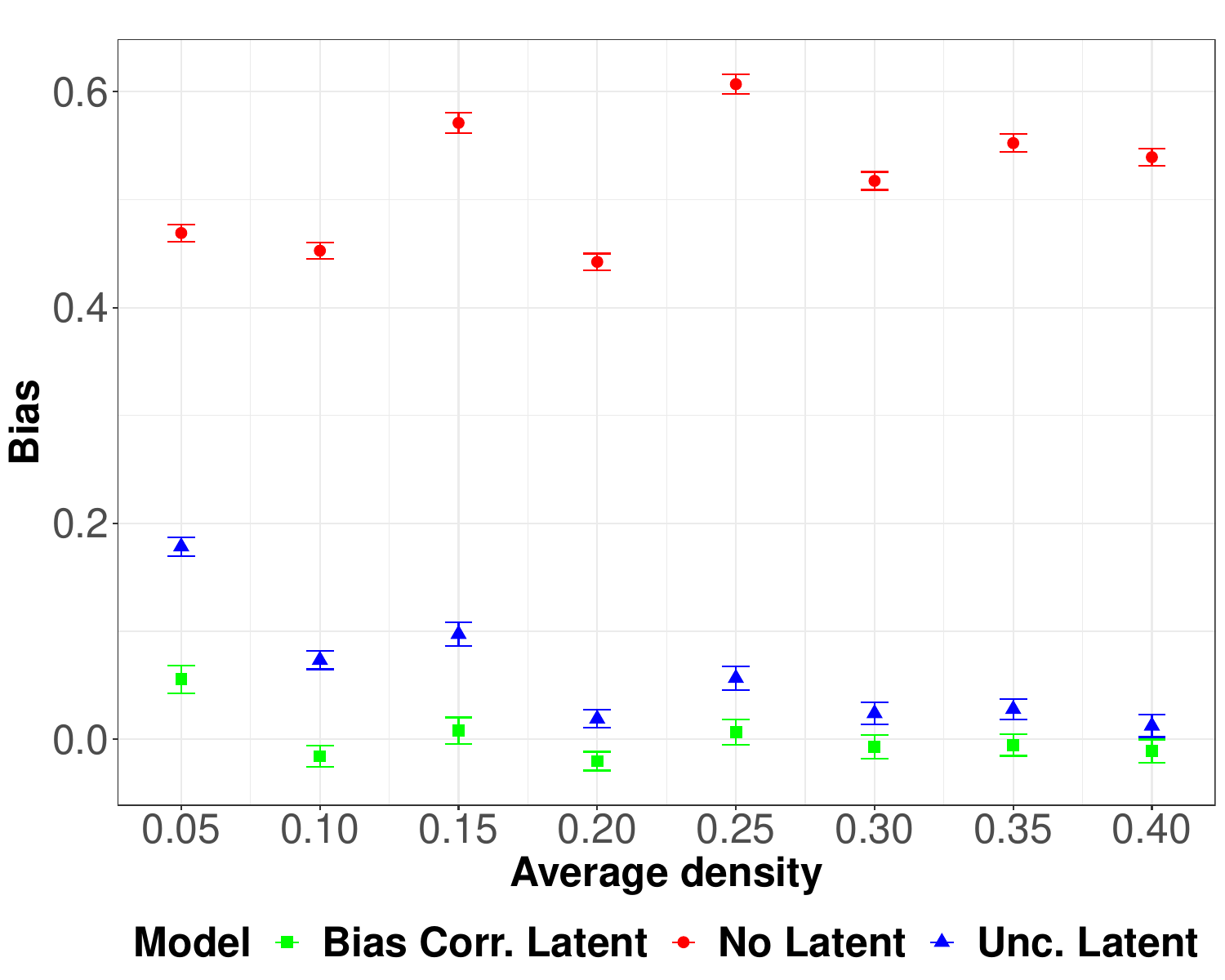} 
\subcaption{Increasing density of the graph}
\end{subfigure}
 
\end{figure}

\subsection{DCSBM graph with increasing density}
Next, we again generate the networks from a DCSBM model but increase the average density of the graphs from approximately 0.05 to 0.40, fixing the number of nodes at 200. For this simulation, we set the same $J$ and $\beta$ as the previous simulation, and the parameters in $\Theta$ are generated from the lognormal distribution with the same mean and SD as in the previous simulation. As Figure \ref{fig:dcsbm} (panel (b)) shows, the bias-corrected latent factor method outperforms the other two methods. We also see that the bias of both the uncorrected estimator and the bias-corrected decreases with increasing density as predicted by Theorem \ref{asympbias}, while the bias of the estimator without latent factors remains high.

 \subsection{Increasing nodes, SBM underlying graph} Now we consider the case when the data is generated from the stochastic block model. For this purpose, we set the $J$ matrix to a matrix which has 2 columns and 4 rows as $
J = \begin{pmatrix}
  0.7 &  0.2 \\
 0.1 &  0.6  \\
 0.2 &  0.2  \\
0.5 & 0.5  \\
\end{pmatrix}$. The matrix of latent positions $U$ is then formed as $U=HJ$. This leads to $d=2$ and the number of communities $K=4$.
Figure \ref{fig:sbm} (a) shows the performance of the estimators with an increasing number of nodes when $\beta$ is set to $(1,2)$. As noted in \cite{mcfowland2021estimating} for the case of SBM, the estimator without homophily correction remains biased even with increasing $n$. We notice that in smaller sample sizes, the bias-corrected estimator improves substantially upon the non-biased estimator.

\begin{figure}[!htbp]
\caption{Comparison of estimates of the peer influence parameter $\rho$}
 \label{fig:sbm}
\begin{subfigure}{0.5 \textwidth}
     \centering
\includegraphics[width= \textwidth]{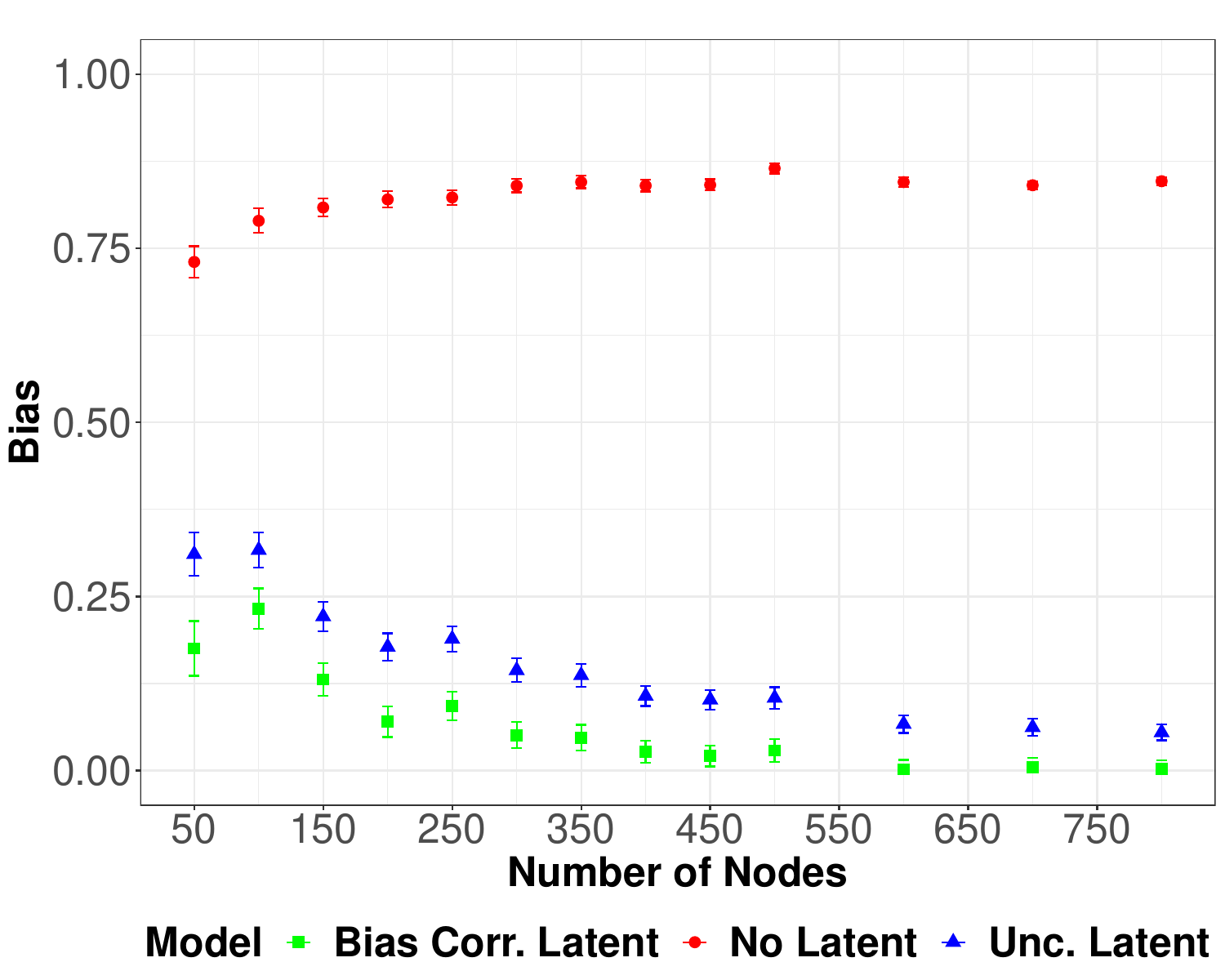}
\subcaption{Increasing number of nodes}
\end{subfigure}%
\begin{subfigure}{0.5 \textwidth}
     \centering
\includegraphics[width= \textwidth]{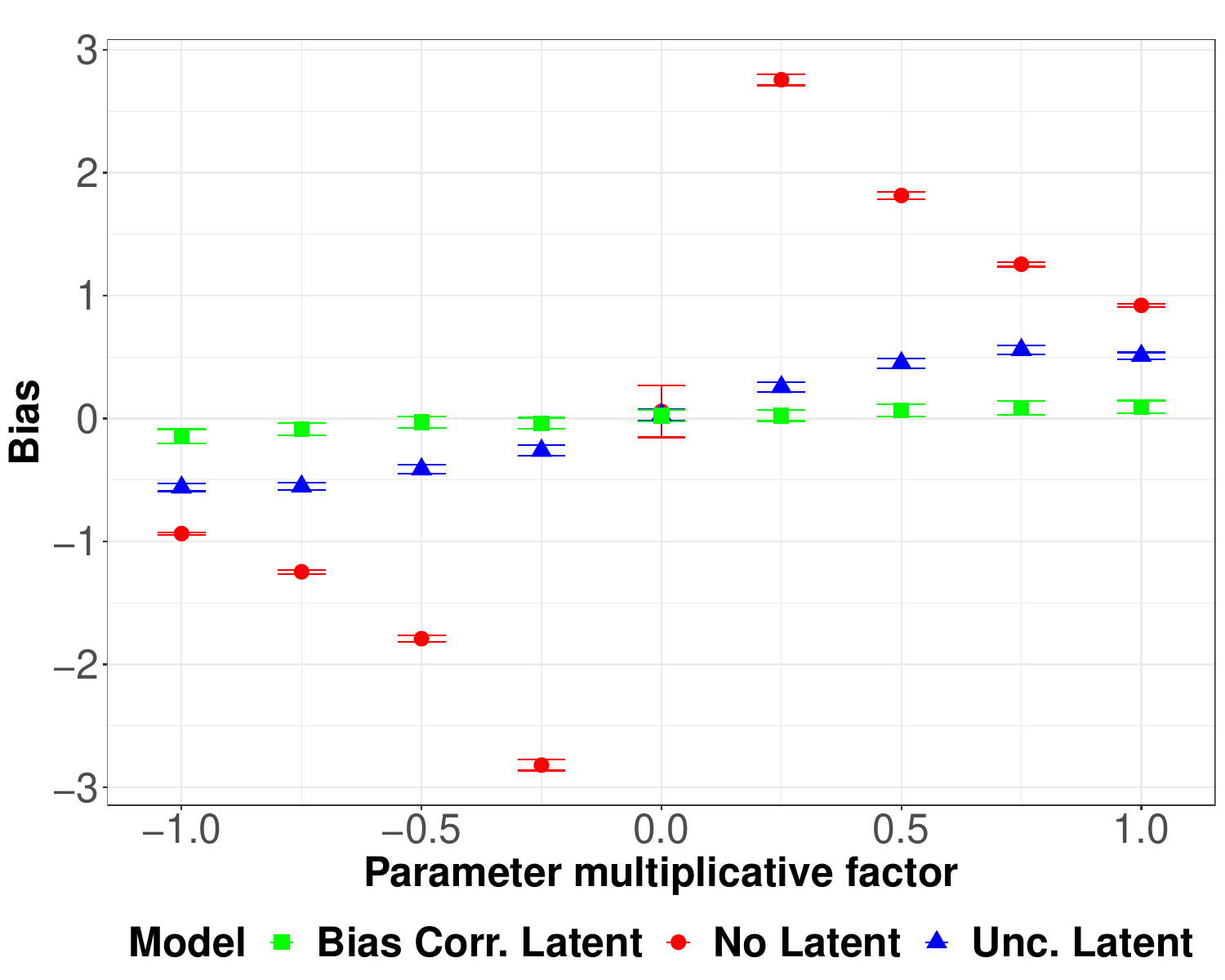}
\subcaption{Model with no latent factors has both positive and negative bias}
\end{subfigure}
 \end{figure}
\subsection{Negative and positive bias}
Next, we design a simulation set-up to test the ability of the bias correction procedure to reduce bias in both positive and negative directions. We set $J$ and $\beta$ as in simulation B but keep $diag(\Theta)$ as 1 for all nodes for simplicity (therefore making the model SBM). However, the parameter associated with $U$ in the previous time period is changed to $\beta'=m\beta$, where $m$ varies from $-1$ to $1$ in increments of $0.25$. The negative $m$'s will create negative bias in the uncorrected estimator. The highest magnitude of the bias is seen at $-0.25$ and $0.25$. In Figure \ref{fig:sbm}(b), we see that the bias-corrected estimator succeeds in correcting the bias of the estimator both when the estimator with no latent factor is biased with negative and positive bias.

\subsection{Comparison of bias correction method}
Next, we compare how the bias correction method performs when we estimate a separate covariance matrix for each of the estimated latent positions using the general formula in Equation \ref{Deltai} and when we estimate only one covariance matrix per community using the formula assuming the network is generated from SBM in Equation  \ref{Deltaq}. Figure \ref{fig:biascorrcomp} shows that both the bias-corrected estimators perform similarly for estimating the peer influence parameter when the network is generated from the degree-corrected block model. Moreover, we do not see much difference in their performances even when the graph is sparse. This is possibly because the formula for the bias corrected estimator involves sum of the covariances of the latent positions over all nodes. Consequently, even though we expect the estimates of the individual covariance matrices for the latent positions in Equation \ref{Deltai} to be not so good when the density of the network is low, the estimate for the overall sum of the covariance matrices is likely comparable.

\begin{figure}[!htbp]
\caption{Comparison of the two bias correction methods using general and SBM-specific covariance matrix estimators}
 \label{fig:biascorrcomp}
\begin{subfigure}{0.5\textwidth}
     \centering
\includegraphics[width= \textwidth]{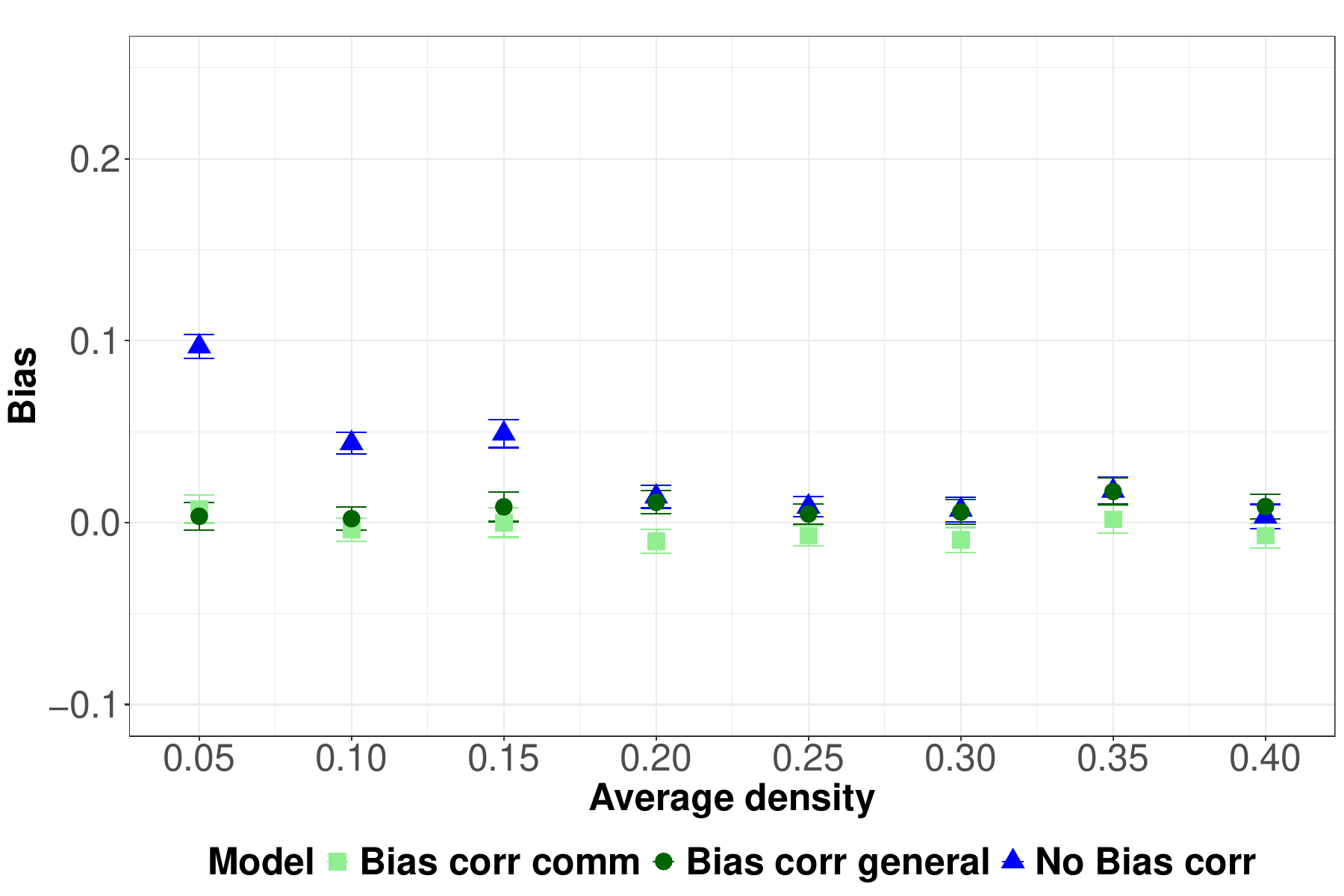} 
\end{subfigure}
 
\end{figure}

\subsection{Probit Model for Binary Reponse Variable} We next perform a simulation study to understand the finite sample properties of the latent homophily variable corrected probit estimator. This estimator was shown to asymptotically provide unbiased estimate of the peer effect parameter in Theorem \ref{theorem3probit}. We consider the same simulation settings as section \ref{section31} (DCSBM graph with increasing nodes), but instead of the linear outcome model, we consider the probit regression outcome model described in Equation \ref{probitnar}. We compare the models with and without latent homophily vectors in terms of bias in estimating the peer influence parameter in \ref{fig:biascorrprobit}. This figure suggests that the bias shrinks rapidly with increasing number of nodes $n$ when the estimated latent factors are included in the probit model, as predicted by Theorem \ref{theorem3probit}. On the other hand the estimate from the model that does not include latent homophily remains biased even when $n$ increases.
\begin{figure}[!htbp]
\caption{Reduction of bias in estimating peer influence parameter by incorporating estimated latent homophily factor when the binary outcome is generated from a Probit Model}
 \label{fig:biascorrprobit}
\begin{subfigure}{0.50\textwidth}
     \centering
\includegraphics[width= \textwidth]{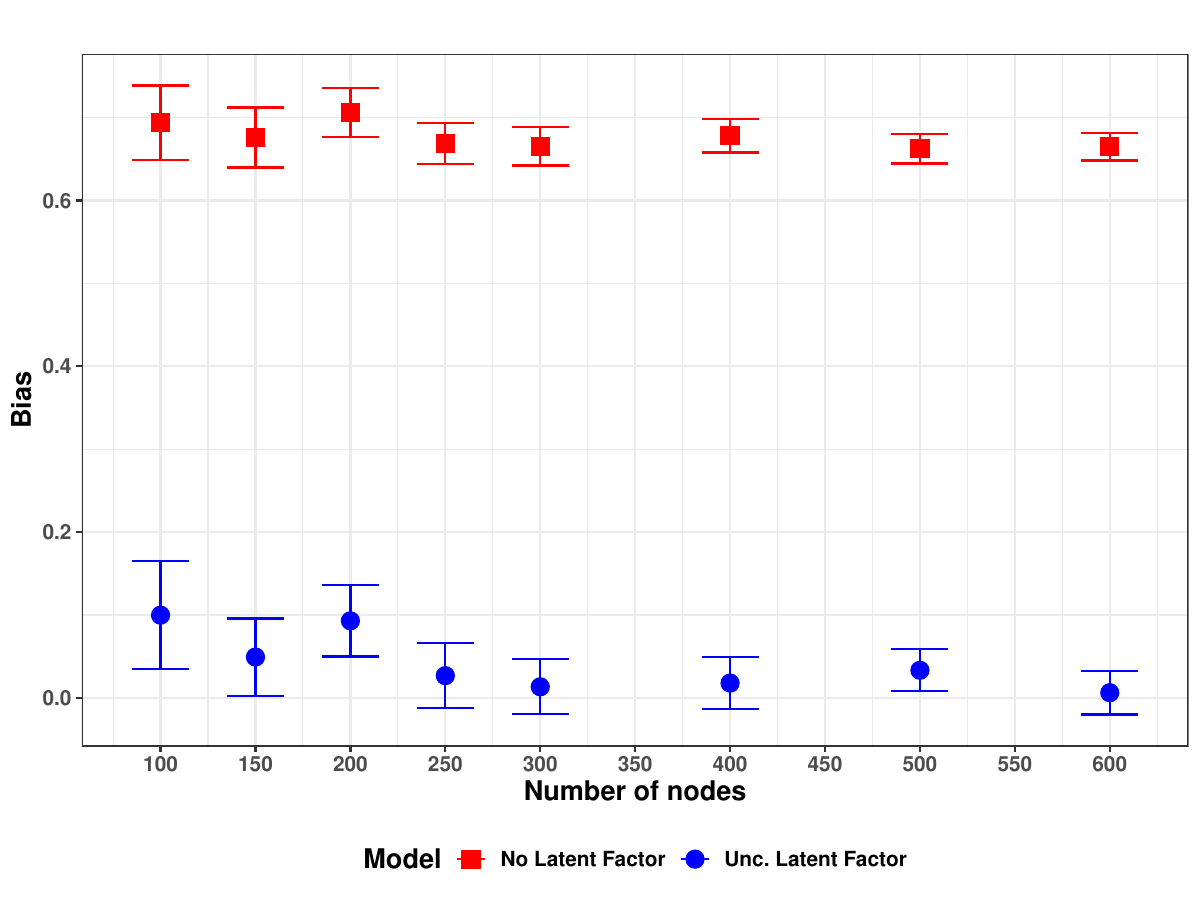} 
\end{subfigure}
\end{figure}

\section{Results}

\subsection{Data}
\label{dataempiricalsection}
The data from the three TCs spans three or four years, depending on the unit. We use the data only on those residents for whom we observe non-missing values for outcome, covariates, and if they sent and received affirmations at least once during their time in the TC. The residents entered the units at different points in time (Panel (a) in Figure \ref{time}) and spent varying amounts of time (Boxplot in Panel (b) in Figure \ref{time}). The median time of stay for residents in male unit 2 and female unit is 149 and 150 days, respectively, while the same for residents in male unit 1 is considerably lower at 124 days. The participants remain in these units for a maximum of six months. The information on entry and exit dates is critical for estimating the causal role model effect. Also, the residents of the TCs interact only within the TC, and there is no interaction across TCs.

The TCs maintained records of residents' socio-demographic characteristics, behavioral aspects, and graduation status. Moreover, the officials implemented a system of mutual feedback among the residents. This took the form of positive affirmations of the prosocial behavior of peers (\cite{campbell2021eigenvector, warren2021resident,warren2021difference,warren2020tightly}). The male units 1 and 2 have about 7,400 and 16,000 affirmations, respectively. The female unit includes a little over 61,000 instances of affirmations over a three-year period. For each of these instances, we observe the anonymized IDs of the sender and the receiver and the message's timestamp. In addition, feedback also involved sending written corrections of behavior that contravenes TC norms. In a separate analysis, we explore peer influence that propagates through the corrections network (see section \ref{correctionssection}). 

The affirmations and corrections record a series of communications between the residents. Related work in the TC literature often treats these communications as friendship networks. Several reasons support this assumption. First, TC residents are likelier to notice positive behaviors if a friend does them.  Second, TC residents are more likely to reward people with an affirmation if they are on good terms. Moreover, empirical evidence also supports this.  There is a lot of reciprocity in the affirmations, transitive relationships are common and appear to exert influence, and there is homophily by graduation status, which tends to suggest that residents with a similar commitment to the program hang out together \citep{warren2020tightly,warren2020building,doogan2017network}. 

The outcome variable of interest is the final graduation status ($S_{i}$). Among the female residents, 79.7\% graduated successfully from the TC. In male unit 1, we observe that 88\% of the residents graduated successfully during the study period; in male unit 2, the corresponding number is 89\%. Given the nature of the peer mentoring program in TCs, we observe that the longer the residents stay in the program, the higher their chances of graduation. Figure \ref{timeTCgrad} shows histograms of time in TC by graduation status for all three units. These graphs suggest that those who stay shorter periods are much less likely to have graduated successfully.

Several reasons can result in residents dropping out of the TC. First, they may not like the program due to various factors. This could include a dislike of the structured nature of the community, such as the requirement to wake up early and perform community tasks. \cite{mccorkel1998treatment} interviewed women who dropped out of a TC and found that they did not regard the peer structure as a legitimate treatment. They preferred a top-down, professionally delivered approach. Second, residents are dismissed immediately from TCs if they break a ``cardinal rule". The cardinal rules are usually no substance use (tobacco is often an exception, as is caffeine), no violence, and no theft. Finally, TC residents proceed through ``phases" in which they earn additional rights by working on their recovery and achieving goals within the TC system. Residents could be terminated for poor progress in the phase system, although we do not have evidence of such dismissal in the TC programs considered in this paper. 

\begin{figure}[!htbp] 
  \caption{Time in TC and Graduation Status}
    \label{timeTCgrad} 

 \begin{minipage}[b]{0.33\linewidth}
       \begin{center}

    \includegraphics[width=\linewidth]{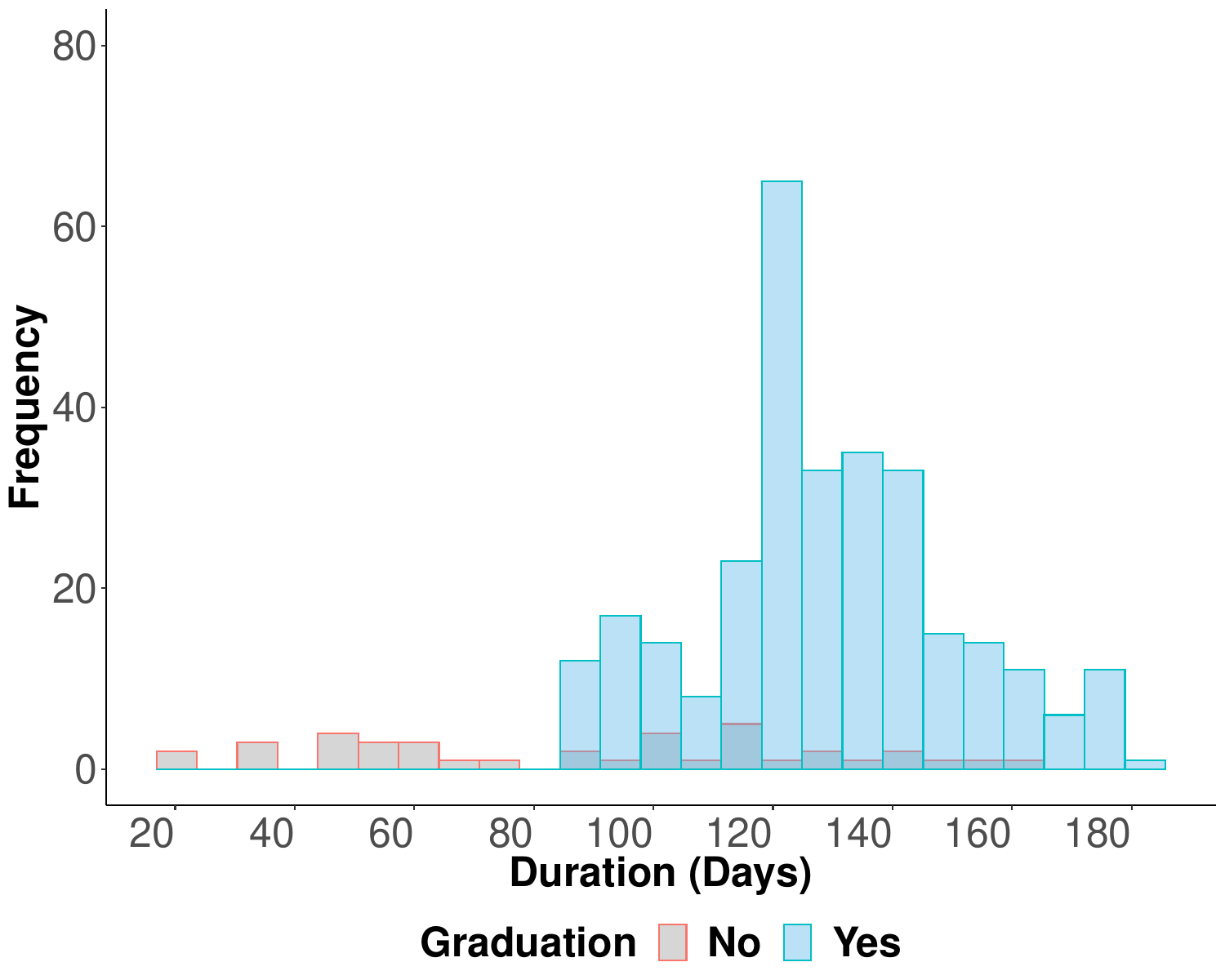}
    \caption*{a. Male Unit 1} 
        \end{center}

  \end{minipage}%%
  \begin{minipage}[b]{0.33\linewidth}
       \begin{center}

    \includegraphics[width=\linewidth]{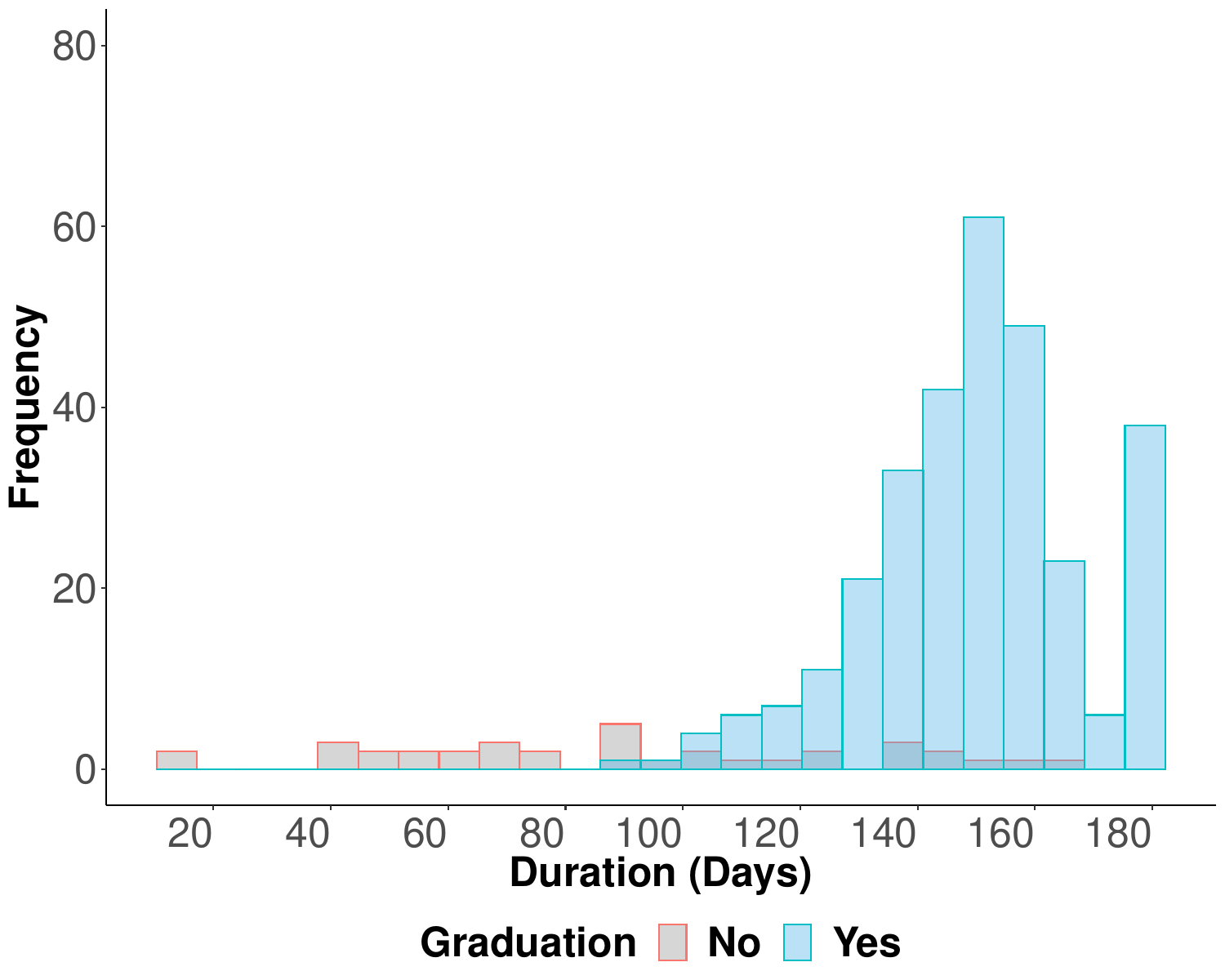}
    \caption*{b. Male Unit 2} 
        \end{center}

  \end{minipage}%%
    \begin{minipage}[b]{0.33\linewidth}
    
         \begin{center}

    \includegraphics[width=\linewidth]{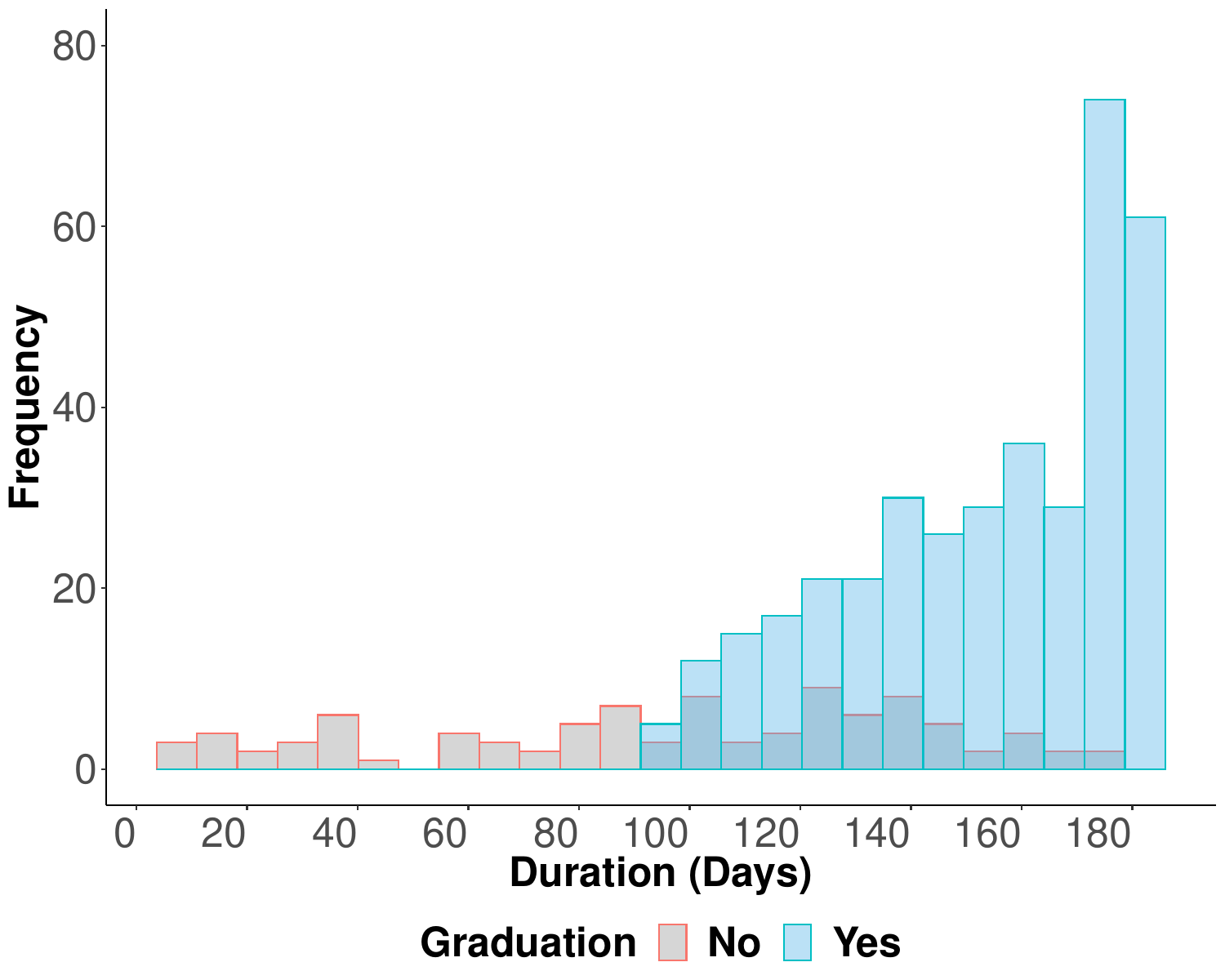}
    \caption*{c. Female Unit} 
          \end{center}

  \end{minipage}
  
  \begin{minipage}{15.5cm}
\footnotesize{
    {Notes: These graphs display histograms for the number of days spent in TC by graduation status in each of the male and female units. Time spent in TC is positively correlated with the propensity to graduate.}}
    \end{minipage}
\end{figure}

The primary explanatory variable is the weighted average of the graduation status of peers ($\frac{\sum_{j} A_{ij}Y_{j,(t-1)}}{\sum_j A_{ij}}$) as observed by resident $i$ just before his/her time of exit. This variable is a function of the affirmations sent and received and the entry and exit dates of $i$ and his/her peers. Using the affirmations data and the time stamps, we can extract all the peers who sent or received affirmations from/to $i$ starting $t_{i}^{\text{entry}}$ to date $t_{i}^{\text{exit}}-1$. In Figure \ref{corrownpeer} (in Appendix), we see that residents who graduated (not graduated) had a higher (lower) fraction of peers who graduated by time $t_{i}^{\text{exit}}$. The correlation is positive, but it is likely confounded with unobserved homophily.

\begin{figure}[!htbp] 
  \caption{Predictors of Graduation Status}
    \label{LSIgrad} 

 \begin{minipage}[b]{0.33\linewidth}
       \begin{center}

    \includegraphics[width=\linewidth]{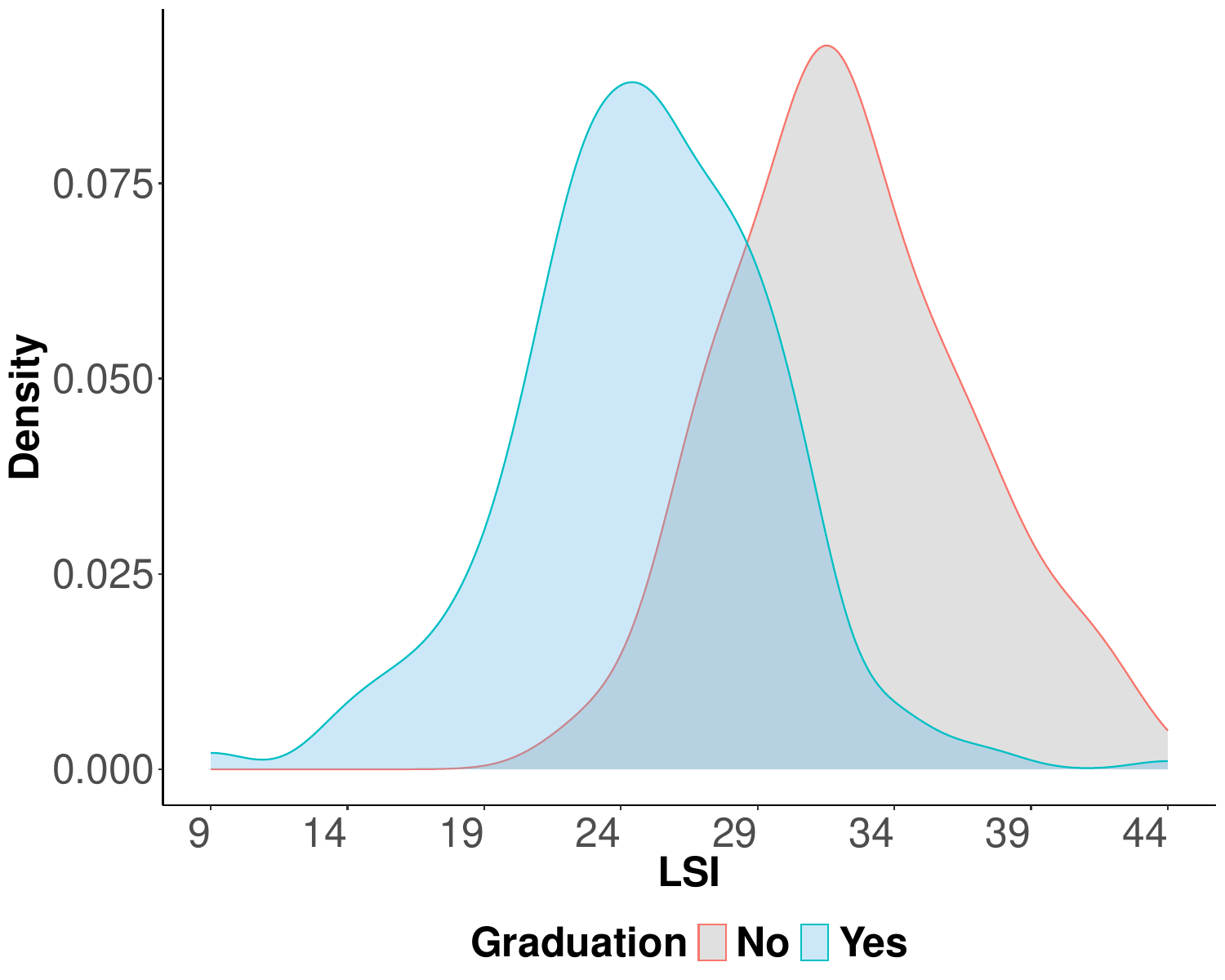}
    \caption*{a. Male Unit 1} 
        \end{center}

  \end{minipage}%%
  \begin{minipage}[b]{0.33\linewidth}
       \begin{center}

    \includegraphics[width=\linewidth]{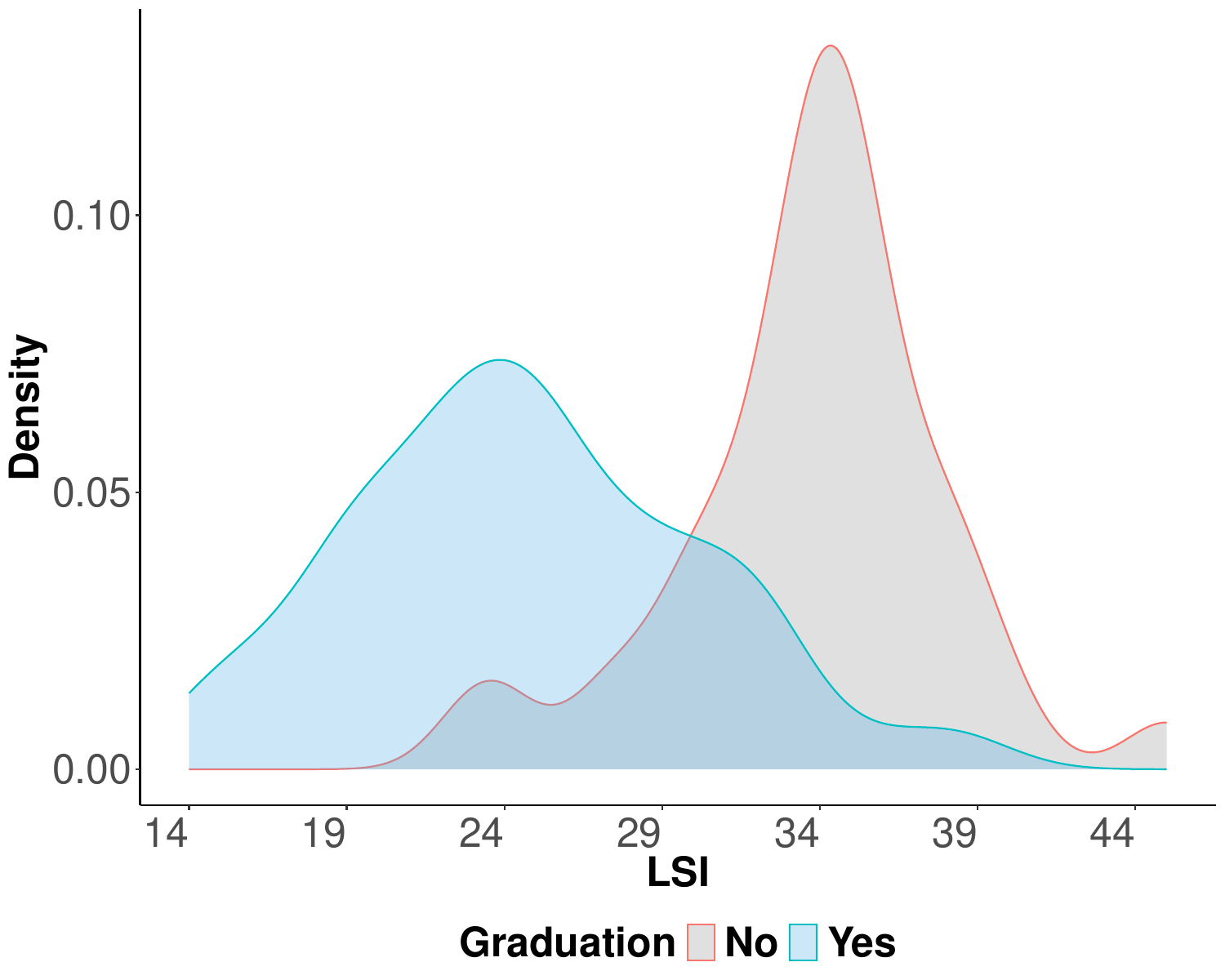}
    \caption*{b. Male Unit 2} 
        \end{center}

  \end{minipage}%%
    \begin{minipage}[b]{0.33\linewidth}
    
         \begin{center}

    \includegraphics[width=\linewidth]{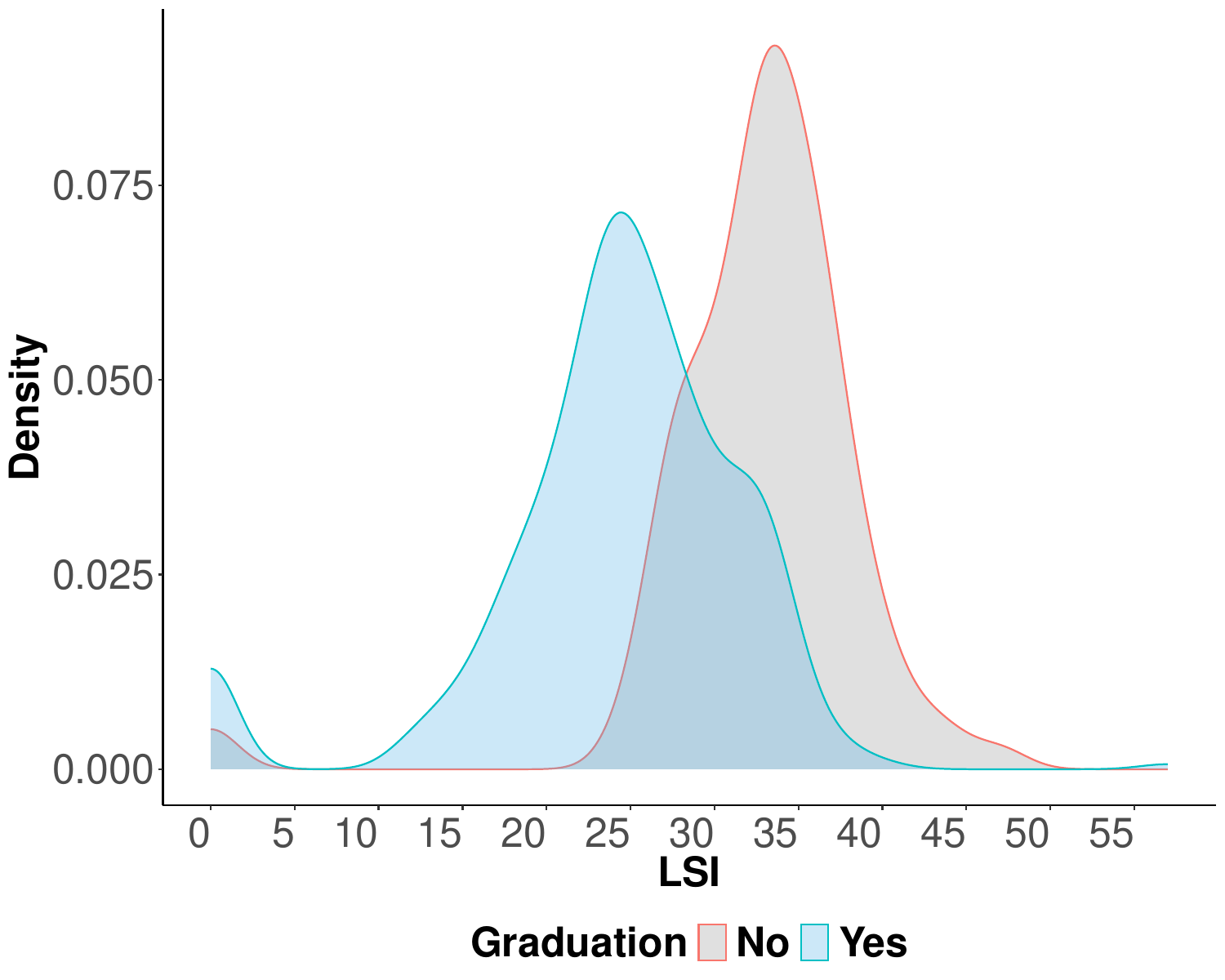}
    \caption*{c. Female Unit} 
          \end{center}

  \end{minipage}
 \begin{minipage}[b]{0.5\linewidth}
       \begin{center}

    \includegraphics[width=\linewidth]{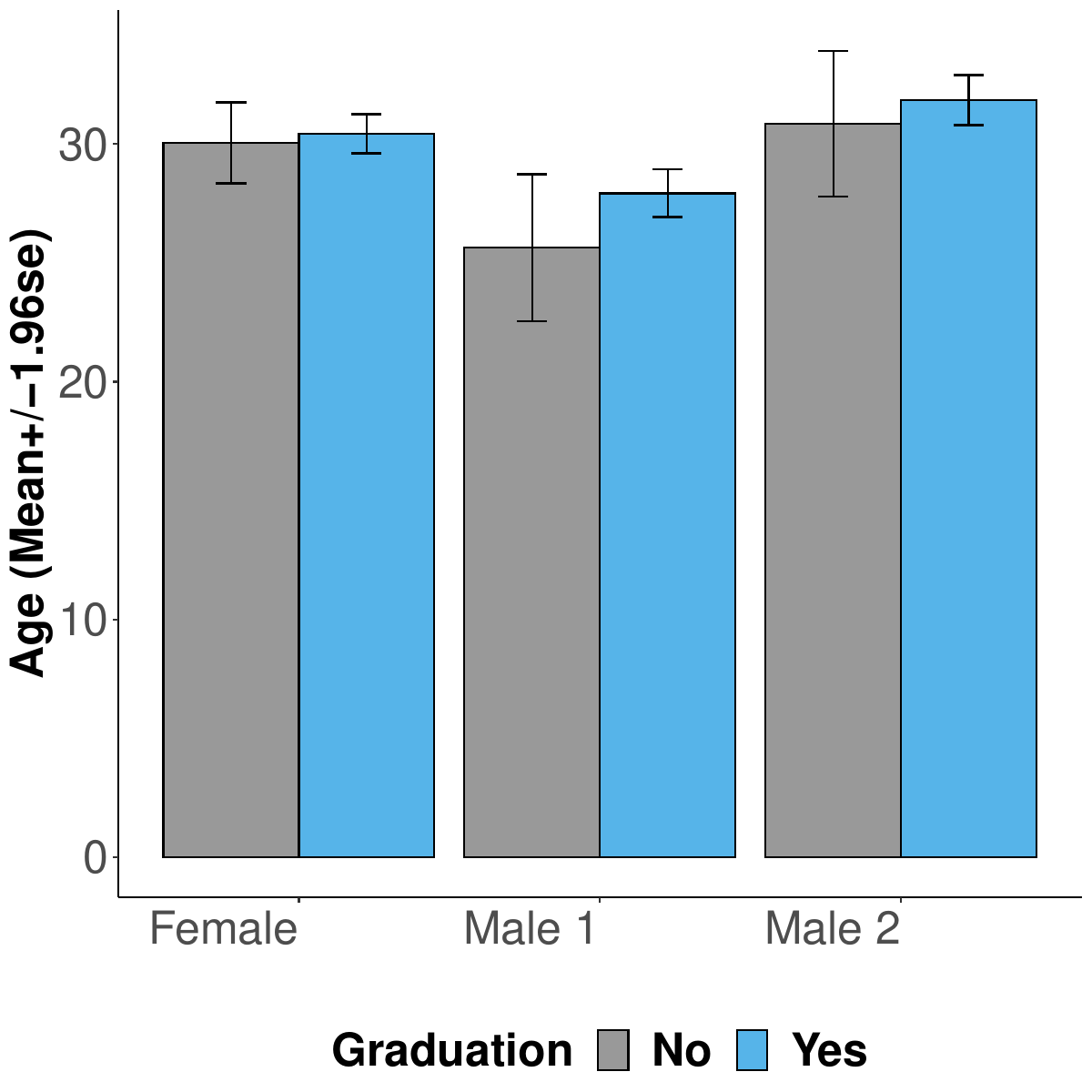}
    \caption*{d. Graduation status and age} 
        \end{center}

  \end{minipage}%%
  \begin{minipage}[b]{0.5\linewidth}
       \begin{center}

    \includegraphics[width=\linewidth]{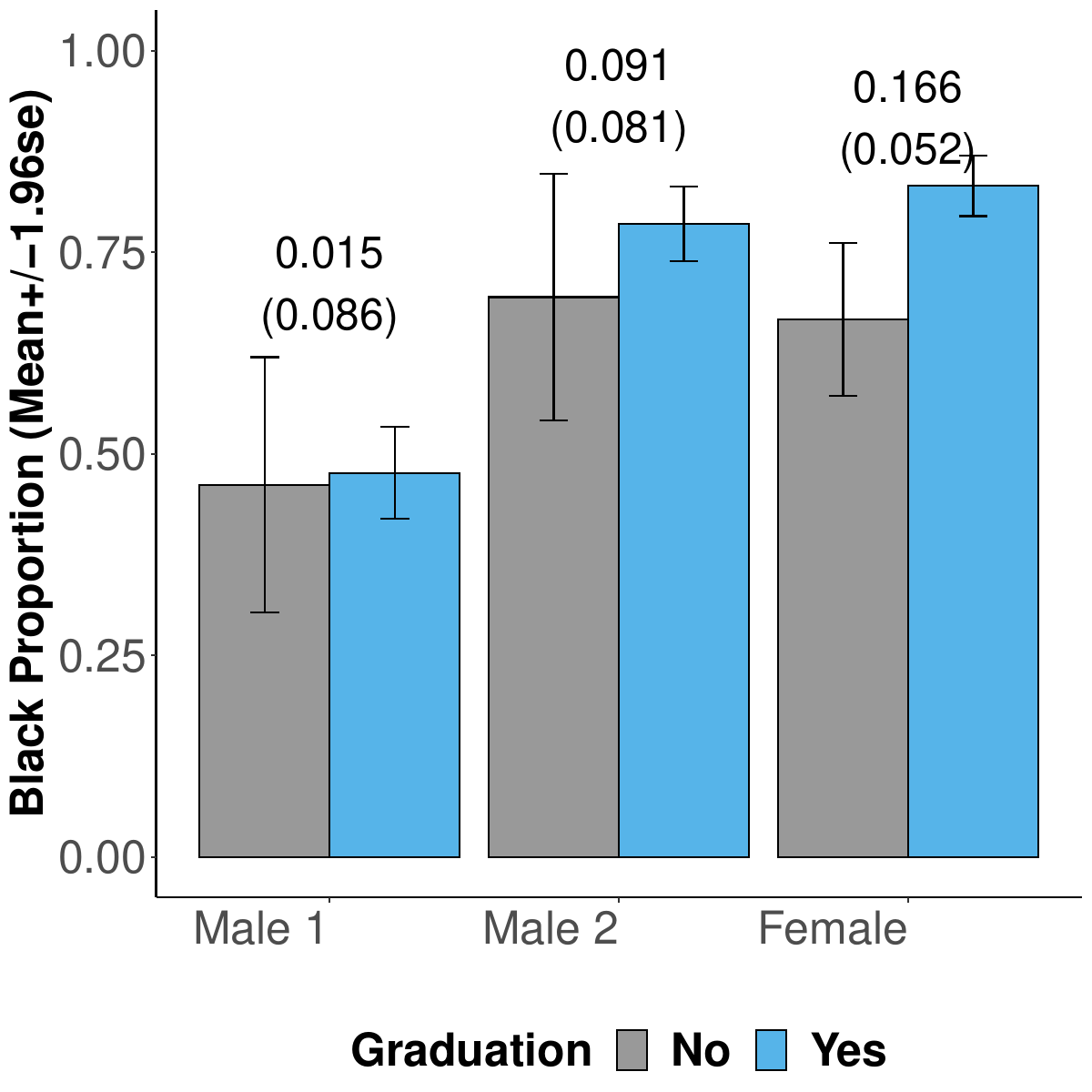}
    \caption*{e. Graduation and Proportion of Blacks} 
        \end{center}

  \end{minipage}
  
  \begin{minipage}{14.0cm}
\footnotesize{
    {Notes: First row (a,b, and c) illustrates the correlation between LSI and $S_{i}$. Panel (d) shows the variation in age by graduation status. Panel (e) shows the proportion of black residents among the residents who were able to graduate successfully and those who did not graduate from the unit. The difference in means for the proportion of blacks and the standard error of the difference is shown at the top of the bar and in parenthesis.}}
    \end{minipage}
\end{figure}

Moreover, we observe a vector of covariates that we control for in our empirical specification. Table \ref{summarytableall3units} (in Appendix) provides summary statistics. Regarding socio-demographic characteristics, the sample is 80\% white in the female unit. The corresponding percentages for the male units 1 and 2 are 47.5\% and 77.6\%, respectively. Age distribution has a mean of 30.4 years for the female unit,  27.6 years for male unit 1, and 31.7 years for the second male unit. The facility also recorded the LSI \citep{andrews1995level} at the time of entry of each resident.  The LSI is a standardized instrument that rates the service needs of residents based on a set of factors, such as substance abuse and family relations, which are known to predict criminal recidivism. The average LSI score is slightly over 25 across all 3 units. 

We show that these covariates likely have some explanatory power for $S_{i}$. Figure \ref{LSIgrad} shows the distribution of LSI by $S_{i}$. We see that for all units, the distribution of LSI is shifted to the right for the non-graduates relative to the graduates. Next, we show the relationship between the age of the residents and their graduation status. We find small differences in the means for age between those who graduated and those who did not graduate from the unit. 

Finally, panel (e) in Figure \ref{LSIgrad} shows the proportion of blacks between the graduates and non-graduates. We see several interesting patterns. First, male unit 2 and female units have a much smaller proportion of black residents than male unit 1. Interestingly, we see large differences in the proportion of blacks among the graduates and non-graduates in these two units. However, this difference is marginal in male unit 1, where there is a balanced proportion of residents by race.

\subsection{Peer Influence with Affirmations Network}

We provide the results for role model effects in all units. Table \ref{rhotable} has 3 panels, one for each unit. Every panel has two columns. The first column presents the parameter estimates for peer graduation without correcting for homophily. In the second column, we provide the estimates using the estimator developed in this paper, which both corrects for latent homophily and additionally adjusts for bias. The sample sizes are 337 in male unit 1, 339 in male unit 2, and 472 in the female unit.

The estimation for our preferred specification (columns (2), (4), and (6)) is done using the following steps. First, we do SVD of the adjacency matrix of affirmations and obtain $\hat{U}$ as described in the methodology section. The selection of the dimension $d$ is done via cross-validation using the method of \cite{li2020network}. The method randomly holds out a portion of node-pairs by selecting each pair independently with a certain probability. Then for a candidate value of $d$, a low-rank matrix completion method is used to obtain an estimate of the expected adjacency matrix. This estimate is used to predict the presence or absence of edge on the held-out node pairs. The method assumes the expected adjacency matrix is a low-rank matrix and was shown to be consistent (in terms of not underselecting $d$) under the RDPG model using a squared-error loss function (Theorem 2 of \cite{li2020network}). We plot the out-of-sample AUC values for $d\in \{1,2,..,20\}$. The AUC rises steeply as we increase $d$ initially, but it stabilizes to a high value around or above 0.9 after a few initial points. We select $d$ corresponding to the point where the AUC becomes stable. Figure \ref{outofsampleauc} (in Appendix) shows the out-of-sample AUC as we increase $d$. Second, we use these estimated $\hat{U}$ as covariates in our outcome model and correct for the bias arising from using $\hat{U}$ instead of true $U$.

We find a decline in the estimates for the peer graduation effect once we correct for homophily and bias. We find the sharpest decline for the female unit ($\approx$ 19\%). The changes for male units are marginal. While the impact of peer graduation status is lower with homophily and bias correction, it continues to increase the residents' graduation substantially. The parameters on peer graduation in columns (2), (4), and (6) are statistically significant in the sense that the 95\% confidence intervals around the point estimates do not include the null effects. We find that a 10\% increase in peers' graduation improves a resident's likelihood of graduating by 4.8-7.7 percentage points (pp).

\begin{table}[!htbp]
\centering
\captionsetup{width=7cm}
\caption{Peer Effects in TC}
\label{rhotable}
\begin{tabular}{p{1.9cm}P{1.55cm}P{1.55cm}P{1.55cm}P{1.55cm}P{1.55cm}P{1.55cm}}\hline\hline
&\multicolumn{6}{c}{Dependent Variable: $S_{i}$}\\\\
&(1)&(2)&(3)&(4)&(5)&(6)\\
Variable  & {OLS} & {Homophily and Bias Adj.} & {OLS} & {Homophily and Bias Adj.} & {OLS} & {Homophily and Bias Adj.}  \\\hline\\
&\multicolumn{2}{c}{\textbf{a. Male Unit 1}}&\multicolumn{2}{c}{\textbf{b. Male Unit 2}}&\multicolumn{2}{c}{\textbf{c. Female Unit}}\\
Peer Grad. & 0.483   & 0.481   & 0.500   & 0.494   & 0.940   & 0.766   \\
                & (0.073) & (0.074) & (0.082) & (0.082) & (0.099) & (0.108) \\
Age             & 0.000   & 0.000   & 0.000   & 0.000   & 0.002   & 0.000   \\
                & (0.002) & (0.002) & (0.002) & (0.002) & (0.002) & (0.002) \\
White           & -0.074  & -0.077  & 0.015   & 0.019   & 0.079   & 0.068   \\
                & (0.030) & (0.031) & (0.035) & (0.035) & (0.039) & (0.038) \\
LSI             & -0.028  & -0.028  & -0.021  & -0.021  & -0.016  & -0.016  \\
                & (0.003) & (0.003) & (0.002) & (0.002) & (0.002) & (0.002) \\
                
\hline\hline
\multicolumn{7}{c}{ \begin{minipage}{13.5 cm}{\footnotesize{Notes: Standard errors are provided in parenthesis. Peer Grad. is constructed using the precise entry and exit dates and the (undirected) affirmations network between the residents as defined in equation \ref{eqrolemodel}. The latent homophily is also estimated from the (undirected) affirmations network.}}
\end{minipage}} \\
\end{tabular}
\end{table}

As for the other covariates, we find some interesting patterns. Being white slightly lowers graduation in male unit 1, where we observe a better balance of white and black residents (refer to Figure \ref{summarytableall3units}). However, in male unit 2 and the female unit, we see a positive (statistically insignificant) coefficient on the white dummy. This specification does not inform us if the impact of peer graduation status varies by race. Importantly, we explore whether the racial identity of role models in peer groups matters for graduation from the unit in the next specification. 

\begin{table}[!htbp]
\centering
\captionsetup{width=13.5cm}
\caption{Peer Effects and Race (Homophily and Bias Adj.)}
\label{rhotable1}
\begin{tabular}{p{6.2cm}P{1.9cm}P{1.9cm}P{1.9cm}}\hline\hline
&\multicolumn{3}{c}{Dependent Variable: $S_{i}$}\\\\
 & (1) & (2) & (3) \\
Variable & Male Unit 1 & Male Unit 2 & Female Unit \\
 \hline

Peer Grad. (White)             & 0.225   & 0.609   & 0.560   \\
                                    & (0.097) & (0.161) & (0.235) \\
Peer Grad. (Non-White)         & 0.194   & 0.113   & 0.202   \\
                                    & (0.101) & (0.123) & (0.235) \\
White                               & -0.140  & 0.140   & 0.064   \\
                                    & (0.071) & (0.080) & (0.092) \\
Peer Grad. (White) $\times$ White     & 0.064   & -0.195  & 0.079   \\
                                    & (0.134) & (0.189) & (0.259) \\
Peer Grad. (Non-White) $\times$ White & 0.093   & -0.116  & -0.069  \\
                                    & (0.141) & (0.137) & (0.254) \\
Age                                 & 0.000   & 0.000   & 0.000   \\
                                    & (0.002) & (0.002) & (0.002) \\
LSI                                 & -0.029  & -0.021  & -0.016  \\
                                    & (0.003) & (0.002) & (0.002) \\
\hline\hline
\multicolumn{4}{c}{ \begin{minipage}{12.5 cm}{\footnotesize{Notes: Standard errors are provided in parenthesis. The baseline category is the mean graduation for a non-white resident. Peer Grad. is constructed separately for white and non-white peers using the precise entry and exit dates and the affirmations network between the residents as defined in equation \ref{eqrolemodel}. The latent homophily is also estimated from the affirmations network.}}
\end{minipage}} \\
\end{tabular}
\end{table}

To capture the heterogeneous effect, these peer variables are constructed separately as follows. For each resident $i$, we extract the graduation status of white peers who left the unit before $i$. Using this subset of peers, which is specific to $i$, we calculate the weighted average of graduation status for white peers $(\frac{\sum_{j_{\text{white}}} A_{ij}Y_{j,(t-1)}}{\sum_{j_{\text{white}}} A_{ij}})$. Similarly, we repeat this exercise for non-white peers and construct $ (\frac{\sum_{j_{\text{non-white}}} A_{ij}Y_{j,(t-1)}}{\sum_{j_{\text{non-white}}} A_{ij}}) $. We use these peer variables on the RHS and interact with the white dummy. Table \ref{rhotable1} presents the estimates. The intercept of the regression (not shown) provides us with the mean graduation status for the baseline group (non-white residents), while the coefficients for Peer Grad (white) and Peer Grad (Non-white) provide effects of white peers' graduation and non-white peers' graduation on non-white residents. Across all 3 units, we find that the peer graduation of both white and non-white residents impacts the graduation status positively. However, the standard errors for the coefficient on peer graduation for non-white residents are very high, especially in male unit 2 and female unit due to small samples of non-white residents. In the male unit 1, where there is a better balance of white and non-white residents, we see the effects are almost identical. We do not have enough power to distinguish a differential impact of peer graduation of white on white and peer graduation of non-white on white residents. 

In addition to the undirected network, we provide estimates for the sender and receiver affirmations network separately in Table \ref{rhotablerobustnessdirected} in the robustness checks. 

\subsection{Second Definition of Role Model}
In this section, we use the second definition of role models described in the methodology (Section \ref{rolemodeleffect}). According to the second definition, a peer group only consists of those peers who have exited the unit before resident $i's$ exit date. Consequently, resident $i$ observes the final graduation status of this subset of peers. We report the estimates in Table \ref{rhotableseconddef}. The direction of the role model effect continues to be positive and large compared to their standard errors, as seen in Table \ref{rhotable}. However, the magnitude of the coefficients drops substantially. This could be because this definition makes the set of peers much smaller than the first definition. Nevertheless, we see a decline in the coefficient on peer graduation status when we control for homophily and adjust for the bias. In line with the observation in Table \ref{rhotable}, the drop is large in the female unit ($>$ 50\%) relative to the two male units. The large change in the peer influence parameter in the female unit with latent homophily adjustment and measurement error correction underscores the importance of adjusting for homophily when estimating peer influence in observational data.

\begin{table}[!htbp]
\centering
\captionsetup{width=12cm}
\caption{Peer Effects in TC (Second Def. of Role Model)}
\label{rhotableseconddef}
\begin{tabular}{p{1.8cm}P{1.5cm}P{1.5cm}P{1.5cm}P{1.5cm}P{1.5cm}P{1.5cm}}\hline\hline
&\multicolumn{6}{c}{Dependent Variable: $S_{i}$}\\\\
&(1)&(2)&(3)&(4)&(5)&(6)\\
Variable  & {OLS} & {Homophily and Bias Adj.} & {OLS} & {Homophily and Bias Adj.} & {OLS} & {Homophily and Bias Adj.}  \\\hline\\
&\multicolumn{2}{c}{\textbf{a. Male Unit 1}}&\multicolumn{2}{c}{\textbf{b. Male Unit 2}}&\multicolumn{2}{c}{\textbf{c. Female Unit}}\\
Peer Grad. & 0.288   & 0.286   & 0.179   & 0.172   & 0.437   & 0.212   \\
                & (0.073) & (0.076) & (0.090) & (0.091) & (0.104) & (0.107) \\
Age             & 0.000   & 0.000   & 0.000   & 0.001   & 0.001   & -0.001  \\
                & (0.002) & (0.002) & (0.002) & (0.002) & (0.002) & (0.002) \\
White           & -0.045  & -0.049  & 0.016   & 0.019   & 0.106   & 0.083   \\
                & (0.031) & (0.032) & (0.037) & (0.037) & (0.042) & (0.040) \\
LSI             & -0.028  & -0.029  & -0.023  & -0.023  & -0.019  & -0.017  \\
                & (0.003) & (0.003) & (0.002) & (0.002) & (0.002) & (0.002) \\
\hline\hline
\multicolumn{7}{c}{ \begin{minipage}{13.5 cm}{\footnotesize{Notes: Standard errors are provided in parenthesis. Peer Grad. is constructed using the precise entry and exit dates and the affirmations network between the residents as defined in equation \ref{eqrolemodel1}. The latent homophily is also estimated from the affirmations network.}}
\end{minipage}} 
\end{tabular}
\end{table}

\subsection{Corrections Network}
\label{correctionssection}
So far, we have analyzed peer effects using the affirmations network. Next, we use the corrections network to estimate peer influence in TCs. A primary difference between the affirmations and the corrections network is that we see residents send each other many more corrections than affirmations. Consequently, our analysis sample is larger when we use the corrections network than the affirmations network. 

Although the corrections constitute negative feedback, the objective is to benefit the recipient by teaching them about the pro-social norms. Some examples of corrections from the TC data are unprosocial language, leaving silverware on a tray, disrespectful behavior toward another resident at the lunch table, and lying and misrepresentation.

The sample sizes are 774 in male unit 1, 391 in male unit 2, and 1046 in the female unit. However, the impact of peers who sent affirmations can differ from those who sent corrections. Table \ref{rhotablecorrect} shows that peer graduation positively impacts the likelihood of graduation from the unit. We find that LSI negatively impacts the propensity to graduate from the unit, which is similar to the coefficients on LSI in Table \ref{rhotable}. Lastly, the age and white dummy coefficients are not statistically significantly different from 0 for most specifications. We examine the network as directed sender's and receiver's network in Appendix Table \ref{rhotablerobustnessdirectedcorrectionsnetwork}. The results from the directed corrections network suggest that both the sender's and receiver's network outcomes positively and significantly impact the graduation status of the residents, but the magnitude of the impact is larger if the residents received corrections from a role model than if they sent corrections to that role model. We also examine heterogeneity by race, which is reported in Table \ref{rhotablecorrectionshte} in the Appendix.

\begin{table}[!htbp]
\centering
\captionsetup{width=12cm}
\caption{Peer Effects in TC (Homophily in Corrections Network)}
\label{rhotablecorrect}
\begin{tabular}{p{1.8 cm}P{1.5cm}P{1.5cm}P{1.5cm}P{1.5cm}P{1.5cm}P{1.5cm}}\hline\hline
&\multicolumn{6}{c}{Dependent Variable: $(S_{i})$}\\\\
&(1)&(2)&(3)&(4)&(5)&(6)\\
Variable  & {OLS} & {Homophily and Bias Adj.} & {OLS} & {Homophily and Bias Adj.} & {OLS} & {Homophily and Bias Adj.}  \\\hline\\
&\multicolumn{2}{c}{\textbf{a. Male Unit 1}}&\multicolumn{2}{c}{\textbf{b. Male Unit 2}}&\multicolumn{2}{c}{\textbf{c. Female Unit}}\\
Peer Grad. & 0.359   & 0.339   & 0.723   & 0.722   & 0.544   & 0.523   \\
                & (0.054) & (0.056) & (0.077) & (0.077) & (0.056) & (0.056) \\
Age             & 0.001   & 0.001   & -0.000  & -0.001  & 0.001   & 0.001   \\
                & (0.001) & (0.001) & (0.001) & (0.001) & (0.001) & (0.001) \\
White           & -0.035  & -0.035  & 0.004   & 0.009   & 0.071   & 0.074   \\
                & (0.020) & (0.021) & (0.032) & (0.032) & (0.022) & (0.022) \\
LSI             & -0.025  & -0.025  & -0.023  & -0.024  & -0.016  & -0.016  \\
                & (0.002) & (0.002) & (0.002) & (0.002) & (0.001) & (0.001) \\
\hline\hline
\multicolumn{7}{c}{ \begin{minipage}{13.5 cm}{\footnotesize{Notes: Standard errors are provided in parenthesis. Peer Grad. is constructed using the precise entry and exit dates and the corrections network between the residents as defined in equation \ref{eqrolemodel}. The latent homophily is also estimated from the corrections network.}}
\end{minipage}}
\end{tabular}
\end{table}

\subsection{Robustness Checks}
In this section, we provide robustness checks for our estimates of peer effects. We begin the robustness checks by changing the computation of the peer graduation variable. We use a weighted network of peers for all our analyses so far. Here, we binarize this network. In other words, we associate a value 1 for an $i,j^{th}$ pair of residents if they have sent and/or received affirmations to/from each other at least once during their stay in the unit. Panel (a) in Table \ref{rhotablerobustness} suggests a positive role model effect whose magnitude is comparable to the weighted network case. We change the model for estimating the latent homophily variables and use a latent space model to extract the (uncorrected) latent factors \cite{hoff2002latent,handcock2007model}. We report these estimates in panel (b) in Table \ref{rhotablerobustness}. Again, we find a positive role model effect. 

\begin{table}[!htbp]
\centering
\captionsetup{width=8.5cm}
\caption{Robustness Checks}
\label{rhotablerobustness}
\begin{tabular}{p{3.5 cm}P{1.95cm}P{1.95cm}P{1.95cm}}\hline\hline
&\multicolumn{3}{c}{Dependent Variable: $S_{i}$}\\\\
 & (1) & (2) & (3) \\
Variable & Male Unit 1 & Male Unit 2 & Female Unit \\
 \hline
                 \multicolumn{4}{c}{a. Binarizing the network (Homophily and Bias Adj.)}\\
 Peer Graduation & 0.521   & 0.553   & 0.777   \\
                & (0.080) & (0.091) & (0.123) \\
\\
            \multicolumn{4}{c}{b. Latent Space Models (Homophily Corrected)}\\
Peer Graduation & 0.483   & 0.492   & 1.104   \\
                & (0.074) & (0.083) & (0.100) \\\\
                \multicolumn{4}{c}{c. Latent Space Models with Node Level Covariates (Homophily Corrected)}\\
Peer Graduation & 0.482   & 0.495   & 0.956   \\
                & (0.073) & (0.082) & (0.099) \\
                \\

\hline\hline
\multicolumn{4}{c}{ \begin{minipage}{11.0 cm}{\footnotesize{Notes: Standard errors are provided in parenthesis. We control for additional covariates in all specifications, i.e., age, white dummy, and LSI. Peer Grad. is constructed using the precise entry and exit dates and the affirmations network between the residents as defined in equation \ref{eqrolemodel}. The latent homophily is also estimated from the affirmations network.}}
\end{minipage}} \\
\end{tabular}
\end{table}

In panel (c), we include node-level covariates with multiplicative latent factors in the latent space model. Figure \ref{latentspacewithcovariates} in Appendix  shows histograms of relevant network properties from networks simulated from the fitted posterior distribution to the data. In each case, the red vertical line denotes the observed property. We observe that the model can replicate heterogeneity (differences) in row and column means quite well, possibly due to the presence of node-level predictors. Note since this is an undirected network, the dyadic dependencies are not relevant and are mechanically set to 1. The model can also replicate the observed cycle and transitive dependencies, manifestations of homophily due to the multiplicative latent variables. This suggests the model fits the data quite well, and the multiplicative latent vectors contain information on homophily. The peer effect parameter using the additive and multiplicative effects model with covariates in panel (c) is also comparable in sign and magnitude to the peer effect seen in the main analysis.

As described in the main analysis, we use an undirected network of affirmations for both constructing the peer graduation variable and the singular value decomposition to estimate the latent homophily vectors. However, here, we separately estimate peer effects from the directed sender and receiver affirmations network. Precisely, for the estimation in panel (a) in Table \ref{rhotablerobustnessdirected} we consider only those peers in affirmations network who were sent affirmations from node $i$. On the contrary, in panel (b), results are shown for the network peers from whom node $i$ received affirmations during their stay. Comparing the estimate for $\rho$ in the two panels suggest that for the affirmations network receiving an affirmation from a role model who has successfully graduated during $i's$ stay has a larger effect on likelihood of $i's$ graduation than sending an affirmation to a successful role model.

\begin{table}[!h]
\centering
\captionsetup{width=8.5cm}
\caption{Directed Sender and Receiver Networks}
\label{rhotablerobustnessdirected}
\begin{tabular}{p{3.5 cm}P{1.95cm}P{1.95cm}P{1.95cm}}\hline\hline
&\multicolumn{3}{c}{Dependent Variable: $S_{i}$}\\\\
 & (1) & (2) & (3) \\
Variable & Male Unit 1 & Male Unit 2 & Female Unit \\
 \hline
 \multicolumn{4}{c}{a. Directed Sender Affirmations Network}\\
Peer Graduation & 0.295   & 0.377   & 0.711   \\
                & (0.055) & (0.072) & (0.092) \\
Age             & 0.000   & -0.000  & 0.001   \\
                & (0.002) & (0.002) & (0.002) \\
White           & -0.072  & 0.014   & 0.077   \\
                & (0.031) & (0.036) & (0.039) \\
LSI             & -0.028  & -0.021  & -0.016  \\
                & (0.003) & (0.002) & (0.002) \\\\
 \multicolumn{4}{c}{b. Directed Receiver Affirmations Network}\\
 Peer Graduation & 0.417   & 0.458   & 0.824   \\
                & (0.072) & (0.076) & (0.104) \\
Age             & 0.001   & -0.000  & 0.000   \\
                & (0.002) & (0.002) & (0.002) \\
White           & -0.074  & 0.025   & 0.078   \\
                & (0.031) & (0.035) & (0.039) \\
LSI             & -0.028  & -0.021  & -0.016  \\
                & (0.003) & (0.002) & (0.002) \\
\hline\hline
\multicolumn{4}{c}{ \begin{minipage}{11.0 cm}{\footnotesize{Notes: Standard errors are provided in parenthesis. Peer Grad. is constructed using the precise entry and exit dates and the affirmations network between the residents as defined in equation \ref{eqrolemodel}. The latent homophily is also estimated from the affirmations network. For panel (a), affirmations only consist of directed sender affirmations from node $i$ to peers, and panel (b) consists of directed receiver affirmations.}}
\end{minipage}} \\
\end{tabular}
\end{table}

\begin{table}[!h]
\centering
\captionsetup{width=13.5cm}
\caption{Probit Model for Binary Response}
\label{rhotable1probit}
\begin{tabular}{p{3.2cm}P{1.9cm}P{1.9cm}P{1.9cm}}\hline\hline
&\multicolumn{3}{c}{Dependent Variable: $S_{i}$}\\\\
 & (1) & (2) & (3) \\
Variable & Male Unit 1 & Male Unit 2 & Female Unit \\
 \hline
 Peer Graduation & 0.457   & 0.307   & 0.568   \\
                & (0.074) & (0.062) & (0.097) \\
Age             & -0.000  & -0.000  & 0.001   \\
                & (0.001) & (0.001) & (0.002) \\
LSI             & -0.023  & -0.017  & -0.017  \\
                & (0.003) & (0.002) & (0.002) \\
White           & -0.028  & 0.020   & 0.053   \\
                & (0.025) & (0.027) & (0.032) \\
\hline\hline
\multicolumn{4}{c}{ \begin{minipage}{9.5 cm}{\footnotesize{Notes: Standard errors are provided in parenthesis. The latent homophily is estimated from the affirmations network.}}
\end{minipage}} \\
\end{tabular}
\end{table}

\begin{figure}[!h] 
  \caption{Average Partial Effect Probit Model}
    \label{AMEgraphs} 

 \begin{minipage}[b]{0.33\linewidth}
       \begin{center}

    \includegraphics[width=\linewidth]{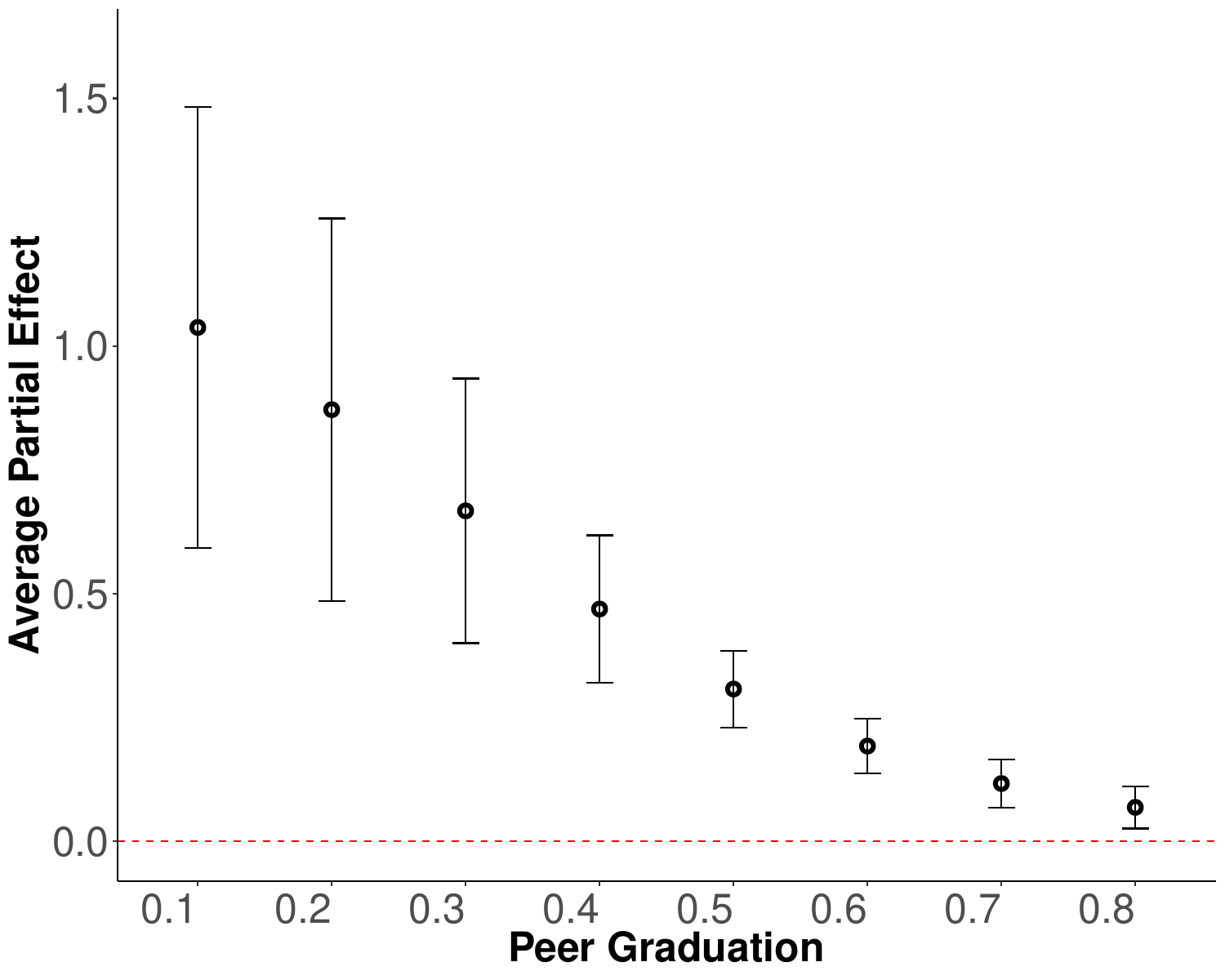}
    \caption*{a. Male Unit 1} 
        \end{center}

  \end{minipage}%%
  \begin{minipage}[b]{0.33\linewidth}
       \begin{center}

    \includegraphics[width=\linewidth]{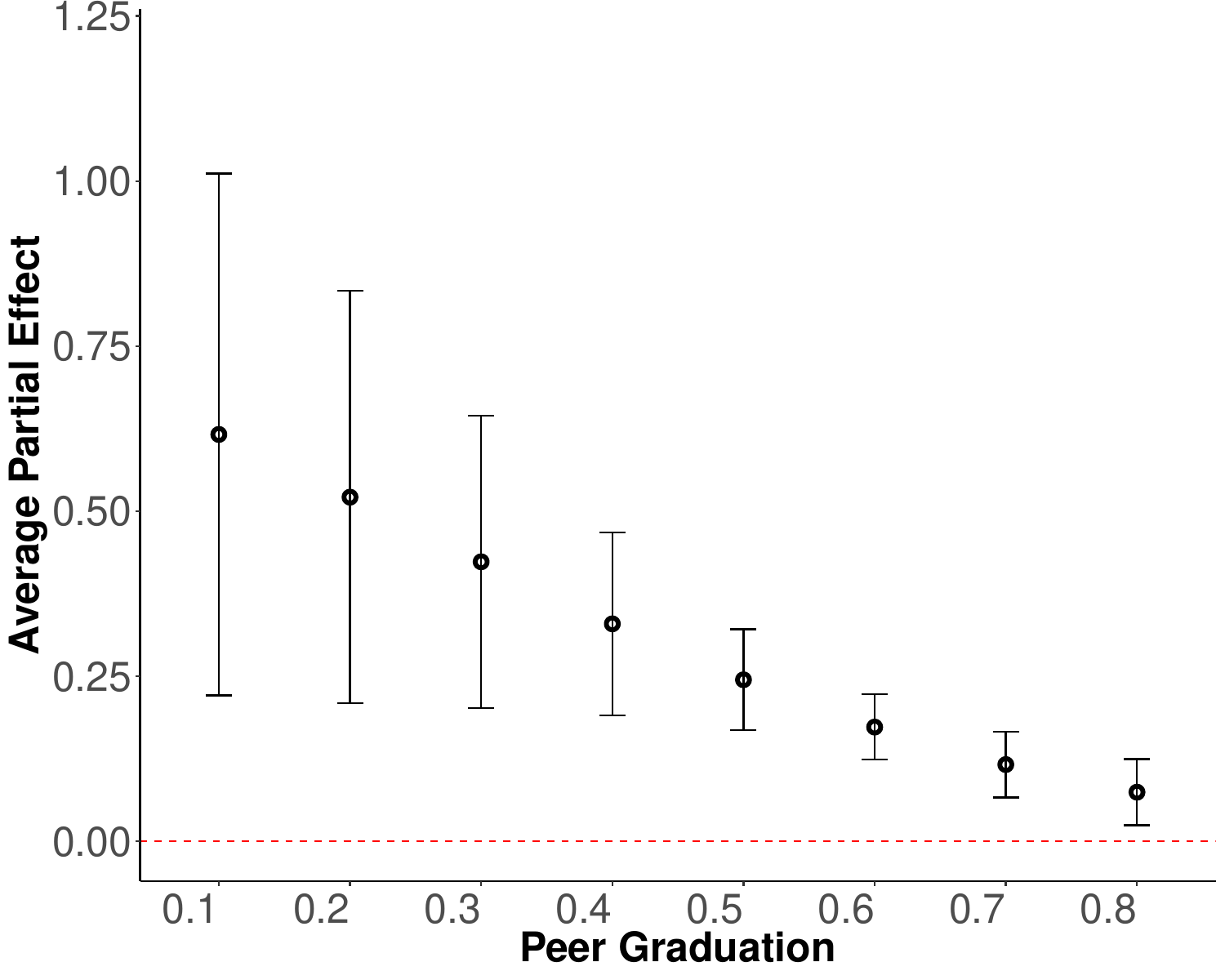}
    \caption*{b. Male Unit 2} 
        \end{center}

  \end{minipage}%%
    \begin{minipage}[b]{0.33\linewidth}
    
         \begin{center}

    \includegraphics[width=\linewidth]{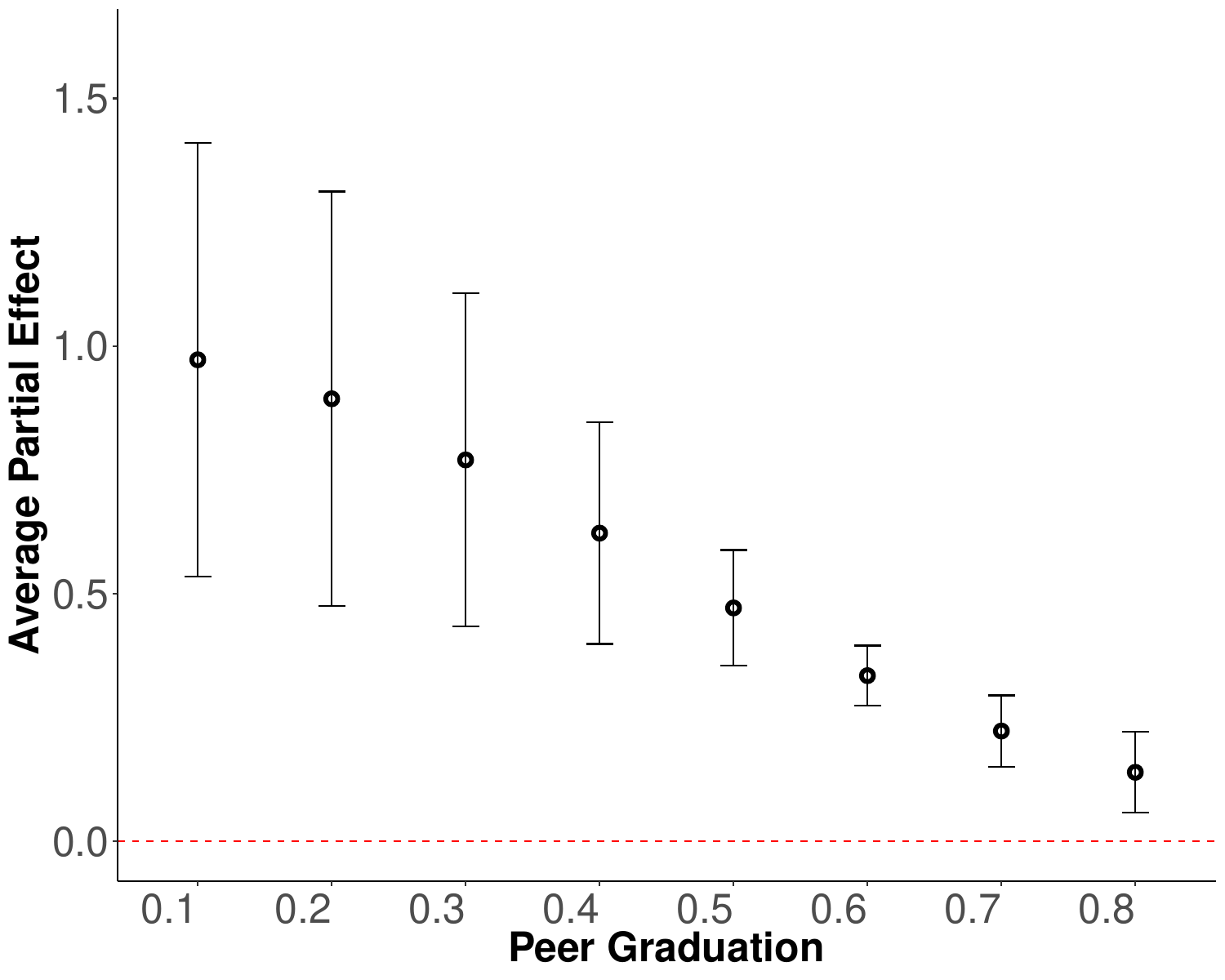}
    \caption*{c. Female Unit} 
          \end{center}

  \end{minipage}
  
  \begin{minipage}{15.5cm}
\footnotesize{
    {Notes: These graphs display the average partial effects (APE) for the binary response with respect to peer graduation variable. In addition to the point estimate, 95\% CI is plotted where the standard errors are computed using the delta method.}}
    \end{minipage}
\end{figure}

\subsection{Binary Response Variable} As noted before, the response variable for the empirical section is binary. The analysis so far has focused on linear peer effects model which for binary responses can be thought of as a linear probability model (LMP). As discussed in section 2.6, while OLS produces consistent and unbiased estimates in LPM, researchers often prefer a probit regression model for binary responses. We extended the result for Theorem \ref{theorem1} for the binary case in Theorem \ref{asympbiasprobit}. Using that result we estimate the model \ref{probitnar} using probit regression and report the results in Table \ref{rhotable1probit}. The table reports the average partial effects for every covariate.
As discussed in section 2.6, the partial effect of a covariate can be interpreted as the change in the expected response (or probability of being 1) due to a 1-unit change in a covariate. Since the probit is a non-linear model, the partial effect would vary for each individual. We compute the partial effect for all residents and then take an average to calculate the average partial effects reported in Table \ref{rhotable1probit}. We find comparable results to those reported in Table \ref{rhotable}. For example, the 3rd column of \ref{rhotable1probit} shows that the probability of successful graduation for an individual on an average increases by 0.568 for one unit increase in peer graduation (or by about 0.0568 for a 0.1 increase in peer graduation). We also provide a plot of average partial effects (APE) with varying levels of the peer graduation variable in Figure \ref{AMEgraphs}. As typical with non-linear models, the APE of peer graduation is higher for lower values of peer graduation. The APE gradually shrinks towards zero as the value of peer graduation increases, indicating a diminished return from marginal increases in peer graduation.

\section{Discussion}

\subsection{Counterfactual Analysis}

\begin{figure}[!htbp] 
  \caption{Toy Example: Direct and Cascading Impact of Intervention}
    \label{toyexample} 

 \begin{minipage}[b]{0.24\textwidth}
       \begin{center}

    \begin{tikzpicture}
\Vertex[color=red,x=2,y=1,size=0.3]{A} \Vertex[x=2,color=green!70!blue,x=1,y=-1,size=0.3]{B} \Vertex[x=2,y=-1,color=red,size=0.3]{C} \Vertex[x=2,y=-2,color=green!70!blue,size=0.3]{D}
\Edge[color=black](A)(C)
\Edge[color=black](B)(D)
\Edge[color=black](A)(D)
\end{tikzpicture}

    \caption*{a.} 
        \end{center}

  \end{minipage}%
  \begin{minipage}[b]{0.24\textwidth}
       \begin{center}

    \begin{tikzpicture}
\Vertex[color=red,x=2,y=1,size=0.3]{A} \Vertex[x=2,color=green!70!blue,x=1,y=-1,size=0.3]{B} \Vertex[x=2,y=-1,color=red,size=0.3]{C} \Vertex[x=2,y=-2,color=green!70!blue,size=0.3]{D} \Vertex[color=green!70!blue,x=2,y=2,style=dashed,size=0.3]{E} 
\Edge[color=black](A)(C)
\Edge[color=black](B)(D)
\Edge[color=black](A)(D)
\Edge[color=black,style={dashed}](A)(E)
\end{tikzpicture}
    \caption*{b. } 
        \end{center}

  \end{minipage}%
    \begin{minipage}[b]{0.24\textwidth}
    
         \begin{center}

    \begin{tikzpicture}
\Vertex[color=green!70!blue,x=2,y=1,style={dashed},size=0.3]{A} \Vertex[x=2,color=green!70!blue,x=1,y=-1,size=0.3]{B} \Vertex[x=2,y=-1,color=red,size=0.3]{C} \Vertex[x=2,y=-2,color=green!70!blue,size=0.3]{D} \Vertex[color=green!70!blue,x=2,y=2,style=dashed,size=0.3]{E} 
\Edge[color=black](A)(C)
\Edge[color=black](B)(D)
\Edge[color=black](A)(D)
\Edge[color=black,style={dashed}](A)(E)
\end{tikzpicture}

    \caption*{c. } 
          \end{center}

  \end{minipage}%
     \begin{minipage}[b]{0.24\textwidth}
    
         \begin{center}

    \begin{tikzpicture}
\Vertex[color=green!70!blue,x=2,y=1,style={dashed},size=0.3]{A} \Vertex[x=2,color=green!70!blue,x=1,y=-1,size=0.3]{B} \Vertex[x=2,y=-1,color=green!70!blue,style={dashed},size=0.3]{C} \Vertex[x=2,y=-2,color=green!70!blue,size=0.3]{D} \Vertex[color=green!70!blue,x=2,y=2,style=dashed,size=0.3]{E} 
\Edge[color=black](A)(C)
\Edge[color=black](B)(D)
\Edge[color=black](A)(D)
\Edge[color=black,style={dashed}](A)(E)
\end{tikzpicture}

    \caption*{d. } 
            \end{center}

  \end{minipage}

  \begin{minipage}{15cm}
\footnotesize{
    {Notes: The failed residents are denoted in \textcolor{red}{red} and successful are shown in \textcolor{green}{green}. (a) displays the existing network. (b), shows the intervention where one \textcolor{red}{red} peer is assigned a \textcolor{green}{green} buddy. In (c), we see that receiving a higher fraction of affirmations from successful peers can help treated residents graduate. Finally, in (d), we see that the treated resident would have a cascading effect on the other \textcolor{red}{red} peer, which helps them to graduate.}}
    \end{minipage}
\end{figure}

Here, we perform counterfactual simulations to predict the results of a few interventions. We start by estimating our preferred model (Columns (2), (4), and (6) in Table \ref{rhotable}). We use the estimated coefficients with data to compute the fitted value for the propensity of graduation. We calculate the graduation probability threshold by computing the misclassification error at each threshold. Figure \ref{misclassification} (see Appendix) shows the threshold selection for each unit. For each unit, we select that threshold that minimizes the misclassification error (indicated in \textcolor{blue}{blue}).

We can identify residents whose true outcome is 0. We can nudge these residents by assigning them a successful buddy. The intervention is that a successful peer (assigned buddy) is asked to send affirmations to these ``at-risk" residents. This intervention will move the predicted likelihood of graduation for some of them above the threshold for graduation. Now, we replace the outcome for this subset as 1 and keep the outcome for all other residents the same as before. With this modified outcome, we can re-estimate the model parameters of peer graduation and predict the outcome for all the residents. This second step captures how the intervention can propagate through the network (Figure \ref{toyexample}). Post this prediction exercise, we compute residents who are still ``at-risk".

Table \ref{countefactuaLSImulations} (appendix) shows the number of residents on whom we perform the intervention in step 1 and the additional residents who are pushed above the threshold due to the cascading effect of the intervention. For male unit 1, we find that we would be treating 24 residents as these have failed to graduate and, after assigning a buddy, cross the graduation threshold. In step 2, these treated residents and others can influence each other, and we find that post-simulation, we have only 2 residents below the graduation threshold instead of 39 failed residents in true data for male unit 1. Similarly, for male unit 2, we can push an additional 15 residents above the threshold and end up with 21 residents below the threshold instead of 36 in true data. Finally, we see that this intervention results in only 25 residents below the threshold in the female unit instead of 96 in true data. 
\begin{table}[!htbp]
\centering
\caption{Counterfactual Scenario with LSI}
\label{countefactuaLSImulations2}
\begin{tabular}{p{7cm}p{1.5cm}p{1.5cm}p{1.5cm}}
  \hline\hline
 & Male Unit 1 & Male Unit 2 & Female Unit \\ 
  \hline\\
  \multicolumn{4}{c}{\textbf{a. Treat residents with LSI $>90^{th}$ percentile}}\\
  Threshold for Graduation & 0.59 & 0.66 & 0.47 \\ 
  Residents who cleared the threshold of graduation due to a buddy & 19.00 & 8.00 & 12.00 \\ 
  Failures (True Data) & 39.00 & 36.00 & 96.00 \\ 
  Residents whose predicted graduation is below the threshold (Post-simulation) & 9.00 & 34.00 & 78.00 \\ \\
    \multicolumn{4}{c}{\textbf{b. Treat residents with LSI $>80^{th}$ percentile}}\\
Threshold for Graduation & 0.59 & 0.66 & 0.47 \\ 
  Residents who cleared the threshold of graduation due to a buddy & 39.00 & 38.00 & 30.00 \\ 
  Failures (True Data) & 39.00 & 36.00 & 96.00 \\ 
  Residents whose predicted graduation is below the threshold (Post-simulation) & 5.00 & 25.00 & 64.00 \\ \\
        \multicolumn{4}{c}{\textbf{c. Treat residents with LSI $>75^{th}$ percentile}}\\
        Threshold for Graduation & 0.59 & 0.66 & 0.47 \\ 
  Residents who cleared the threshold of graduation due to a buddy & 56.00 & 49.00 & 50.00 \\ 
  Failures (True Data) & 39.00 & 36.00 & 96.00 \\ 
  Residents whose predicted graduation is below the threshold (Post-simulation) & 4.00 & 23.00 & 57.00 \\ 
   \hline\hline
   \multicolumn{4}{c}{ \begin{minipage}{13.0 cm}{\footnotesize{Notes: Each panel targets ``at-risk" residents for intervention using a different percentile cut-off on LSI. The estimated parameters for the even-numbered columns (homophily and bias-adjusted) in Table \ref{rhotable} are used for simulating these counterfactual scenarios.}}
\end{minipage}} \\
\end{tabular}
\end{table}

The counterfactual exercise above uses the true graduation status. During the intervention, this variable will not be observed by the researcher. Hence, we would ideally like to use a variable that is highly correlated with $S_{i}$  and is measured pre-intervention. One such variable is LSI (Figure \ref{LSIgrad}). We treat residents at the top $x$ percentile of the LSI distribution and perform the same procedure as described above. Table \ref{countefactuaLSImulations2} shows the estimates of the counterfactual simulations.

We provide results for 3 cut-offs of LSI to identify the residents to be treated using our intervention. Panel (a) in Table \ref{countefactuaLSImulations2} shows the estimated number of treated residents and residents whose predicted probability of graduation is below the cut-off after taking into account the cascading effect once we treat residents whose LSI is above the $90^{th}$ percentile. We observe that using this LSI cut-off results in treating a few individuals. However, even this intervention on a few individuals exceeds the threshold once we consider the cascading effect. We relax the cut-off for LSI to include more residents under the treated group in panels (b) and (c), respectively. LSI is highly correlated with graduation but not perfectly correlated. Therefore, we see the number of residents who could be finally moved above the threshold is lower for all panels in Table \ref{countefactuaLSImulations2} relative to the results in Table \ref{countefactuaLSImulations}. A critical lesson for TC clinicians and policymakers is that one can compute the optimal number of residents to be nudged by comparing the additional cost of assigning a buddy to a resident with the benefit of graduation. As long as the additional benefits exceed the additional cost, one can treat an additional resident and stop when the cost equals the benefit.

\subsection{Simultaneous response peer effects} Another peer influence model has been widely applied in the literature when the response $Y$ is observed only at one-time point \cite{bramoulle2009identification,lin2010identifying,goldsmith2013social}. The structural model has $Y_i$s being simultaneously dependent on each other. In the matrix form, the model is 
\[
Y = \alpha 1_n + \rho LY + Z \gamma + V.
\]
It is typical to assume that the error term is uncorrelated with $L$ and $Z$ for identification, i.e., $E[V_i |Z_i, L] = 0$ \cite{bramoulle2009identification}. This model cannot be estimated by OLS due to the presence of the response $Y$ on both sides, or in other words, due to $LY$ being correlated with $V$. However, \cite{bramoulle2009identification} showed that one can asymptotically identify the peer influence parameter $\rho$ with a Two-stage least squares instrumental variable estimator using $S=[Z, LZ, L^2Z]$ as instruments. The result in \cite{bramoulle2009identification} does not hold if one assumes that there are latent homophily variables that affect both the network formation and the response, which is the setup of the current paper.

We do not analyze this model here. Our methodology was motivated by how we define peer effects in the context of the therapeutic community data (which we call the role model effect). For a person to be a role model to someone else, the variation in terms of time is critical. This allowed us to define the causal effect in terms of a directed acyclic graph, which we would not have if there was a simultaneous dependence between $Y_i$s in the same time period. 

Nevertheless this simultaneous response model is an important model to study when the response is available only at one time point. We conjecture that a similar methodology that we have proposed here, along with instrumental variables, will succeed in the asymptotic identification of peer influence in this case. We keep a study of this model as a future research direction.

\subsection{Takeaway for TC clinicians} These results carry several lessons for TC clinicians. Most obviously, the finding that peer graduation status impacts graduation \cite{de1982therapeutic} confirms the community's importance as a treatment method in the TC. Previous research has found the existence of homophily among TC graduates \cite{warren2020senior}. However, this paper shows that network influence, once we control for homophily, substantially impacts the propensity to graduate from TC.
The social learning emphasis of TC clinical theory \cite{warren2021resident} would imply that social networks play their classic role as conduits of information, which would explain the direct effect of network connections on graduation. Previous research suggests that social network roles influence and constrain residents as they go through TC treatment \citep{campbell2021eigenvector,warren2020tightly,de2000therapeutic}. It is plausible that interaction between stronger and weaker members of the TC allows the former to experience the role of helper to the latter and that such experience is of value in recovery \citep{riessman1965helper}.  Moreover, it is likely that results found in TCs would extend to other mutual aid-based programs, such as recovery housing \citep{jason2022dynamic,mericle2023social} and 12-Step groups \citep{white2008twelve}.

\section*{Acknowledgement} We are grateful to Prof. Cosma Shalizi, Prof. Edward McFowland III, Prof. Joshua Cape, and participants in a session of the CMStatistics 2022 conference for helpful discussions and feedback.

\bibliographystyle{plainnat}
\bibliography{bibliography}

\clearpage

\appendix

\setcounter{table}{0}
\renewcommand{\thetable}{B\arabic{table}}
\setcounter{figure}{0}
\renewcommand{\thefigure}{B\arabic{figure}}

\section{Appendix}
\subsection{Proofs}

\begin{proof}[\textbf{Proof of Proposition} \ref{asympbias}]

We will show that $\rho$ is equivalent to $\frac{1}{n} \sum_i \phi_i$. The proof follows similar arguments as in \cite{sridhar2022estimating}. Recall that in our setup the causal effect $\phi_{ij}$ involves expectations conditioning on $A_{ij}>0$ and fixing either $Y_j^i =1$ or $Y_j^i=0$. Those expectations then  are taken by marginalizing or integrating over other connections (non $j$) of $i$ denoted by $A_{i,(-j)}$, and the observed covariates $Z_i$.
\begin{align*}
\phi_{ij} & = E[S_i | A_{ij} >0, do(Y^{(i)}_j = 1)] - E[S_i | A_{ij} >0, do(Y^{(i)}_j = 0) ]\\
& = E_{A_{i,(-j)},Z_i} \bigg[ E[S_i | A_{ij} >0, do(Y^{(i)}_j = 1),A_{i,(-j)},Z_i]  - E[S_i | A_{ij} >0, do(Y^{(i)}_j = 0), A_{i,(-j)},Z_i ] \bigg].
\end{align*}
Using the notation of \cite{sridhar2022estimating}, we define,
\[
\mu_{ij}(a,y) = E[S_i | 1\{A_{ij} >0\}=a, Y^{(i)}_j = y,A_{i,(-j)},Z_i, U_i],
\]
where $1\{A_{ij} >0\}$ is the indicator function for the event $\{A_{ij} >0\}$. Both $a$ and $y$ can take two values, 1 and 0. We noted from the DAG that conditioning on $U_i, U_j$ closes the backdoor paths and the causal effect is identified. Now we can write the backdoor adjustment formula \cite{pearl2009causality} as
\[
\phi_{ij} =  E_{U_i}[E_{A_{i,(-j)},Z_i}[\mu_{ij}(1,1)-\mu_{ij}(1,0)]].
\]
Then, from our model, we can calculate,
\begin{align*}
    \mu_{ij}(1,1)-\mu_{ij}(1,0) & = \big(\alpha_0 + \gamma^T Z_i + \frac{\rho A_{ij}}{\sum_j A_{ij}} \\
    & +  \rho \frac{\sum_{l\neq j} A_{li} Y_j^{(i)}}{\sum_l A_{li}} +\beta^T U_i \big)  - \big(\alpha_0 + \gamma^T Z_i +  \rho \frac{\sum_{l\neq j} A_{li} Y_j^{(i)}}{\sum_l A_{li}} +\beta^T U_i\big) \\
& = \frac{\rho A_{ij}}{\sum_j A_{ij}}
\end{align*}
Therefore, for all $i$,
\[
\phi_i =  \sum_j \phi_{ij} =\sum_j E_{U_i}[E_{A_{i,(-j)},Z_i}[ \frac{\rho A_{ij}}{\sum_j A_{ij}}]]=\rho.
\]

\end{proof}

\begin{proof}[\textbf{Proof of Theorem} \ref{asympbias}]
    
Note we can rewrite the model in \eqref{narmodel} as follows:
\begin{align*}  
Y_{i,t}&=\alpha_0 + \alpha_1 Y_{i,(t-1)} +  \beta_1^T \hat{U}_{i} +  \rho  \frac{\sum_{j} A_{ij}Y_{j,(t-1)}}{\sum_j A_{ij}} + \\&\gamma^T Z_i + V_{i,t} + \beta^T U_i - \beta_1^T \hat{U}_i.
%\label{regmodel}
\end{align*}
Then when we use $\hat{U}_i$ in place of $U_i$ in the model, the error term in the regression model now consists of $V_{i,t} + \beta^T U_i - \beta_1^T \hat{U}_i$. From Lemmas 1 and 2 in \cite{mcfowland2021estimating} the covariance between $\frac{\sum_j A_{ij}Y_{j,t-1}}{\sum_j A_{ij}}$ and $\beta^T U_i - \beta_1^T \hat{U}_i$ conditioning on $A, \hat{U}_i,\hat{U}_j,Y_{i,(t-1)}$  is given by 
%\footnotesize
\begin{align}
& Cov\left( \frac{\sum_j A_{ij}Y_{j,t-1}}{\sum_j A_{ij}}, \beta^T U_i - \beta_1^T \hat{U}_i \Bigg |A, \hat{U}_i, \hat{U}_j, Y_{i,t-1}\right) \nonumber \\
& = \frac{1}{\sum_j A_{ij}}\sum_j A_{ij}\beta^T \big(c_1 Cov (U_i, U_j|A) + c_2 Var(U_i|A)  + \frac{c_3}{d_i} \sum_{l \neq(i,j)} Cov(U_i, U_l|A)\big)\beta,
\label{covar}
\end{align}
%\normaLSIze
where $d_i = \sum_j A_{ij}$ and $c_1, c_2, c_3$ are suitable constants. We omitted the conditioning on $\hat{U}_i$ and $\hat{U}_j$ since it is equivalent to conditioning on $A$ as they are deterministic functions of each other. 
Now define the event $G_n$ as an indicator (random) variable for the following event
\[
\max_{i=1,\ldots, n}\|\hat{U}_i-U_iR\| \leq c_4 \frac{(\log n)^{c}}{n^{1/2}},
\]
for some constant $c_4>0$, where $R $ is an orthogonal matrix, $c>1$ is a constant mentioned in the statement of the theorem, and $\|.\|$ denotes the vector Euclidean norm.
From Theorem 1 in \cite{rubin2022statistical} we have
$P(G_n=1) = 1-\frac{1}{n^{c'}}$ for some $c'>0$.

We will use the law of total expectation and covariance (also known as tower property) conditioning on the event $G_n$. First, for the unconditional covariance we have the following formula,
\begin{align*}
    Cov[U_i,U_j|A] &= E[Cov(U_i, U_j|A,G_n)|A] + Cov[E[U_i|A,G_n],E[U_j|A,G_n]|A].
\end{align*}
Now, we note from the triangle property of the Euclidean vector norm,
\[
\|U_i\| \leq \|\hat{U}_i\| + \|U_i - \hat{U}_i\|.
\]
Therefore if $G_n=1$, then we have $\|U_i\| = \|\hat{U}_i\|+ O(\frac{(\log n)^c}{n^{1/2}})$ for all $i$. This implies that \citep{mcfowland2021estimating}
\[
Cov(U_i, U_j|A,G_n=1) = O \left(\frac{(\log n)^{2c}}{n}\right),
\]
where by $O(.)$, we mean all elements of the matrix on the left-hand side are bounded by the quantity on the right-hand side asymptotically.
On the other hand, when $G_n=0$, while we do not have an upper bound on the closeness between $U$ and $\hat{U}$, we can use Popoviciu's theorem and population latent position conditions to bound the variances and covariances. First note that $U_i$s are bounded random vectors in the sense that every element $U_{ik}$ is bounded by the maximum norm of the vectors, $m=\max_{i=1,\ldots,n} \|U_i\| =\theta_n^{1/2}$. Therefore repeatedly applying Popoviciu's theorem, to the elements of the matrices $Var(U_i)$ and $Cov(U_i, U_j)$, we have
\[
Var(U_i) = O( \theta_n) \text{ and } cov(U_i, U_j) =  O( \theta_n),
\]
where again $O(.)$ implies element-wise asymptotic bound on the elements of the matrices on the left-hand side. Then, we can compute the expectation of the conditional covariance as,
\begin{align*}
E[Cov(U_i, U_j|A,G_n)|A] &= (1-\frac{1}{n^c})Cov(U_i, U_j|A,G_n=1) \\
& + \frac{1}{n^{c'}} Cov(U_i, U_j|A,G_n=0)\\
& = O\left( \frac{(\log n)^{2c}}{n} +  \theta_n \frac{1}{n^{c'}}\right).
\end{align*}

Now turning our attention to the conditional expectations, we have the conditional expectation  of $U_i$ given $G_n=1$ is
\[
E[U_i|A, G_n=1] = \hat{U}_i +  O\left(\frac{(\log n)^{c}}{n^{1/2}}\right).
\]
When $G_n=0$, we define $\tilde{U}_i = E[U_i|A, G_n=0]$. Note $\frac{1}{(\theta_n)^{1/2}}\tilde{U}_i$ takes a value in the Euclidean ball defined by $\|x\|_2 \leq 1$ and is a function of $A$. This implies, $\tilde{U}_i = O((\theta_n)^{1/2})$.
Now, we can write the conditional expectation of $U_i$ given $G_n$ as,
\[
E[U_i|A,G_n] = G_n\hat{U}_i +  G_n\frac{(\log n)^{c}}{n^{1/2}} + (1-G_n)\tilde{U}_i.
\]
Then similar to the computation in \cite{mcfowland2021estimating} we have
%\footnotesize
\begin{align*}
Cov & [E[U_i|A,G_n],E[U_j|A,G_n]|A] \\
& = Cov[G_n\hat{U}_i +  G_n\frac{(\log n)^{c}}{n^{1/2}} + (1-G_n)\tilde{U}_i, \, \, G_n\hat{U}_j +  G_n\frac{(\log n)^{c}}{n^{1/2}} + (1-G_n)\tilde{U}_j|A]\\
& = Var(G_n|A) (\hat{U}_i \hat{U}_j^T + (\hat{U}_i + \hat{U}_j) \frac{(\log n)^{c}}{n^{1/2}} + \frac{(\log n)^{2c}}{n}) \\
& \quad + Var(1-G_n|A)\tilde{U}_i\tilde{U}_j^T  + Cov(G_n(1-G_n)|A)(\hat{U}_i \tilde{U}_j^T  + \tilde{U}_i \hat{U}_j^T + (\tilde{U}_i + \tilde{U}_j) \frac{(\log n)^{c}}{n^{1/2}})\\
& = Var(G_n|A) (O(\theta_n + \frac{(\log n)^{c}(\theta_n)^{1/2}}{n^{1/2}} + \frac{(\log n)^{2c}}{n})) \\ 
& \quad +Var(1-G_n|A)O(\theta_n) + Cov(G_n, 1-G_n|A)   O(\theta_n + \frac{(\log n)^{c}(\theta_n)^{1/2}}{n^{1/2}} ). 
\end{align*}
%\normaLSIze
The last line follows from the following arguments. We note $\hat{U}_i, \hat{U}_j, \tilde{U}_i, \tilde{U}_j $ are all functions of $A$ and therefore are non-random when conditioned on $A$, and given our assumption that the sparsity of the network $\theta_n$ is known, they are $O((\theta_n)^{1/2})$. 
Further, note that $G_n$ is an indicator random variable. Therefore $Var(G_n|A) = var(1-G_n|A) = \frac{1}{n^{c'}}(1-\frac{1}{n^{c'}}) = O(\frac{1}{n^{c'}}).$ Moreover, since $G_n(1-G_n) = 0$, we have $Cov(G_n(1-G_n)|A) = - E[G_n|A] E[1-G_n|A] = -\frac{1}{n^{c'}}(1-\frac{1}{n^{c'}}) = O(\frac{1}{n^{c'}}) $. Then the above becomes
\begin{align*}
& Cov[E[U_i|A,G_n],E[U_j|A,G_n]|A] = \frac{1}{n^{c'}} O\left(\theta_n + \frac{(\log n)^{c}(\theta_n)^{1/2}}{n^{1/2}} + \frac{(\log n)^{2c}}{n})\right).
\end{align*}
Therefore combining the two results we have
\[
Cov[U_i, U_j|A] = O\left(\frac{(\log n)^{2c}}{n} +  \theta_n \frac{1}{n^{c'}} +  \frac{(\log n)^{c}(\theta_n)^{1/2}}{n^{c'} n^{1/2}}\right).
\]

Now note that since $U_i$ is $O(\theta_n)^{1/2}$, the corresponding coefficient $\beta$ should be $O(\theta_n)^{-1/2}$ in order for the total term $\beta^TU_i$ to be constant as a function of $n$. 
Therefore, using the above estimate of the covariance in \eqref{covar}, the asymptotic order of bias in estimating $\rho_n$ is given by multiplying the above expression by $\theta_n^{-1}$, as,
\[
E[\hat{\rho}_n-\rho|A] = O\left(\frac{(\log n)^{2c}}{n\theta_n} +   \frac{1}{n^{c'}} +  \frac{(\log n)^{c}}{n^{c'} (n\theta_n)^{1/2}}\right).
\]

\end{proof}

\begin{proof}[\textbf{Proof of Theorem} \ref{theo2}]

We can write down the log-likelihood function associated with the linear regression model as
\begin{align*}
    &l(\beta, W, U_{-i}, \hat{U}_i, Y) = -\frac{n}{2} \log 2\pi - n \log \sigma - \frac{1}{2\sigma^2}\{\sum_{j \neq i} (Y_{j}- \eta^T W_j - \beta^TU_j)^2 + (Y_{i}- \eta^T W_i - \beta^T\hat{U}_i)^2 \}.
\end{align*}
The score equation is given by
\begin{align*}
 & (\sum_j W_jW_j^T + \sum_{j\neq i} U_j U_j^T + \hat{U}_i\hat{U}_i^T)\begin{pmatrix}
    \eta \\  \beta
\end{pmatrix}   = (\sum_j Y_jW_j + \sum_{j\neq i} Y_j U_j + Y_i\hat{U}_i),
\end{align*}
which leads to the estimator $\begin{pmatrix}
    \hat{\eta} \\  \hat{\beta}
\end{pmatrix}^{(old)}$ given in the statement of the theorem.
As $n\to \infty$ this estimator converges in probability to the following limit
\[
(E[W_1W_1^T] + \Delta_F/\theta + \Sigma(U_i))^{-1}(E[W_1W_1^T] +\Delta_F/\theta) \begin{pmatrix}
    \eta_0 \\  \beta_0
\end{pmatrix},
\]
since $Cov(\hat{U}_i - U_i) \to \Sigma(U_i)$ and $E[U_1U_1^T] \to \Delta_F/\theta$ as $n \to \infty$, and $\Sigma(\cdot)$ is the function defined in the main paper. 

To derive the bias-corrected estimator, we note that
\begin{align*}
    E&[l(\beta,W,U_{-i}, \hat{U}_i, Y)] = l(\beta, W, U,Y) + \frac{1}{2\sigma^2}[ \beta^T \Sigma(U_i) \beta]
\end{align*}

Then, the corrected log-likelihood function is
\begin{align*}
l^{*}&(\beta, W, U_{-i}, \hat{U}_i,Y) = l(\beta, W, U_{-i}, \hat{U}_i,Y)-\frac{1}{2\sigma^2}\beta^T \Sigma(U_i) \beta.
\end{align*}
The corrected score equation is
\begin{align*}
& (\sum_j W_jW_j^T + \sum_{j\neq i} U_j U_j^T + \hat{U}_i\hat{U}_i^T-\Sigma(U_i))\begin{pmatrix}
    \eta \\  \beta
\end{pmatrix} = (\sum_j Y_jW_j + \sum_{j\neq i} Y_j U_j + Y_i\hat{U}_i),
\end{align*}
which leads to the estimator $\begin{pmatrix}
    \hat{\eta} \\  \hat{\beta}
\end{pmatrix}^{(new)}$ given in the statement of the theorem.
Then as $n\to \infty$ this estimator converges in probability to
\begin{align*}
& (E[W_1W_1^T] + \Delta_F/\theta)^{-1}(E[W_1W_1^T] +\Delta_F/\theta) \begin{pmatrix}
    \eta_0 \\  \beta_0
\end{pmatrix} = \begin{pmatrix}
    \eta_0 \\  \beta_0
\end{pmatrix}
\end{align*}
\end{proof}

\begin{proof}[\textbf{Proof of Theorem} \ref{asympbiasprobit}]
We have already seen in the proof of Theorem \ref{asympbias} that
\[
Cov\left( \frac{\sum_j A_{ij}Y_{j,t-1}}{\sum_j A_{ij}}, \beta^T U_i - \beta_1^T \hat{U}_i | \hat{U}_i, \hat{U}_j, A_{ij}\right) = O\left(\frac{(\log n)^{2c}}{n\theta_n} +   \frac{1}{n^{c'}} +  \frac{(\log n)^{c}}{n^{c'} (n\theta_n)^{1/2}}\right).
\]
This covariance, which is the same as $Cov(LY_{t-1}, q|A) $ converges to 0 as $n\to \infty$, since $n \theta_n \to \infty$ as $n \to \infty$. With the assumptions of normality for $LY_{t-1}$ and $q$, this implies that $\delta$ converges to 0 as $n \to \infty$. Moreover, we note that
\[
Var(\beta^T U_i - \beta_1^T \hat{U}_i|A) = \beta^T Var(U_i|A) \beta = O\left(\frac{(\log n)^{2c}}{n\theta_n} +   \frac{1}{n^{c'}} +  \frac{(\log n)^{c}}{n^{c'} (n\theta_n)^{1/2}}\right), 
\]
which also converges to $0$ as $n \to \infty$. Therefore $\tau^2 \to 0$ and $n \to \infty$.
This implies
\begin{align*}
P(Y_{i,t}=1|W_i, \hat{U}_i) & = P(V_{it} + q < - \eta^TW_i - \beta_1^T\hat{U}) \\
& = P\left(\frac{V_{it} + q - \delta (LY_{t-1})_i }{\sqrt{\tau^2 +1}} < \frac{- \eta^TW_i - \beta_1^T\hat{U}_i-\delta (LY_{t-1})_i}{\sqrt{\tau^2 +1}}\right) \\
& = \Phi (\frac{ \eta^TW_i + \beta_1^T\hat{U}_i+\delta (LY_{t-1})_i}{\sqrt{\tau^2 +1}}) \to \Phi ( \eta^TW_i + \beta_1^T\hat{U}_i) .
\end{align*}
The above calculation implies that $\rho$ is consistently estimated.  Moreover, the average partial effect also converges to the true average partial effect.
\[
APE =  \frac{\rho+\delta}{\sqrt{\tau^2 +1}}  \phi\left(\frac{ \eta^TW_i + \beta_1^T\hat{U}+\delta (LY_{t-1})_i}{\sqrt{\tau^2 +1}}\right) \to \rho  \phi(\eta^TW_i + \beta_1^T\hat{U}_i).
\]
\end{proof}

\subsection{Additional Figures and Tables}

\begin{table}[!htbp] 
\centering 
  \caption{Summary Statistics} 
  \label{summarytableall3units} 
\begin{tabular}{@{\extracolsep{5pt}}p{4cm}p{1.0cm}p{1.0cm}p{1.0cm}p{1.0cm}p{1.0cm}} 
\\[-1.8ex]\hline 
\hline \\[-1.8ex] 
Variable & \multicolumn{1}{c}{N} & \multicolumn{1}{c}{Mean} & \multicolumn{1}{c}{St. Dev.} & \multicolumn{1}{c}{Min} & \multicolumn{1}{c}{Max} \\ 
\hline \\[-1.8ex] 
&\multicolumn{5}{c}{\textbf{A. Male unit 1}}\\
Graduation Status & 337 & 0.884 & 0.320 & 0 & 1 \\ 
Age & 337 & 27.665 & 8.945 & 18 & 61 \\ 
White & 337 & 0.475 & 0.500 & 0 & 1 \\ 
LSI & 337 & 25.727 & 5.215 & 9 & 44 \\ 
Peer Graduation & 337 & 0.416 & 0.203 & 0.000 & 1.000 \\ \\
&\multicolumn{5}{c}{\textbf{B. Male unit 2}}\\
Graduation Status & 339 & 0.894 & 0.309 & 0 & 1 \\ 
Age & 339 & 31.746 & 9.360 & 18 & 60 \\ 
White & 339 & 0.776 & 0.418 & 0 & 1 \\ 
LSI & 339 & 25.661 & 6.046 & 14 & 45 \\ 
Peer Graduation & 339 & 0.425 & 0.178 & 0.000 & 1.000 \\ \\ 
&\multicolumn{5}{c}{\textbf{C. Female unit}}\\
Graduation Status & 472 & 0.797 & 0.403 & 0 & 1 \\ 
Age & 472 & 30.358 & 8.203 & 18 & 60 \\ 
White & 472 & 0.799 & 0.401 & 0 & 1 \\ 
LSI & 472 & 25.862 & 8.378 & 0 & 57 \\ 
Peer Graduation & 472 & 0.381 & 0.159 & 0.000 & 0.868 \\ 
\hline\hline
\multicolumn{6}{c}{ \begin{minipage}{12 cm}{\footnotesize{Notes: The summary statistics on the outcome variable (graduation status), covariates and the primary explanatory variable (weighted peer graduation status) for all three units are provided in this table.}}
\end{minipage}} \\
\end{tabular} 
\end{table} 

\begin{table}[!htbp]
\centering
\captionsetup{width=10.5cm}
\caption{Peer Effects and Race (Homophily and Bias Adj. Corrections Network)}
\label{rhotablecorrectionshte}
\begin{tabular}{p{6.5 cm}P{1.75cm}P{1.75cm}P{1.75cm}}\hline\hline
&\multicolumn{3}{c}{Dependent Variable: $(S_{i})$}\\\\
 & (1) & (2) & (3) \\
Variable & Male Unit 1 & Male Unit 2 & Male Unit 3 \\
 \hline
Peer Grad. (White)             & 0.189   & 0.635   & 0.724   \\
                                    & (0.075) & (0.134) & (0.141) \\
Peer Grad. (Non-White)         & 0.100   & 0.302   & -0.015  \\
                                    & (0.077) & (0.102) & (0.098) \\
White                               & -0.098  & 0.134   & 0.177   \\
                                    & (0.048) & (0.070) & (0.065) \\
Peer Grad. (White) x White     & 0.086   & -0.188  & -0.383  \\
                                    & (0.102) & (0.158) & (0.153) \\
Peer Grad. (Non-White) x White & 0.066   & -0.157  & 0.173   \\
                                    & (0.106) & (0.117) & (0.110) \\
Age                                 & 0.001   & -0.001  & 0.001   \\
                                    & (0.001) & (0.001) & (0.001) \\
LSI                                 & -0.025  & -0.024  & -0.016  \\
                                    & (0.002) & (0.002) & (0.001) \\
Intercept                           & 1.393   & 1.122   & 0.893   \\
                                    & (0.077) & (0.098) & (0.079) \\\\
N&774&391&1046\\\\
\hline\hline
\multicolumn{4}{c}{ \begin{minipage}{13.5 cm}{\footnotesize{Notes: Standard errors are provided in parenthesis. Two new variables are constructed for this specification. We compute the peer graduation status of white and non-white residents separately for this analysis. The latent homophily vectors are estimated from the corrections network in these specifications.}}
\end{minipage}} \\
\end{tabular}
\end{table}

\begin{table}[!htbp]
\centering
\captionsetup{width=8.5cm}
\caption{Directed Sender and Receiver Networks}
\label{rhotablerobustnessdirectedcorrectionsnetwork}
\begin{tabular}{p{3.5 cm}P{1.95cm}P{1.95cm}P{1.95cm}}\hline\hline
&\multicolumn{3}{c}{Dependent Variable: $S_{i}$}\\\\
 & (1) & (2) & (3) \\
Variable & Male Unit 1 & Male Unit 2 & Female Unit \\
 \hline
 \multicolumn{4}{c}{a. Directed Sender Corrections Network}\\
Peer Graduation & 0.289   & 0.563   & 0.421   \\
                & (0.050) & (0.075) & (0.055) \\
Age             & 0.002   & -0.001  & 0.001   \\
                & (0.001) & (0.001) & (0.001) \\
White           & -0.021  & 0.017   & 0.073   \\
                & (0.021) & (0.033) & (0.023) \\
LSI             & -0.024  & -0.024  & -0.016  \\
                & (0.002) & (0.002) & (0.001) \\\\
 \multicolumn{4}{c}{b. Directed Receiver Corrections Network}\\
 Peer Graduation & 0.356   & 0.606   & 0.544   \\
                & (0.047) & (0.065) & (0.051) \\
Age             & 0.002   & 0.000   & 0.001   \\
                & (0.001) & (0.001) & (0.001) \\
White           & -0.020  & 0.032   & 0.078   \\
                & (0.021) & (0.032) & (0.022) \\
LSI             & -0.024  & -0.023  & -0.016  \\
                & (0.002) & (0.002) & (0.001) \\
\hline\hline
\multicolumn{4}{c}{ \begin{minipage}{11.0 cm}{\footnotesize{Notes: Standard errors are provided in parenthesis. Peer Grad. is constructed using the precise entry and exit dates and the affirmations network between the residents as defined in equation \ref{eqrolemodel}. The latent homophily is also estimated from the affirmations network. For panel (a), corrections only consist of directed sender corrections from node $i$ to peers, and panel (b) consists of directed receiver corrections.}}
\end{minipage}} \\
\end{tabular}
\end{table}

\begin{table}[!htbp]
\centering
\caption{Counterfactual Scenario}
\label{countefactuaLSImulations}
\begin{tabular}{p{7cm}p{1.5cm}p{1.5cm}p{1.5cm}}
  \hline\hline
 & Male Unit 1 & Male Unit 2 & Female Unit \\ 
  \hline
Threshold for Graduation & 0.59 & 0.66 & 0.47 \\ 
Failures (True Data) & 39.00 & 36.00 & 96.00 \\ 
  Residents Who cleared the Threshold of Graduation due to a Buddy & 24.00 & 12.00 & 42.00 \\ 
  Residents whose predicted graduation is below threshold (Post-Simulation) & 2.00 & 21.00 & 25.00 \\ 
   \hline\hline
   \multicolumn{4}{c}{ \begin{minipage}{13 cm}{\footnotesize{Notes: We show the results of the counterfactual exercise in this table. For this, we use the true $S_{i}$ to identify the ``at-risk" residents. Then a successful buddy is assigned to these ``at-risk" residents. Using the estimates from Table \ref{rhotable} the buddy assignment helps some of these residents cross the graduation threshold. The $S_{i}$ is modified, and then we re-estimate our role model effect. This second re-estimation takes into account the indirect cascading effect of buddy assignment. At the end of this estimation, we calculate every resident's propensity to graduate and report the remaining ``at-risk" residents in the last row of this table.}}
\end{minipage}} \\
\end{tabular}
\end{table}

\begin{figure}[!htbp]
  \caption{Variation in peer graduation and graduation success}
    \label{corrownpeer} 

 \begin{minipage}[b]{0.5\linewidth}
       \begin{center}

    \includegraphics[width=\linewidth]{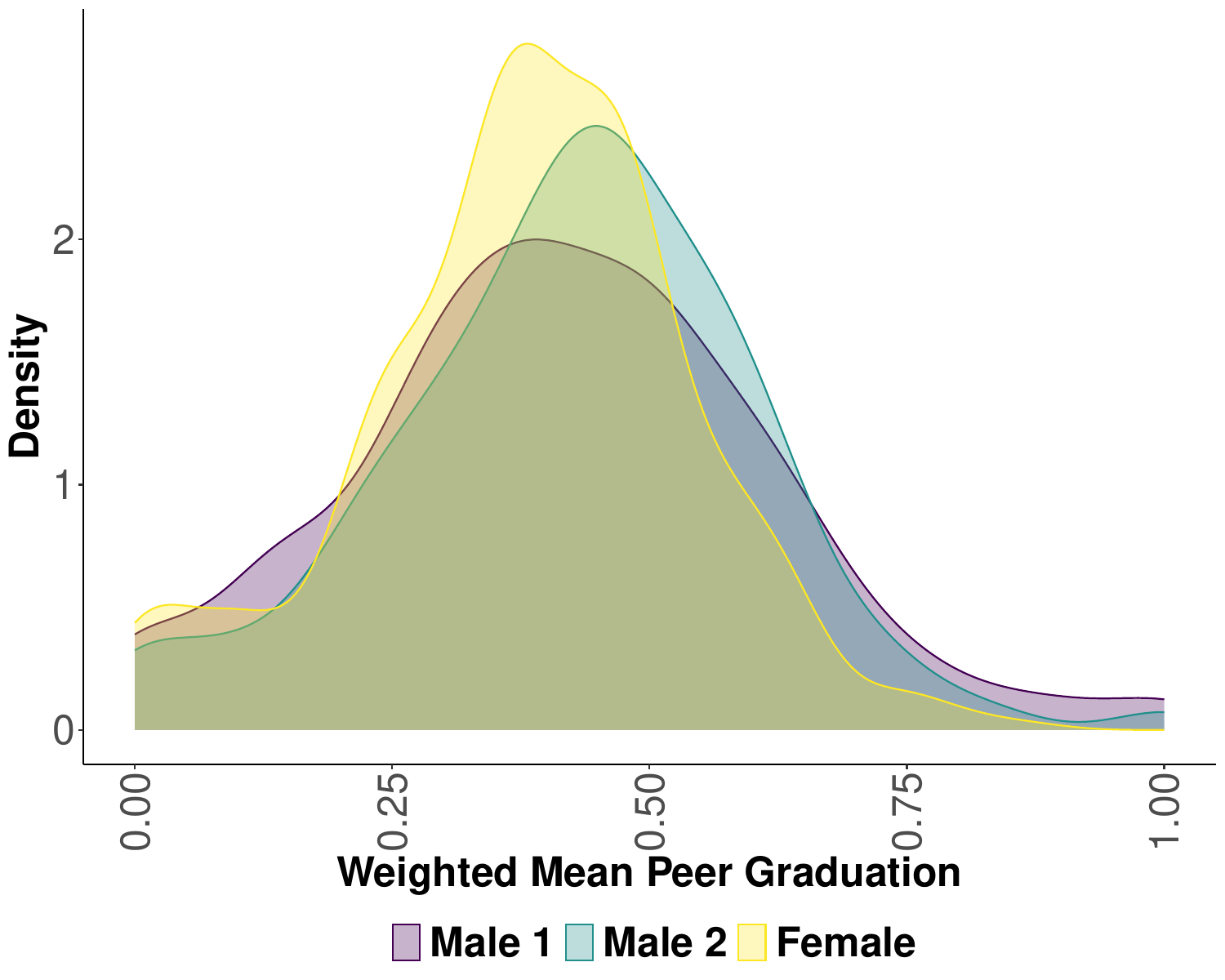}
    \caption*{a. Graduation Status of peers} 
        \end{center}

  \end{minipage}%%
  \begin{minipage}[b]{0.5\linewidth}
       \begin{center}

    \includegraphics[width=\linewidth]{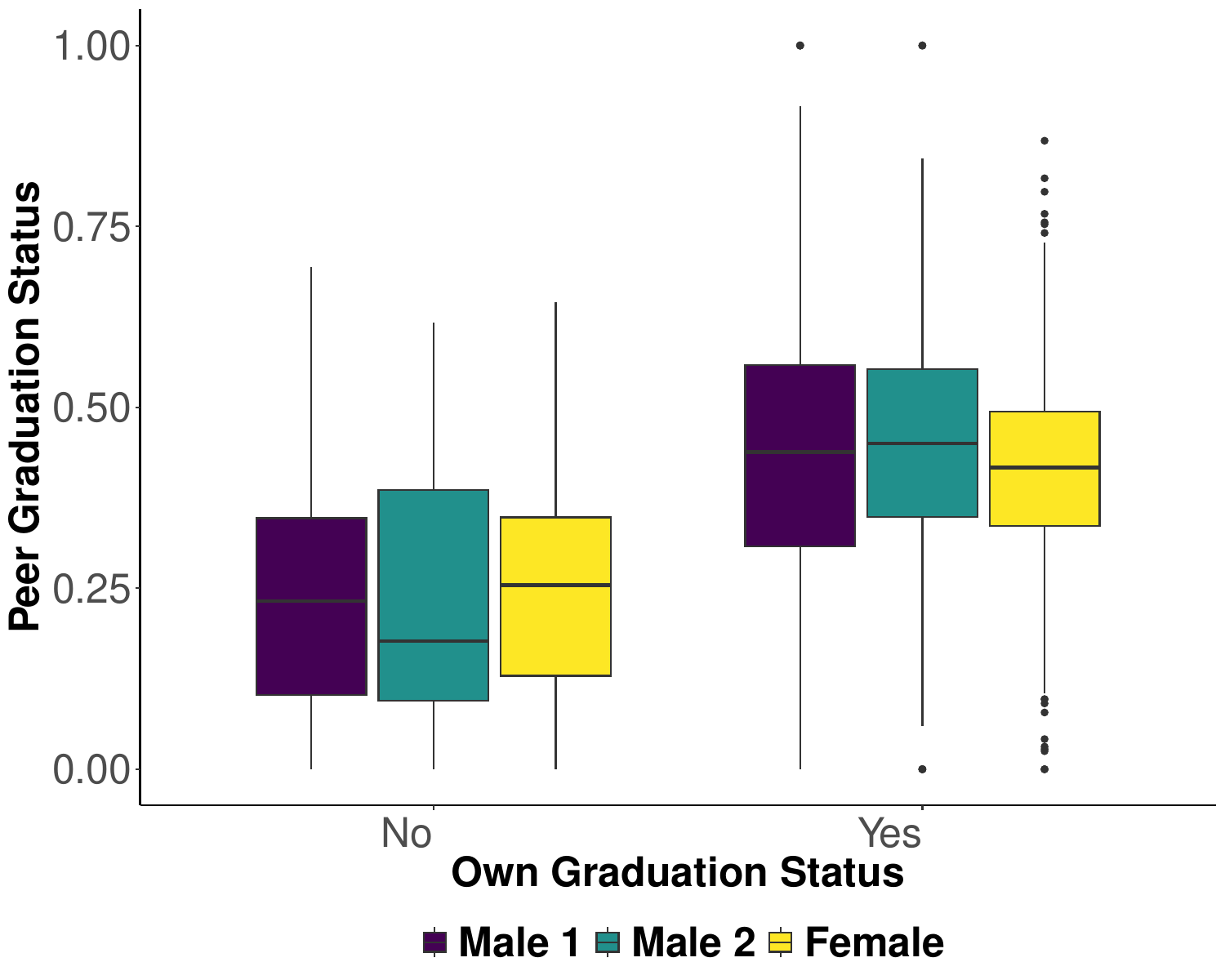}
    \caption*{b. Peer graduation and success} 
        \end{center}

  \end{minipage}
  
  \begin{minipage}{15.5cm}
\footnotesize{
    {Notes: We pool the data across three units for these figures. The final data set includes the residents for whom we observe the complete set of covariates and the affirmations network. Panel (a) shows the distribution of peer graduation status weighed by the nodes of the affirmation network as observed by residents at time $t$ when survey for graduation status was conducted. Panel (b) displays the correlation between the peer graduation as observed by each resident and their own graduation status.}}
    \end{minipage}
\end{figure}

\begin{figure}[!htbp]
  \caption{Out of Sample AUC for Selection of $d$ using Cross-validation}
    \label{outofsampleauc} 

 \begin{minipage}[b]{0.33\linewidth}
       \begin{center}

    \includegraphics[width=\linewidth]{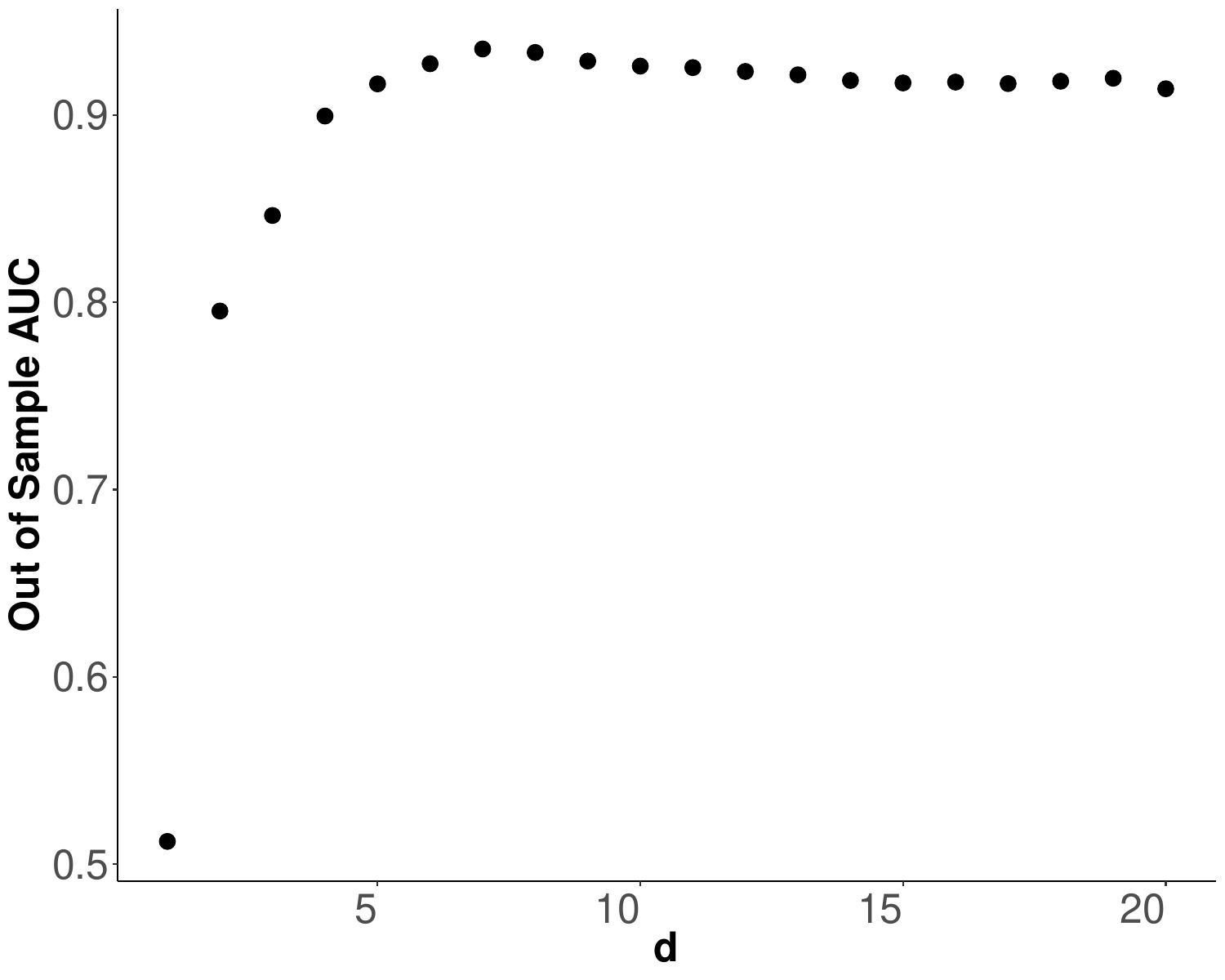}
    \caption*{a. Male Unit 1} 
        \end{center}

  \end{minipage}%%
  \begin{minipage}[b]{0.33\linewidth}
       \begin{center}

    \includegraphics[width=\linewidth]{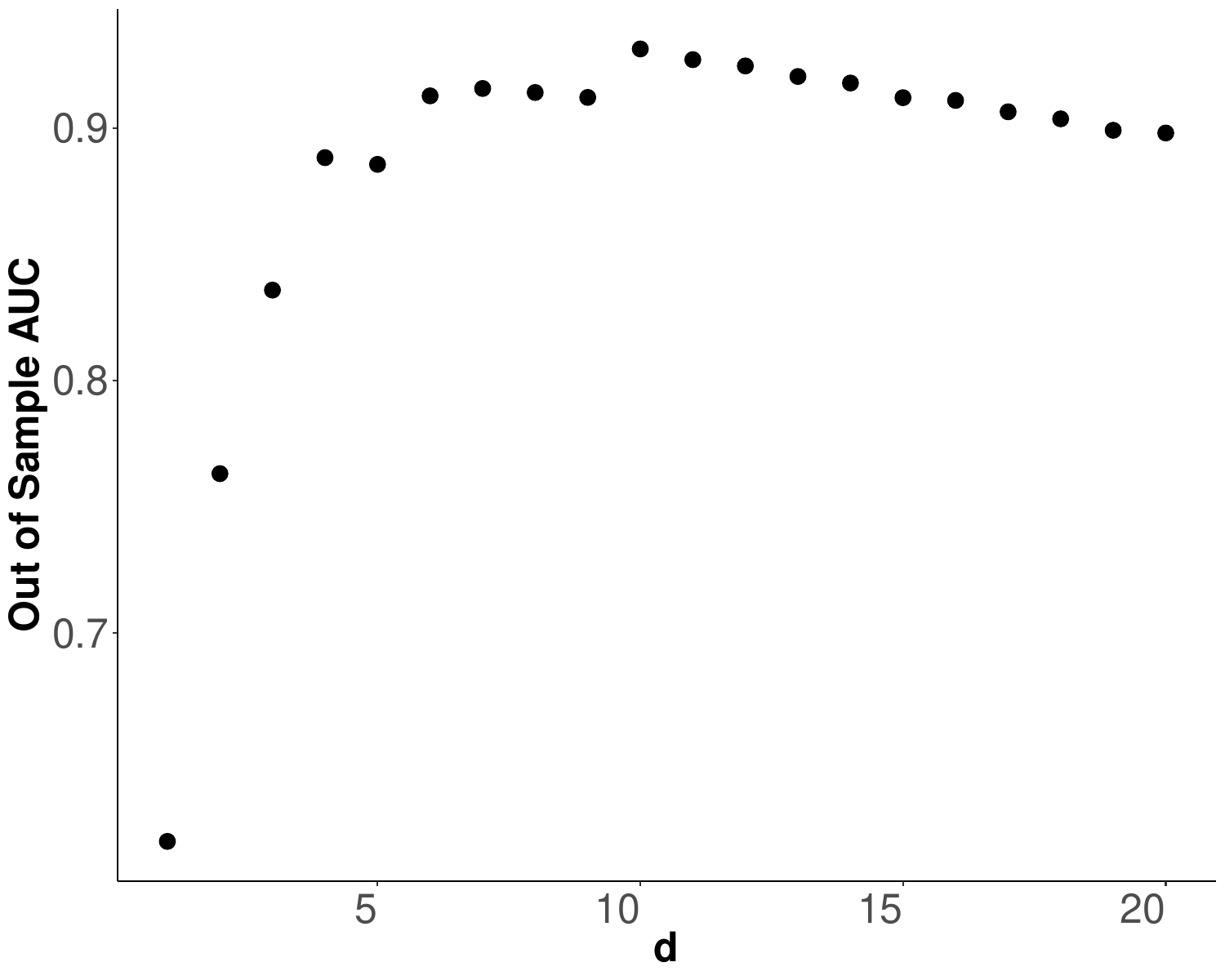}
    \caption*{b. Male Unit 2} 
        \end{center}

  \end{minipage}%%
    \begin{minipage}[b]{0.33\linewidth}
    
         \begin{center}

    \includegraphics[width=\linewidth]{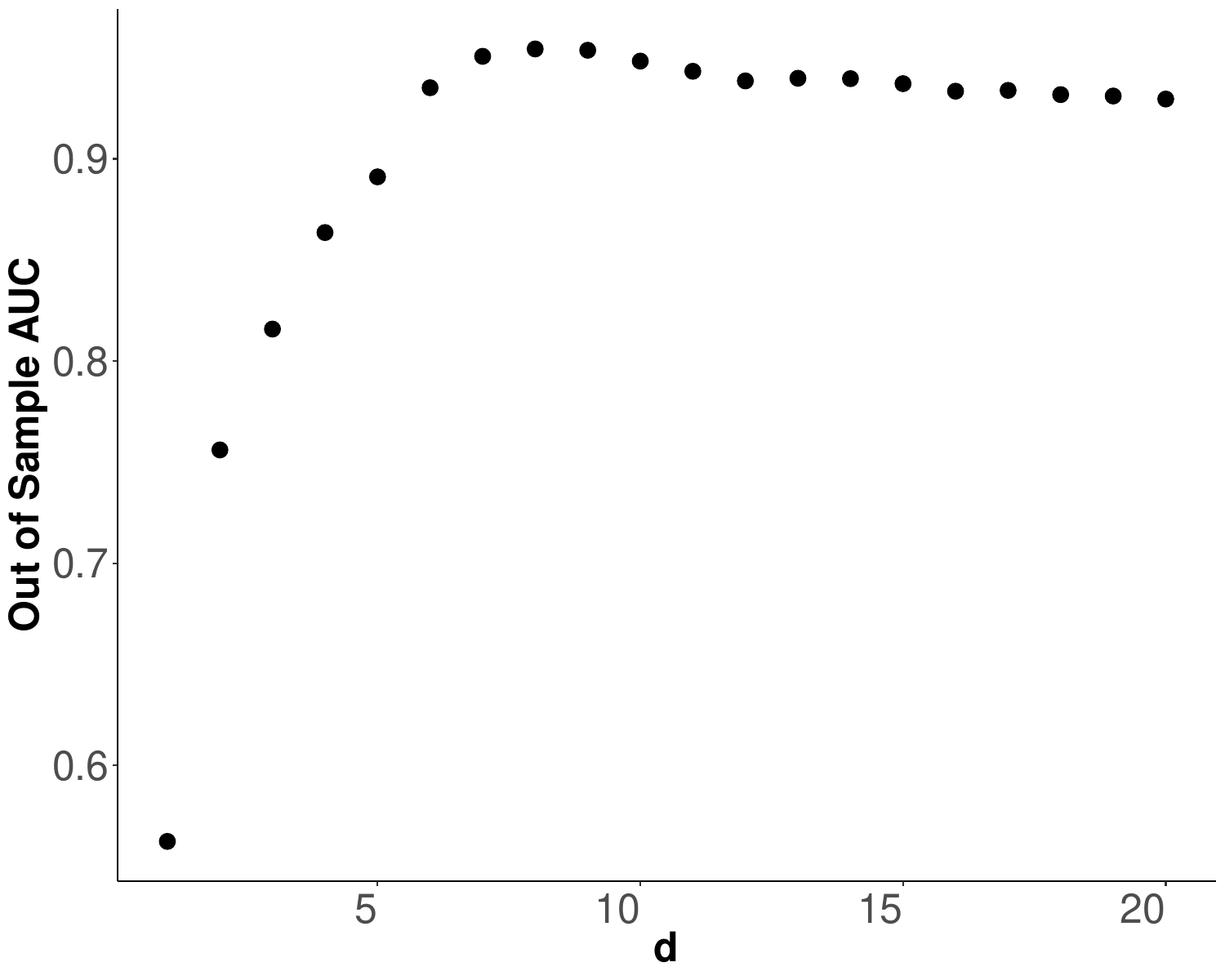}
    \caption*{c. Female Unit} 
          \end{center}

  \end{minipage}
  
  \begin{minipage}{15.5cm}
\footnotesize{
    {Notes: Figures (a,b, and c) illustrate the out-of-sample as we increase $d$ from 1 to 20 in male unit 1, male unit 2, and female unit, respectively. The $d$ for each unit is chosen via cross-validation using the out-of-sample AUCs.}}
    \end{minipage}
\end{figure}

\begin{figure}[!htbp] 
  \caption{Misclassification Error}
    \label{misclassification} 

 \begin{minipage}[b]{0.33\linewidth}
       \begin{center}

    \includegraphics[width=\linewidth]{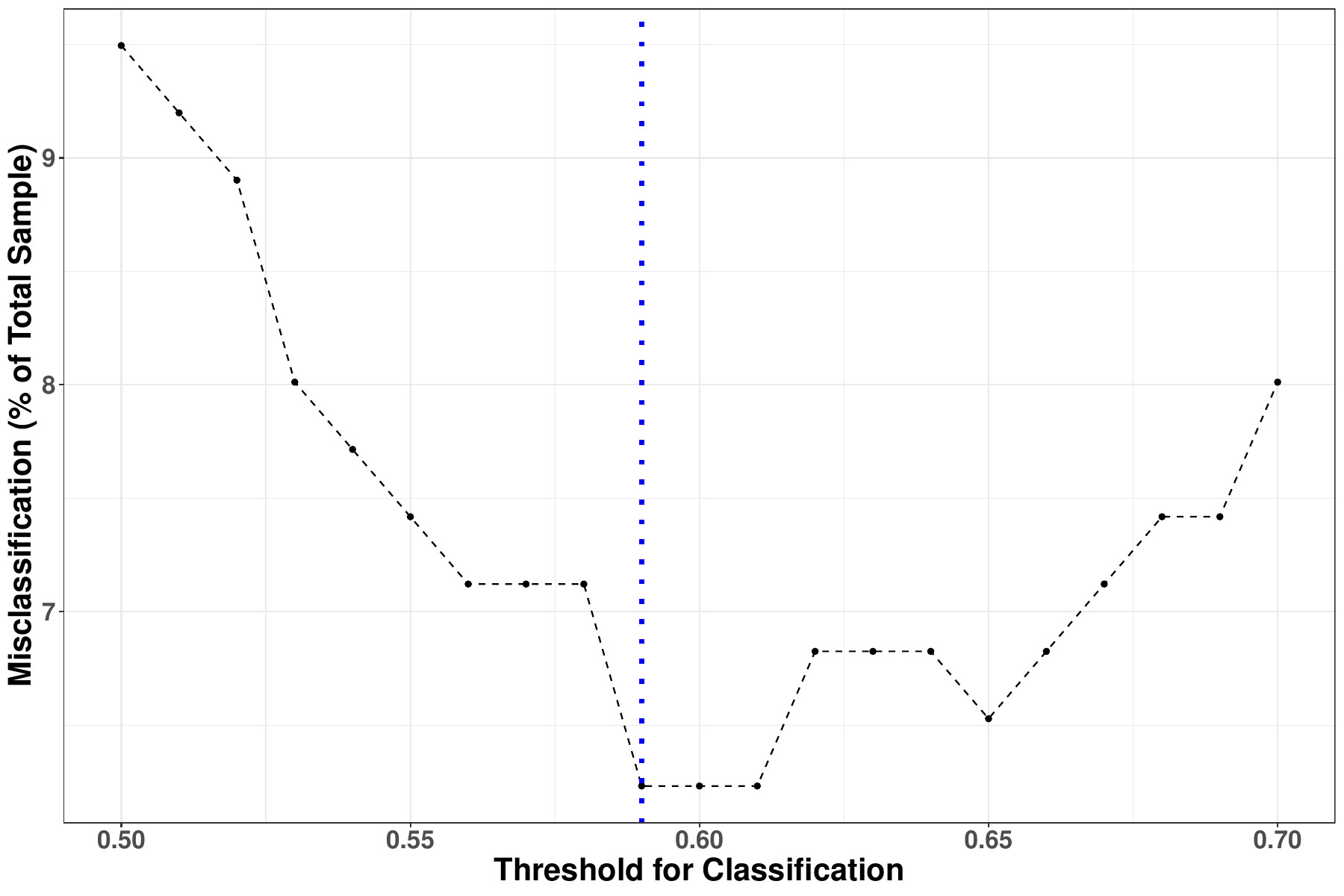}
    \caption*{a. Male Unit 1} 
        \end{center}

  \end{minipage}%%
  \begin{minipage}[b]{0.33\linewidth}
       \begin{center}

    \includegraphics[width=\linewidth]{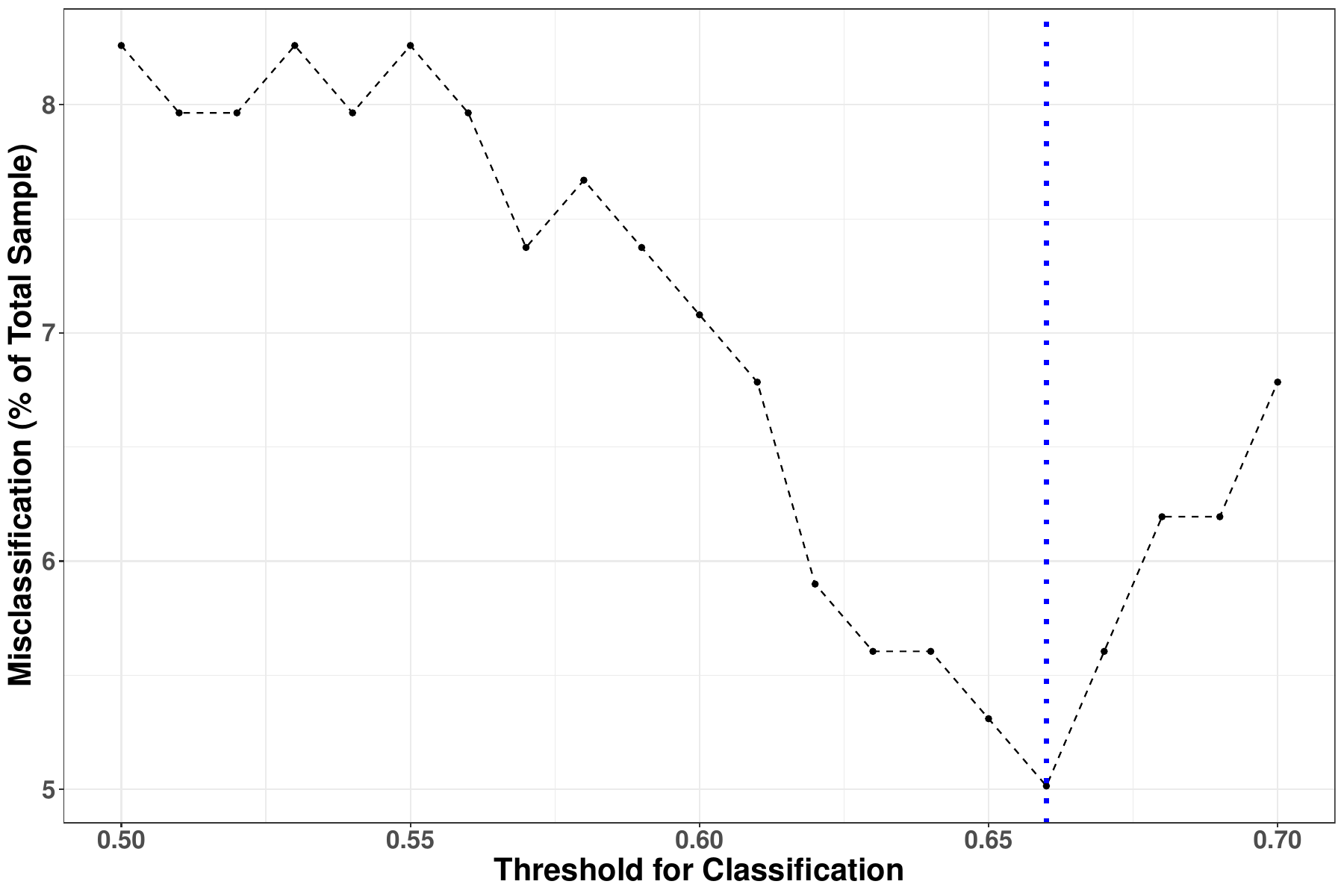}
    \caption*{b. Male Unit 2} 
        \end{center}

  \end{minipage}%%
    \begin{minipage}[b]{0.33\linewidth}
    
         \begin{center}

    \includegraphics[width=\linewidth]{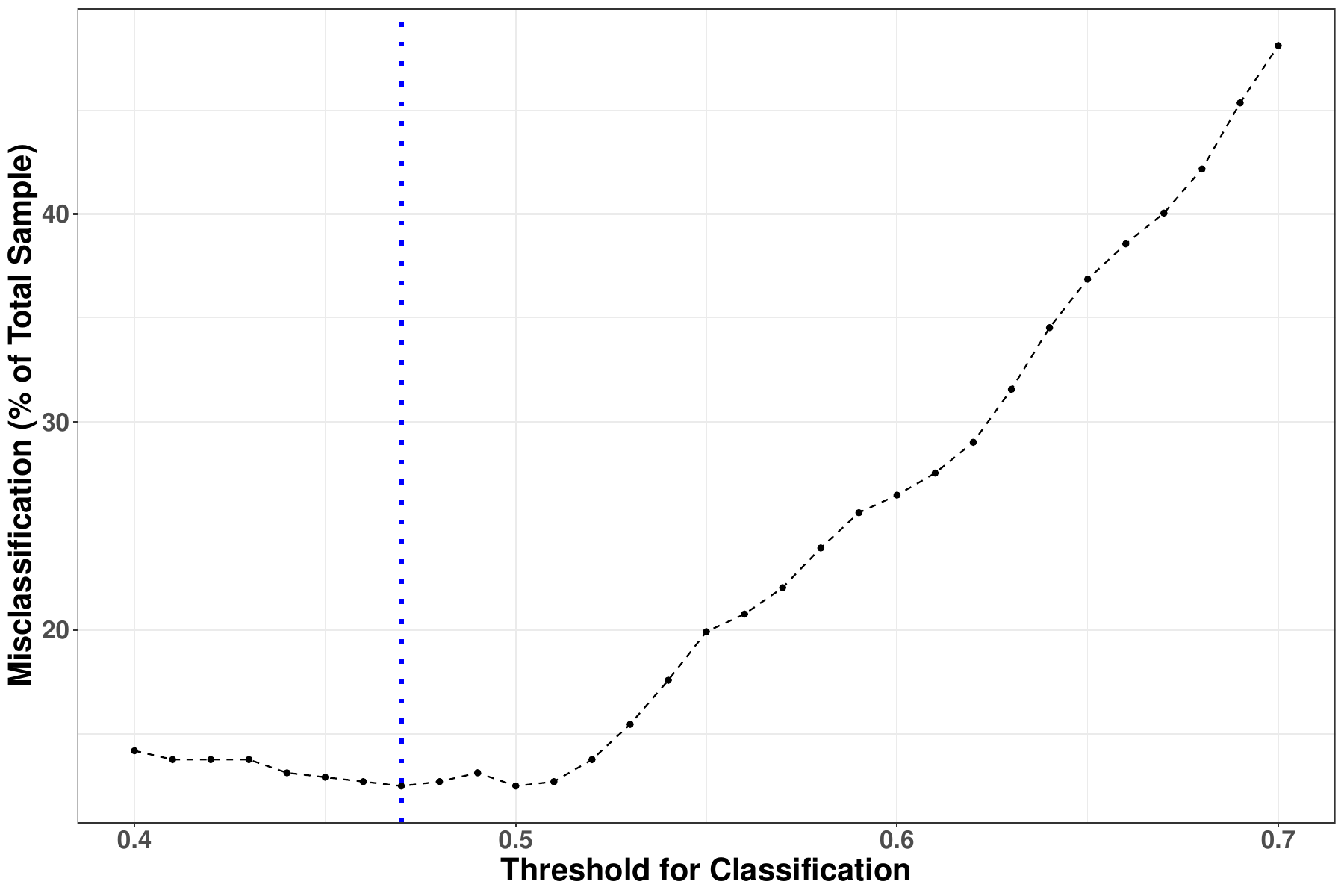}
    \caption*{c. Female Unit} 
          \end{center}

  \end{minipage}
  \begin{minipage}{15.5cm}
\footnotesize{
    {Notes: Figures (a,b, and c) illustrate the misclassification error as the threshold changes. We choose the threshold corresponding to the case with the lowest misclassification error. }}
    \end{minipage}
\end{figure}

\begin{figure}[h]
\caption{Updated DAG with Covariates for Network Formation Process}
\label{updateddag}
    \begin{center}
    \tikz{
    \node[hidden1] (a) {$Y_{j}^{i}$};
    \node[hidden1] (z) [right = of a] {$S_{i}$};
    \node[hidden1] (b) [above right = of z] {$T_{i}$};
    \node[hidden1] (x) [below left = of z] {$T_{j}$};
    \node[hidden1] (c) [below right = of x] {$A_{ij}$};
        \node[hidden] (d) [left = of x] {$U_{j}$};
                \node[hidden] (e) [right = of b] {$U_{i}$};
                \node[hidden1] (f) [right = of e] {$Z_{i}$};
                \node[hidden1] (g) [left = of d] {$Z_{j}$};
    \path (e) edge (c);
    \path (d) edge (a);
        \path (d) edge (x);
    \path (a) edge (z);
    \path (b) edge (z);
    \path (c) edge (z);
    \path (d) edge (c);
       \path (x) edge (a);
              \path (e) edge (b);
       \path (e) edge (z);
        \path (f) edge (z);
                \path (g) edge (a);
                \path (g) edge (c);
                \path (f) edge (c);

}
\end{center}
\end{figure}

\begin{figure}[!htbp] 
  \caption{Latent Space Model with Covariates}
    \label{latentspacewithcovariates} 

 \begin{minipage}[b]{0.8\linewidth}
       \begin{center}

    \includegraphics[width=\linewidth]{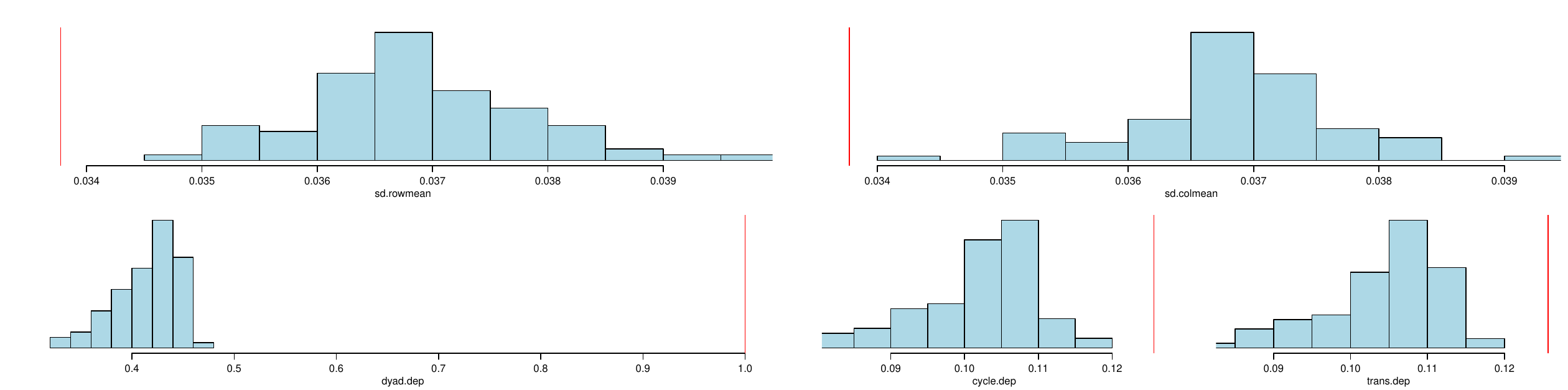}
    \caption*{a. Male Unit 1} 
        \end{center}

  \end{minipage}
  \begin{minipage}[b]{0.8\linewidth}
       \begin{center}

    \includegraphics[width=\linewidth]{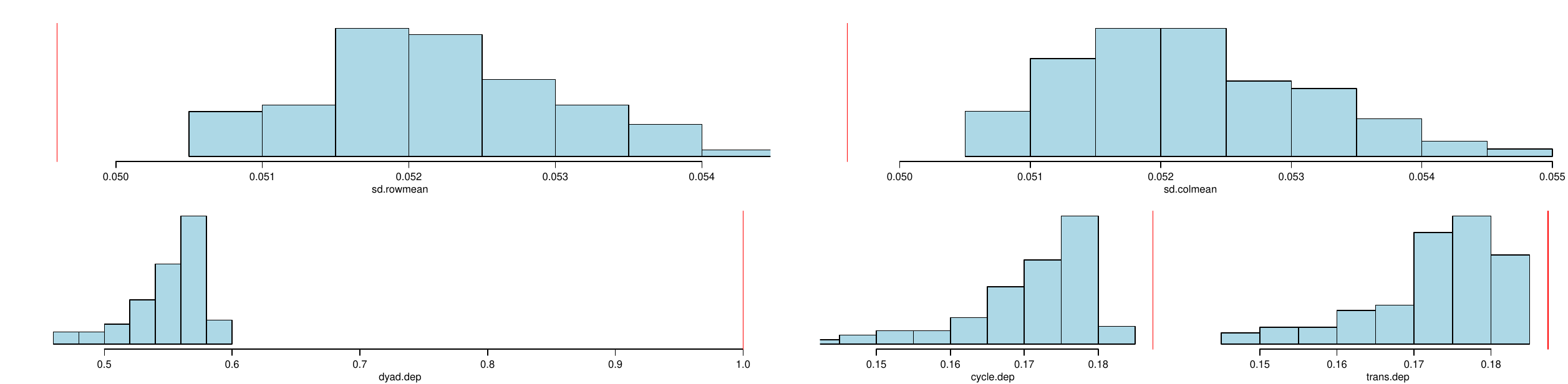}
    \caption*{b. Male Unit 2} 
        \end{center}

  \end{minipage}
    \begin{minipage}[b]{0.8\linewidth}
         \begin{center}
    \includegraphics[width=\linewidth]{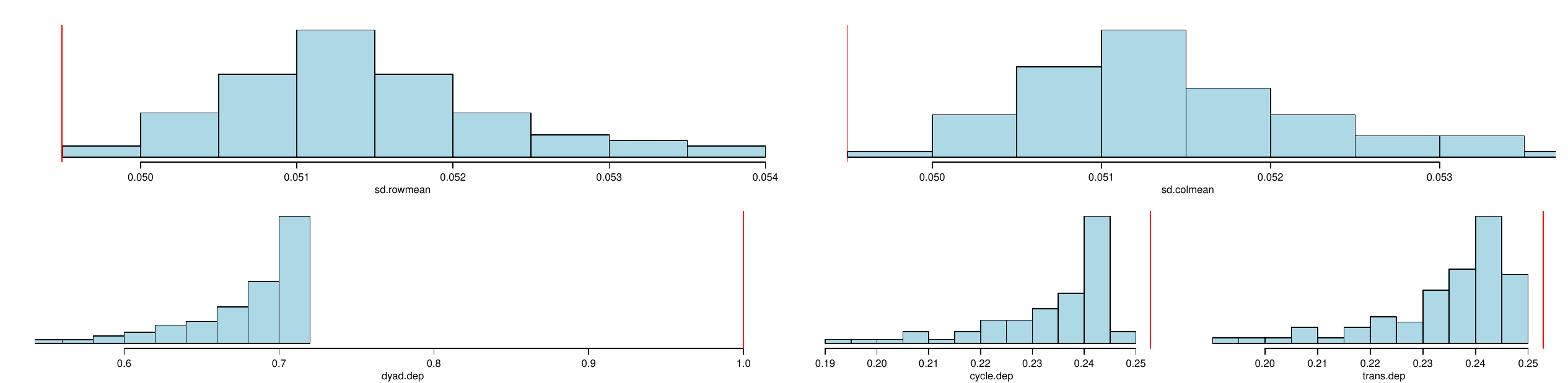}
    \caption*{c. Female Unit} 
          \end{center}
  \end{minipage}
  \begin{minipage}{15.5cm}
\footnotesize{
    {Notes: Figures (a,b, and c) illustrate a few relevant network properties from networks simulated from the fitted posterior distribution to the data. }}
    \end{minipage}
\end{figure}

\end{document}